\documentclass[a4paper,11pt]{article}

\title{Ascending Auctions for Combinatorial Markets with~Frictions:\\ A Unified Framework via Discrete Convex Analysis}
\author{Taihei Oki\thanks{Institute for Chemical Reaction Design and Discovery (ICReDD), Hokkaido University, Sapporo, Hokkaido, Japan. D3 Center, The University of Osaka, Osaka, Japan. Center for Advanced Intelligence Project, RIKEN, Tokyo, Japan. \texttt{oki@icredd.hokudai.ac.jp}}\and  Ryosuke Sato\thanks{School of Science and Engineering, Chuo University, Tokyo, Japan. \texttt{rsato636@g.chuo-u.ac.jp}}}
\date{\today}
\usepackage[top=1truein,bottom=1truein,left=0.7truein,right=0.7truein]{geometry}
\usepackage{amsthm}
\usepackage{amsmath}
\usepackage{amssymb}
\usepackage{hyperref}
\usepackage{amsfonts}
\usepackage{bm}
\usepackage{algorithm}
\usepackage{tikz}

\usepackage{mathtools}
\usepackage{standalone}
\usepackage{multirow}
\usetikzlibrary{positioning,arrows.meta, decorations.pathreplacing}
\usetikzlibrary{calc}
\usepackage{subcaption}
\usepackage{algorithmic}
\usepackage{comment}
\usepackage{booktabs}
\usepackage[numbers,sort&compress]{natbib}

\theoremstyle{plain}
\newtheorem{theorem}{Theorem}[section]      
\newtheorem{proposition}[theorem]{Proposition}
\newtheorem{lemma}[theorem]{Lemma}

\newtheorem{fact}[theorem]{Fact}

\theoremstyle{definition}
\newtheorem{definition}[theorem]{Definition}
\newtheorem{remark}[theorem]{Remark}

\newtheorem{assumption}[theorem]{Assumption}
\newtheorem{problem}[theorem]{Problem}
\newtheorem{observation}[theorem]{Observation}

\newtheorem{example}[theorem]{Example}

\providecommand{\mathbb}[1]{\mathbf{#1}}
\newcommand{\R}{\mathbb{R}}

\newcommand{\Z}{\mathbb{Z}}

\newcommand{\rhot}{\tilde{\rho}}
\newcommand{\rhoc}{\check{\rho}}
\newcommand{\rhoh}{\hat{\rho}}
\newcommand{\Dc}{\check{D}}
\newcommand{\Dh}{\hat{D}}
\newcommand{\cL}{\mathcal{L}}
\newcommand{\LSC}{\textbf{(LSC)}\,}

\let\angle\relax \newcommand{\angle}[1]{\langle #1 \rangle}

\newcommand{\set}[1]{\left\{ #1 \right\}}
\newcommand{\Set}[2]{\left\{ #1 \;\middle|\; #2 \right\}}

\providecommand{\texorpdfstring}[2]{#1}
\providecommand{\url}[1]{\texttt{#1}}

\DeclareMathOperator*{\argmax}{arg\,max}
\DeclareMathOperator*{\argmin}{arg\,min}

\DeclareMathOperator*{\minimize}{minimize}
\DeclareMathOperator*{\maximize}{maximize}

\newcommand{\zeros}{\mathbf{0}}
\newcommand{\ones}{\mathbf{1}}
\newcommand{\Mn}{M\texorpdfstring{${}^\natural$}{-natural}}
\newcommand{\Ln}{L\texorpdfstring{${}^\natural$}{-natural}}

\begin{document}
\maketitle

\begin{abstract}
We develop a unified ascending-auction framework for computing Walrasian equilibria in combinatorial markets with strong substitutes valuations and piecewise-linear payment functions. Our auction extends the celebrated ascending auctions of Gul and Stacchetti (2000) and Ausubel (2006) to accommodate payment \emph{frictions} (e.g., transaction taxes or commission fees). This is achieved by incorporating directional price updates that reflect heterogeneous payment structures. Our framework also generalizes the unit-demand imperfectly transferable utility models of Alkan (1989, 1992) to a fully combinatorial setting, thereby unifying these paradigms. Furthermore, this is the first study to compute the \emph{minimum} -- also known as the buyer-optimal -- equilibrium in combinatorial markets with such frictions.

Our analysis builds upon \emph{discrete convex analysis}. Our main technical contribution is a characterization of valid price-update directions, together with a strongly polynomial-time algorithm for computing them. Notably, the algorithm uses only demand- and exchange-oracle queries and never requires handling information of exponential size. To compute such a direction, we formulate a \emph{lexicographic} extension of the polymatroid sum problem and characterize its dual solution via a reduction to a convex flow problem. Exploiting the L$^\natural$-convexity of the dual objective, we show that the desired direction can be constructed from the minimal dual solution. This convexity also yields transparent economic and potential-based interpretations of the auction dynamics, strengthening the connection between ascending auctions and discrete optimization.
\end{abstract}

\section{Introduction}
\label{intro}
Efficient resource allocation is a central challenge in markets with {\it indivisible goods}, where heterogeneous items are allocated via monetary transfers. A central solution concept is the {\it Walrasian equilibrium} -- an allocation and a price vector such that, at those prices, each buyer obtains their most-preferred bundle and the market clears. In such markets, buyers are typically modeled using combinatorial valuation functions and {\it quasi-linear transferable utility}, where utility equals their valuation minus the total price of the allocated goods. A seminal result by \citet{KC1982} shows that an equilibrium exists for {\it gross substitutes} valuations. \citet{FY2003} prove that gross substitutability is equivalent to \Mn-concavity in \emph{discrete convex analysis}, a framework developed by \citet{M1998,M2003}, establishing a deep connection between economics and discrete optimization.

Ascending auctions provide constructive procedures for computing such equilibria. In unit-demand assignment markets, the ascending auction of \citet{DGS1986} computes the minimum equilibrium price vector. In combinatorial settings with gross substitutes valuations, the celebrated auctions of \citet{GS2000} and \citet{A2006} extend this classical result, ensuring convergence to the minimum equilibrium price vector. Their dynamics admit a Lyapunov function formulation and are closely related to discrete convex analysis. In particular, \citet{MSY2016} show that this potential function satisfies ${\mathrm L}^{\natural}$-convexity -- the conjugate of \Mn-concavity -- revealing that the auction dynamics can be interpreted as steepest descent algorithms for minimizing these ${\mathrm L}^{\natural}$-convex functions.

Many real-world markets depart from this classical framework due to {\it frictions} that distort the mapping from prices to payments, allowing identical prices to yield heterogeneous payments across buyers. Such frictions are ubiquitous; examples include:
\begin{itemize}
\item \textbf{Tiered commodity taxes}: Marginal tax rates may vary across price  tiers, and applicable rates or exemptions often depend on the buyer's legal status. 
\item \textbf{Import tariffs}: Import duties are often determined by transaction-price brackets, with rates typically depending on the origin and destination regions.
\item \textbf{Platform fees}: E-commerce platforms often charge piecewise-constant commission rates across different price thresholds, varying with the buyer's membership tier or contract type. 
\end{itemize}
Beyond these examples, frictions also include buyer-specific subsidies and discounts that reduce effective payments. In all these settings, each good still has a single price; buyers differ only in how this price translates into payments.
Piecewise-linear payment functions, when defined for each buyer-good pair, naturally model such structures. They preserve quasi-linear utility while relaxing the standard assumption that a bundle's payment equals the sum of its prices, necessitating models with imperfectly transferable utility (ITU). While \citet{A1989, A1992} explored ITU models for unit-demand buyers, these works did not address combinatorial~settings.

The goal of this study is to unify these two research lines by leveraging the framework of discrete convex analysis. We address a combinatorial ITU model accommodating buyers with multi-unit demands and goods with multi-unit supplies, in the presence of a broad class of frictions. Specifically, we assume strong substitutes valuations paired with piecewise-linear payment functions that are heterogeneous across buyer-good pairs. Our framework extends the paradigms of \citet{GS2000} and \citet{A2006} to this setting  while substantially generalizing the unit-demand ITU models of \citet{A1989, A1992}. This establishes a general algorithmic framework for combinatorial auctions with frictions, advancing the theoretical frontiers of both auction design and discrete optimization.

\subsection{Our Contributions}
\label{sec:contribution}
We address the computational problem of finding Walrasian equilibria in combinatorial markets with ITU. While existence and the lattice structure of equilibria have been established in more general settings \citep{FJJT2019,S2022}, 
these results are non-constructive. Our focus is primarily algorithmic: we design an ascending auction that finitely converges to the minimum equilibrium price vector. This task is highly non-trivial under frictions, where the uniform price increments of the classical auctions \citep{GS2000,A2006} induce heterogeneous payment increases and can therefore fail (see Example~\ref{need_direction}).
Since the corresponding allocation can be efficiently determined by solving a submodular intersection problem (detailed in Appendix~\ref{Allocation}), our framework computes the minimum Walrasian equilibrium. Our contributions are summarized as follows: 

\textbf{(1) Ascending Auction for Combinatorial ITU Models.} We extend the celebrated ascending auctions \citep{A2006,GS2000} to ITU models by incorporating \emph{directional price updates}. Unlike standard uniform increments, our mechanism employs heterogeneous price increases to accommodate payment functions that vary across buyer-good pairs. Our framework also substantially generalizes the unit-demand ITU models of \citet{A1989, A1992}.

\textbf{(2) Characterization of Admissible Directions.}
We identify a sufficient condition on update directions for convergence to the \emph{minimum} equilibrium price vector.
Specifically, we introduce the \emph{Local Stability Condition} \LSC, which requires that the set of goods targeted for price increases, denoted by $X^*$, remains invariant under small price perturbations. 
This condition captures the core features of the updates in \citep{A2006,GS2000}.\footnote{In their models without frictions, the characteristic vector $\chi_{\scriptscriptstyle X^*}$ inherently satisfies \LSC; see Proposition \ref{simple}.}

\textbf{(3) Strongly Polynomial-time Algorithm for Computing Directions.} 
We develop a strongly polynomial-time algorithm to compute a direction $d^*$ satisfying \LSC. Determining an admissible direction under heterogeneous payment structures is highly non-trivial and thus constitutes our main technical contribution. 

Our approach refines the primal-dual framework of \citet{ENPRVV2025}, which identifies the set $X^*$ via the polymatroid sum problem. 
To simultaneously compute $X^*$ and $d^*$, we formulate a \emph{lexicographic extension} of the primal problem and characterize its dual solution via a reduction to a convex flow problem. We then construct $d^*$ from the unique minimal dual solution, ensuring that it is supported on $X^*$.
Finally, by exploiting the $\mathrm{L}^{\natural}$-convexity of the dual objective, we formally establish that $d^*$ satisfies~\LSC.

To compute the minimal dual solution, we utilize the fact that our primal problem is a weighted polymatroid sum problem. Both the primal and minimal dual solutions can be computed in strongly polynomial time via efficient algorithms for weighted polymatroid intersection. Notably, our approach requires only access to oracles that return an element of the minimal demand set and the exchange capacity of each buyer, avoiding explicit value queries. This aligns well with the auction format by both protecting privacy and accommodating buyers unable to perfectly quantify their valuations.

 \textbf{(4) Interpretations of the Dynamics.}
We provide an economic interpretation of $d^*$. Notably, it corresponds to the minimum equilibrium price vector of an auxiliary assignment market. As this market consists of unit-demand buyers, the result of \citet{L1983} implies that these prices coincide with the VCG~prices. 

Regarding the overall dynamics, we construct a \emph{scaled Lyapunov function} for the separable case -- where marginal payments decompose into buyer- and good-specific constants -- thereby extending the potential-based framework of \citet{A2006} and \citet{MSY2016}. Finally, we identify the theoretical requirements to further generalize this potential-based approach, providing a roadmap for future developments.

\begin{table}[t]
\centering
\small
\caption{Comparison of market models and oracle requirements. All works assume substitute valuations, including unit-demand functions, gross substitutes over $2^M$, and strong substitutes over $\Z^M_+$.}
\label{tab:comparison}
\renewcommand{\arraystretch}{1.2}
\setlength{\tabcolsep}{6pt}
\begin{tabular}{l c c c c c}
\toprule
\multicolumn{1}{c}{\textbf{Literature}} 
& \textbf{Payment} &\textbf{Valuation}   & \textbf{Oracle}&\textbf{Perspective}&\textbf{Minimum} \\ 
 & \textbf{Function} &\textbf{Domain}  & && \textbf{Equilibrium}\\ 
\midrule
\citet{GS2000} & \multirow{2}{*}{Identity} & $2^M$ & \multirow{2}{*}{Demand}&\multirow{2}{*}{Dual} &\checkmark \\ 
\citet{A2006} &  & $\Z^M_+$ & & &\checkmark\\ 
\midrule
\citet{A1992} & \multirow{3}{*}{\parbox[c][3\baselineskip][c]{1.2cm}{
\centering Piecewise-linear}} & $M$ & Demand &Dual &\checkmark\\ 
\citet{BEF2024} &  & $2^M$ & Value &Primal/Dual & \checkmark\rlap{\footnotemark}\\ 
\textbf{Our Work} & & \textbf{$\Z^M_+$} & \textbf{Demand} &\textbf{Dual} &\checkmark\\ 
\bottomrule
\end{tabular}
\smallskip
\end{table}

\subsection{Comparison with Related Literature}
As summarized in Table~\ref{tab:comparison}, our framework unifies several existing market models.
\footnotetext{Balkanski et al.~\cite{BEF2024} showed that their algorithm
computes the minimum equilibrium in the third version of their full paper
(\href{https://arxiv.org/abs/2406.13620v3}{arXiv:2406.13620v3}, July 2026),
which appeared after the first version of the present paper
(\href{https://arxiv.org/abs/2604.10563v1}{arXiv:2604.10563v1}, April 2026).}
Specifically, our model generalizes the unit-demand ITU models \citep{A1989,A1992} and, in the absence of frictions, exactly recovers the celebrated auctions \citep{GS2000,A2006}.
The work most closely related to ours is \citet{BEF2024}, which proposes t\^{a}tonnement procedures for combinatorial ITU assignment games. 
While sharing the goal of computing equilibria, our framework differs from theirs in several key respects. 
Beyond the clear differences in domains and oracles, 
we highlight the following distinctions:

\textbf{(1) Alignment with Canonical Auction Frameworks.} 
Unlike their alternating-path approach, which updates allocations and prices simultaneously, ours adapts the classical approach of computing equilibrium prices first and subsequently recovering the allocation from these prices.
While their price path may differ from those of \citet{GS2000} and \citet{A2006} even in the absence of frictions, ours exactly recovers these auctions in that case (see Remark~\ref{coincide}), ensuring transparent price discovery.

\textbf{(2) Characterization of Directions.}
Unlike their graph-theoretic construction, we characterize the direction explicitly via the minimal solution to an $\mathrm{L}^{\natural}$-convex minimization problem (see Section~\ref{sec:contribution}(3) and~(4)).

\subsection{Other Related Work}
The computation of Walrasian equilibria in markets with indivisible goods has been extensively studied. In the classical assignment game, \citet{SS1971} showed that the core payoffs coincide with the set of optimal solutions to the dual of the assignment linear program. Building on this, \citet{DGS1986} proposed an ascending auction -- a variant of the Hungarian algorithm \citep{K1955} -- that converges to the minimum equilibrium prices. 
For gross substitutes valuations, \citet{KC1982} established the existence of equilibria and introduced a t\^{a}tonnement procedure. Subsequent research (e.g., \citet{BM1997} and \citet{NS2006}) further advanced linear programming formulations and computational perspectives. For a comprehensive survey on these topics, see \citet{P2017}.

The Lyapunov function approach \citep{A2006} provides a foundational framework for analyzing and computing equilibria.  Within this paradigm, \citet{PW2020} recently developed polynomial-time algorithms in the aggregate-demand oracle model. This Lyapunov-based framework has also been extended to richer preference structures and representations, including markets with limited complementarities \citep{SY2009, SY2015}, expressive bidding languages \citep{BK2019, BGKL2024}, and unimodular demand types \citep{FY2026}. While all of these existing works generalize the \textit{preference} side, this work is the first to generalize the \textit{payment} side.

Discrete convex analysis \citep{M1998, M2003} and matroid theory \citep{O2011} have provided powerful mathematical frameworks for analyzing combinatorial structures in economics. \citet{M19961, M19962} developed a valuated matroid intersection algorithm, which can be used to compute equilibrium allocations in single-unit demand settings. This framework was later extended to multi-unit demand environments by \citet{MT2003} via $\mathrm{M}^{\natural}$-convex submodular flows \citep{M1999}. Beyond valuation structures, (poly)matroid theory has also been widely used in auction design to address complex market constraints, such as allocation constraints \citep{BSV2011} and their integration with budget constraints \citep{GMP2015,LMSZ2021}. While these prior works primarily focus on allocation and budget limitations, we leverage these frameworks to incorporate market frictions into combinatorial auction~design.

\subsubsection*{Organization of the Paper}
The remainder of this paper is organized as follows.
Section~\ref{Preliminaries} introduces our model and basic definitions.
Section~\ref{Overview} presents our ascending auction and analyzes the fundamental properties of directional price updates.
Section~\ref{Direction_Computation} details the computation of the directions, utilizing lexicographic optimization and duality.
Section~\ref{Convergence} establishes the finite convergence of our algorithm.
Section~\ref{potential} develops a potential-based analysis and formulates a major open problem.
Section~\ref{Lexico} establishes fundamental properties of lexicographic optimization for convex flow problems and proves one of the key theorems in Section~\ref{Direction_Computation}. 
Section~\ref{Concluding} concludes the paper.
Unless otherwise stated, omitted proofs are provided in the~appendix.

\section{Preliminaries}
\label{Preliminaries}
Let $\R_+$ and $\Z_+$ denote the sets of nonnegative real numbers and integers, respectively.
Let $\R_{++}$ denote the set of positive real numbers.
Let $M = \{1, \dots, m\}$ denote the set of goods.
For a vector $x \in \R^M$, we denote its $j$-th component by $x_j$, its restriction to $X \subseteq M$ by $x|_X$, and the sum of components in $X$ by $x(X) \coloneqq \sum_{j \in X} x_j$. 
Vector inequalities are defined componentwise; i.e., $x \leq y$ if $x_j \leq y_j$ for all $j \in M$, and $x \lneq y$ if $x \leq y$ and $x \neq y$.
Let $\zeros$ denote the zero vector in $\R^M_+$.
For a subset $X \subseteq M$, let $\chi_{\scriptscriptstyle X} \in \{0, 1\}^M$ denote the characteristic vector of $X$, defined by $(\chi_{\scriptscriptstyle X})_j = 1$ if $j \in X$ and $0$ otherwise.
We write $\chi_{\scriptstyle j} \coloneqq \chi_{\scriptstyle\{j\}}$ and set $\chi_{\scriptstyle 0} \coloneqq \zeros$.
Define the positive and negative supports of $x$ as $\operatorname{supp}_+(x) \coloneqq \{ j \in M \mid x_j > 0 \}$ and $\operatorname{supp}_-(x) \coloneqq \{ j \in M \mid x_j < 0 \}$, respectively.
For $a, b \in \Z^M_+$ with $a \le b$, let $[a, b]_{\Z} \coloneqq \{ x \in \Z^M_+ \mid a \le x \le b \}$ be the discrete~interval.

\subsection{Our Model}
Consider a market where a seller offers $m$ types of  \textit{indivisible} goods to a set of buyers 
$N =\{1, 2, \dots, n\}$ with $n\geq 2$.  
The  \textit{supply} is denoted by a vector $s = (s_1, s_2, \dots, s_m) \in \Z_{+}^M$, where $s_j$ is the number of available units of good $j$.  
These goods are allocated through an ascending auction. 
An \textit{outcome} consists of an \textit{allocation} $\{x_i\}_{i \in N}$ and a \textit{price vector} $p = (p_1, \dots, p_m) \in \R_+^M$. 
Here, $x_i \in [\zeros, s]_{\Z}$ represents the bundle assigned to buyer $i$.
The allocation $\{x_i\}_{i \in N}$ is said to be feasible if $\sum_{i\in N} x_i \leq s$.

Each buyer $i \in N$ has a valuation function 
$v_i : [\zeros,s]_{\Z} \to \R_+$ defined over bundles.  
To model \textit{frictions} between prices and payments, we consider \textit{payment functions}, 
which are assumed to be public.\footnote{This assumption is natural, since frictions typically arise from publicly observable factors such as tax rates, transaction fees, or tariffs, which are generally common knowledge among market participants.}
For each pair $(i,j)\in N\times M$, the payment function $q_{ij}: \R_+ \to \R_+$ maps the price $p_j$ of good $j$ to the payment $q_{ij}(p_j)$ made by buyer $i$ per unit of good $j$.
We assume that each payment function $q_{ij}$ is continuous, piecewise-linear, strictly increasing, and satisfies $q_{ij}(0)=0$. 
Furthermore, we assume that, for each $i$ and $j$, there exists $\kappa\in\R_+$ such that $q_{ij}(\kappa) = v_i(\chi_{\scriptstyle j})$. 
The right derivative of $q_{ij}$ is denoted by $q'_{ij}$.\footnote{Since each $q_{ij}$ is piecewise-linear, the auctioneer needs to maintain only its breakpoints and slopes; in the applications listed in Section~\ref{intro}, the number of linear segments is typically small (e.g., a few tax brackets, tariff thresholds, or commission tiers).} 
Given a price vector $p$, the set of most-preferred bundles of buyer $i$,
called the \emph{demand set}, is defined~by
\[
    D_i(p) \coloneqq
      \operatorname*{arg\,max}_{x_i\in [\zeros,s]_{\Z}} u_i(x_i;p),
\]
where $u_i(x_i; p) \coloneqq v_i(x_i) - \sum_{j\in M} q_{ij}(p_j)\,x_{ij}$ is the \emph{utility function} of buyer $i$ for bundle $x_i$ at $p$.\footnote{Formally, $D_i$ is a demand \textit{correspondence}. We refer to the set of optimal bundles $D_i(p)$ at price $p$ as the \textit{demand set}.}

Throughout this study, we assume that each valuation function $v_i$ is:
(i) normalized, i.e., $v_i(\zeros) = 0$;
(ii) monotone, i.e., $v_i(x) \le v_i(y)$ whenever $x \leq y$; and
(iii) an \emph{$\mathrm{M}^\natural$-concave} function, defined as~follows:
\begin{definition}
A valuation function $v : [\zeros,s]_{\Z} \to \R_+ \cup \{-\infty\}$ is said to be 
${\rm M}^\natural$-concave if for all $x, x' \in [\zeros,s]_{\Z}$ and for all 
$j \in \operatorname{supp}_+(x - x')$, there exists 
$k \in \operatorname{supp}_-(x - x')\cup \{0\}$ such that
$v(x) + v(x') \leq v(x - \chi_{\scriptstyle j} + \chi_{\scriptstyle k}) + v(x' + \chi_{\scriptstyle j} - \chi_{\scriptstyle k})$.
\end{definition}

We adopt the $\mathrm{M}^{\natural}$-concavity formulation to leverage the rich framework of discrete convex analysis. It corresponds to classical substitutability conditions. For unit supplies $s \in \{0,1\}^M$, 
$\mathrm{M}^{\natural}$-concavity is equivalent \citep{FY2003} to the \emph{gross substitutes} condition of \citet{KC1982}, which requires that a buyer continue to demand any good whose price has not increased. For general supplies $s \in \Z_+^M$, a valuation is \emph{strong substitutes} if it induces a gross substitutes valuation when every unit is regarded as a separate good \citep{MS2009}; strong substitutability, together with concave-extensibility, is equivalent to $\mathrm{M}^{\natural}$-concavity~\citep[Theorem~4.1]{ST2015}.


Our goal is to compute a \textit{Walrasian equilibrium}.  
Intuitively, it is an outcome in which   
each buyer receives a most-preferred bundle at the prevailing prices 
and the market clears.
\begin{definition}
\label{Walrasian}
A Walrasian equilibrium is a pair $\bigl(\{x_i\}_{i\in N},\,p\bigr)$ such that
(i) $x_i\in D_i(p)$ for every $i\in N$, and
(ii) the allocation clears the market, i.e., $\sum_{i\in N} x_i = s$.
If such an allocation $x\coloneqq \{x_i\}_{i\in N}$ exists for a price vector $p$, we call $p$ an \emph{equilibrium price vector}.
\end{definition}
The existence of equilibria and the lattice structure of equilibrium prices have been established by \citet{FJJT2019} and \citet{S2022}. These results apply here since our assumption of symmetric frictions across units of the same good satisfies the conditions in Remarks~7 and~9 of \citet{S2022}. Building upon these foundations, we now address the computation of such equilibria via an ascending auction format.

\subsection{Polymatroid Theory}
\label{polymatroid_theory}
We recall the basic terminology of polymatroid theory.
Unless otherwise stated, the definitions and results are based on \citet{F2005} and \citet{M2003,M2022}.

Let $\Dc_i(p)$ and $\Dh_i(p)$ denote the sets of minimal and maximal elements of $D_i(p)$, respectively. Formally, $\Dc_i(p) \coloneqq \{ x_i \in D_i(p) \mid \nexists x'_i \in D_i(p) \text{ with } x'_i \lneq x_i \}$ and $\Dh_i(p) \coloneqq \{ x_i \in D_i(p) \mid \nexists x'_i \in D_i(p) \text{ with } x'_i \gneq x_i \}$.
As the valuation function $v_i$ is $\mathrm{M}^{\natural}$-concave, for any price vector $p$, each of $\Dc_i(p)$ and $\Dh_i(p)$ forms the base polyhedron of an integral polymatroid (also known as an M-convex set; see Appendix~\ref{DemandSets}).
Since we consider only integral polymatroids in this paper, we simply refer to them as polymatroids.
\begin{definition}
\label{polymatroid}
Let $\rho: 2^M \to \Z_+$ be a set function satisfying the following properties:
(i) $\rho(\emptyset)=0$ (\emph{normalization});
(ii) $\rho(S)\leq \rho(T)$ for all $S\subseteq T$ (\emph{monotonicity}); and
(iii) $\rho(S\cup \{e\})-\rho(S)\geq \rho(T\cup \{e\})-\rho(T)$ for all $S\subseteq T$ and $e\in M\setminus T$ (\emph{submodularity}).
The set of vectors defined~by
$P (\rho) \coloneqq \bigl\{\, x \in \Z_+^M \bigm| x(S) \le \rho(S) \text{ for all } S \subseteq M \,\bigr\}$
is called a \emph{polymatroid}. 
Such a function $\rho$ is called the \emph{rank function} of the polymatroid. 
When the context is clear, the polymatroid $P(\rho)$ is often denoted by $P$.
The \emph{base polyhedron} $B (\rho)$ is the set of maximal vectors in $P (\rho)$, defined by
$B (\rho) \coloneqq \bigl\{\; x \in P (\rho) \bigm| x(M) =~\rho(M)\; \bigr\}$.
\end{definition}
Note that the rank function $\rho$ is uniquely determined by $P(\rho)$ or by $B(\rho)$ via the~relation:
\begin{equation}
    \label{rank}
    \rho(S) = \max_{x \in P(\rho)} x(S) = \max_{x \in B(\rho)} x(S) \quad \text{for all } S \subseteq M.
\end{equation}

\paragraph{Basic Operations on Polymatroids}
Throughout this paper, we often use the following operations.
Note that these operations naturally extend to the associated base polyhedra.
\begin{itemize}
\setlength{\leftskip}{-0.15cm}
    \item \textbf{Restriction:}
    For a subset $S \subseteq M$, let $\rho|_S$ be the rank function restricted to $S$, i.e., $\rho|_S(S') = \rho(S')$ for any $S' \subseteq S$.
    The polymatroid $P(\rho|_S)$ on $S$ is called the \emph{restriction} of $P$ to $S$, denoted by $P|_{S}$.

\item \textbf{Contraction:}
The \emph{contraction} of $P(\rho)$ by $X \subseteq M$ is a polymatroid on $M \setminus X$, defined by the rank function $\rho^X(S) \coloneqq \rho(S \cup X) - \rho(X)$.
In terms of polyhedra, it is the projection of the face where $X$ is saturated onto $M \setminus X$, given by
$P(\rho^X) = \bigl\{\, x|_{M \setminus X} \mid x \in P(\rho),\; x(X) = \rho(X) \,\bigr\}$.
    \item \textbf{Direct Sum:}
For polymatroids $P_1(\rho_1)$ and $P_2(\rho_2)$ on disjoint sets $M_1$ and $M_2$, their \emph{direct sum} $P_1 \oplus P_2$ is the polymatroid on $M_1 \cup M_2$ defined by $\rho(S) \coloneqq \rho_1(S \cap M_1) + \rho_2(S \cap M_2)$.
A polymatroid $P(\rho)$ on $M_1 \cup M_2$ admits this decomposition if and only if $\rho(M_1 \cup M_2) = \rho(M_1) + \rho(M_2)$.
\end{itemize}

\paragraph{Linear Optimization over Polymatroids}
Consider the linear maximization of $\sum_{j \in M} w_j x_j$ over a polymatroid $P(\rho)$ or its base polyhedron $B(\rho)$ for a given weight vector $w\in \R^M$.
Let $\bar{w}_1 > \dots > \bar{w}_t$ denote the distinct values of the components of  $w$. 
These values partition the ground set $M$ into subsets $M_k \coloneqq \{ j \in M \mid w_j = \bar{w}_k \}$. 
Let $X_k \coloneqq  \bigcup_{l=1}^k M_l$ for $k=1, \dots, t$ (with $X_0 = \emptyset$). 
Then, the following holds:
\begin{fact}
\label{greedy_structure}
The optimal solutions to the problem $\max \sum_{j \in M} w_j x_j$ decompose as follows:
\begin{enumerate}
    \item[(i)] The optimal set over $B(\rho)$ decomposes into the direct sum of base polyhedra:
        \[
        \argmax_{x \in B(\rho)} \sum_{j\in M}w_{j}x_{j} = \bigoplus^{t}_{k=1} B(\rho^{X_{k-1}})|_{M_k}.
    \]
    Note that this set itself forms the base polyhedron of a polymatroid.
\item[(ii)] The optimal set over $P(\rho)$ decomposes into the direct sum of base polyhedra and a~polymatroid:
    \[
        \argmax_{x \in P(\rho)} \sum_{j\in M}w_{j}x_{j} = 
        \bigoplus^{k'}_{k=1} B(\rho^{X_{k-1}})|_{M_k} 
        \oplus P(\rho^{X_{k'}})|_{\{j\in M \mid w_{j}=0\}}
        \oplus \{ 0|_{\{j\in M \mid w_{j}<0\}} \},
    \]
    where $k' = \max \{k \mid \bar{w}_{k} > 0\}$ (with $k'=0$ if $\bar{w}_1 \le 0$).
\end{enumerate}
\end{fact}
We also address the \emph{weighted polymatroid sum problem}. Let $P_1, \dots, P_n$ be polymatroids on $M$. Given a weight $w \in \R_+^{N \times M}$ and a supply $s \in \Z_+^M$, the problem is to find a tuple $(x^*_1, \dots, x^*_n)$ such that
\[
    (x^*_1, \dots, x^*_n) \in \argmax_{x_1, \dots, x_n} 
    \left\{\, \sum_{i=1}^n \sum^{m}_{j=1} w_{ij}x_{ij} \;\middle|\; x_i \in P_i \ (\forall i),\ \sum_{i=1}^n x_i \le s \,\right\}.
\]
This problem can be solved in strongly polynomial time via algorithms for the \emph{weighted polymatroid intersection problem}, given access to the following oracles:
\begin{itemize}
    \item \textbf{Demand Oracle.}
    Given a price vector $p$, return a bundle $x_i \in \Dc_i(p)$.
    \item \textbf{Exchange Oracle.}
    Given a price vector $p$, a bundle $x_i \in D_i(p)$, and distinct indices $j, k \in M\cup \{0\}$, return the value
    $\max\{\alpha \mid x_i - \alpha\chi_{j} + \alpha\chi_{k}\in D_i(p)\}$.
\end{itemize}
These oracles, also assumed in our ascending auction framework (Section~\ref{Overview}), follow those used by \citet{ENPRVV2025} and require no additional information under frictions; a buyer's role is thus unchanged. 
Moreover, membership of a given point $y_i \in \Z^M$ in $\check D_i(p)$ can be tested using polynomially many queries to the demand and exchange oracles.
Specifically, $\check D_i(p)$ contains $y_i$ if and only if $\min \{\, \|x_i' - y_i\|_1 \;|\; x_i' \in \check D_i(p) \, \}=0$, and this M-convex function minimization can be solved by the long-step steepest-descent algorithm using one query to the demand oracle and $\mathrm{O}(m^3)$ queries to the exchange oracle \citep{OS2026}; see also Appendix~\ref{Allocation}.


\subsection{Over-demandedness}
We outline the demand-side structure underlying our analysis.
Given the demand sets, the requirement function of buyer~$i$ at price $p$ is defined by $\mu_i(X;p) \coloneqq \min \{x(X) \mid x \in D_i(p) \}$, representing the minimum amount that buyer~$i$ must demand from a subset $X\subseteq M$ at $p$. 
Using the minimal demand set $\Dc_i(p)$ and its associated rank function $\rhoc_i(\cdot; p)$, this function can be expressed~as
\begin{equation}
\label{mu_rho}
    \mu_i(X;p) = \min \{x(X) \mid x \in \Dc_i(p)\} = \rhoc_i(M; p) - \rhoc_i(M \setminus X; p).
\end{equation}
Based on this, we define the \emph{over-demand function} $\mathcal{O}(\cdot; p): 2^M \to \Z$ at price $p$ by
\begin{equation}
\label{overdemand}
    \mathcal{O}(X; p) \coloneqq \mu(X; p) - s(X)
    = \rhoc(M; p) - \rhoc(M \setminus X; p) - s(X),
\end{equation}
where $\mu(\cdot; p) \coloneqq \sum_{i \in N} \mu_i(\cdot; p)$ and $\rhoc(\cdot; p) \coloneqq \sum_{i \in N} \rhoc_i(\cdot; p)$.
A set $X\subseteq M$ is said to be \emph{over-demanded} at $p$ if $\mathcal{O}(X;p)>0$.
The function $\mathcal{O}(\cdot; p)$ is known to be supermodular (i.e., $-\mathcal{O}(\cdot; p)$ is submodular), 
and thus a unique minimal maximizer of $\mathcal{O}(\cdot; p)$ exists.
This set is referred to as the \emph{minimal maximally over-demanded set}, which we specifically focus on in our analysis.
From \eqref{overdemand}, maximizing the over-demand function is equivalent to the following minimization problem, as the constant term $\rhoc(M; p)$ does not affect the argmin:
\begin{equation}
\label{argmin}
    \argmax_{X \subseteq M} \mathcal{O}(X; p) 
    = \argmin_{X \subseteq M} \{ s(X) + \rhoc(M \setminus X; p) \}.
\end{equation}
The minimal solution of \eqref{argmin} can be computed in strongly polynomial time using submodular function minimization algorithms, e.g., \citet{LSW2015}.
Following \citet{MSY2013}, who first identified the underlying primal-dual structure, \citet{ENPRVV2025} formulated the primal as an (unweighted) polymatroid sum problem and exploited the following formulation to improve computational efficiency in settings without frictions:
\begin{equation}
\label{PD_frictionless}
(P_1) \ \max_{(x_1, \dots, x_n)} \left\{ \sum_{i \in N} x_{i}(M) \,\middle|\, x_i \in P(\rhoc_i(\cdot; p))\, (\forall i),\\ \sum_{i\in N} x_i \leq s \right\},
\ 
(D_1) \ \min_{X \subseteq M} \left\{ s(X) + \rhoc(M\setminus X; p) \right\}.
\end{equation}

Analogously, we define the concept of \emph{under-demandedness}. Let $\rhoh_i(\cdot; p)$ denote the rank function of the maximal demand set $\Dh_i(p)$ for each $i \in N$. The \emph{under-demand function} $\mathcal{U}(\cdot; p): 2^{M} \to \R$ is defined by
$\mathcal{U}(X;p) \coloneqq \rhoh(X;p) - s(X)$, 
where $\rhoh(\cdot; p) \coloneqq \sum_{i \in N} \rhoh_i(\cdot; p)$.
A set $X\subseteq M$ is said to be \emph{under-demanded} at $p$ if $\mathcal{U}(X; p) < 0$. 
Then, equilibrium price vectors are characterized as follows:\footnote{Throughout this paper, proofs omitted from the main text can be found in Appendix~\ref{proofs}, unless otherwise stated.}
\begin{lemma}
\label{price_characterization}
$p$ is an equilibrium price vector if and only if no set is over- or under-demanded at $p$.
\end{lemma}

\section{Ascending Auctions with Directional Price Updates}
\label{Overview}
This section presents an ascending auction framework that guarantees finite convergence to the minimum equilibrium price. A central feature of our approach is the use of directional price updates; we analyze a fundamental property of such updates and the key requirements for the update direction to ensure convergence to the minimum equilibrium price~vector.

\subsection{Overview of the Proposed Mechanism}
The overall procedure is presented in~Algorithm~\ref{algo1}.
\begin{algorithm}[tb]
\caption{\textsc{Ascending Auction with Directional Price Updates}}
\label{algo1}
\begin{algorithmic}[1]
\STATE Initialize the price vector: $p\coloneqq  \zeros$.
\STATE Compute the pair $(X^*, d^*)$, where $X^*$ is the unique minimal maximizer of $\mathcal{O}(\cdot; p)$ and $d^*$ satisfies the Local Stability Condition \LSC (via Algorithm~\ref{alg:direction} presented in Section~\ref{Direction_Computation}).
\STATE If $X^* = \emptyset$, output $p^*\coloneqq p$ and terminate.
\STATE If $X^* \neq \emptyset$, increase the price along $d^*$ until one of the following occurs:
\begin{itemize}
\item[(a)] The set $X^*$ identified in Line 2 ceases to be the minimal maximizer of $\mathcal{O}(\cdot; p)$.\\
\item[(b)] \LSC is violated while $X^*$ remains the minimal maximizer.
\end{itemize}
Let $\bar{c}(p, d^*)$ be the step size at $(p, d^*)$, and update the price as $p \coloneqq  p + \bar{c}(p, d^*) d^*$. 
Go to Line~2.
\end{algorithmic}
\end{algorithm}
The auction is initialized at $p \coloneqq \zeros$. In each iteration, the auctioneer identifies a pair $(X^*, d^*)$, where $X^*$ is the minimal maximally over-demanded set at the current price $p$ (i.e., the unique minimal maximizer of $\mathcal{O}(\cdot; p)$), and $d^* \in \R_+^M$ is a direction satisfying the \emph{Local Stability~Condition}~\LSC:
\begin{flalign}
  \qquad \LSC \qquad & \parbox[t]{0.8\textwidth}{
    There exists $\delta > 0$ such that for any $\varepsilon \in [0, \delta)$, the set $\operatorname{supp}_+(d^*)$ is the unique minimal maximizer of $\mathcal{O}(\cdot; p+\varepsilon d^*)$.
  } && \nonumber
  \phantomsection\def\@currentlabel{LSC}\label{LSC}
\end{flalign}

Note that setting $\varepsilon=0$ in \LSC implies $\operatorname{supp}_+(d^*)=X^*$.
This condition is essential for the stability of the price update; otherwise, the set $X^*$ could change under arbitrarily small price perturbations.

The pair $(X^*, d^*)$ is computed via Algorithm~\ref{alg:direction} (described in Section~\ref{Direction_Computation}) in strongly polynomial time, 
using the current price $p$, the derivative values $\{q'_{ij}(p_j)\}$, and the oracles introduced in Section~\ref{polymatroid_theory}.
If $X^* = \emptyset$, the auction terminates and returns $p^*$.
Otherwise, the price is updated along $d^*$ until one of the stopping events in Line~4 occurs, after which the auction proceeds to the next~iteration.\footnote{Since each $v_i$ is real-valued, the exact computation of the step size is generally intractable even without frictions; see \citet[Remark~6.8]{S2017}. It requires, for each buyer $i$, an oracle that, given $p$ and $d^*$, returns $\min\{\lambda \in \R_+ \mid D_i(p) \neq D_i(p + \lambda d^*)\}$. 
Such an oracle is implementable in $\mathrm{O}(m^2)$ time given a value oracle, so our setting remains more general even when we assume it.}


The following theorems guarantee finite convergence to the minimum equilibrium price~vector. 
\begin{theorem}
\label{finite}
Algorithm~\ref{algo1} terminates in a finite number of iterations.
\end{theorem}
\begin{theorem}
\label{minimum_eq}
The output $p^*$ of Algorithm~\ref{algo1} coincides with the minimum equilibrium price vector.
\end{theorem}
We prove Theorem~\ref{minimum_eq} at the end of this section, as it relies solely on \LSC.
In contrast, the proof of Theorem~\ref{finite} is deferred to Section~\ref{Convergence}, 
as it requires the specific construction of $d^*$ in~Algorithm~\ref{alg:direction}.

Theorem \ref{minimum_eq} guarantees the existence of 
an equilibrium allocation $x^*\coloneqq (x^*_i)_{i\in N}$ 
satisfying $x^*_{i}\in D_i(p^*)$ for each $i\in N$ and $\sum_{i \in N} x^*_i = s$. 
Given access to the oracles, 
we can compute $x^*$ in strongly polynomial time by reducing the problem to submodular intersection; see Appendix~\ref{Allocation}. 
Consequently, the pair $(x^*, p^*)$ constitutes a minimum Walrasian equilibrium, which is of economic interest as the \textit{buyer-optimal equilibrium}.

\begin{remark}
\label{coincide}
In the absence of frictions, Algorithm~\ref{algo1} recovers the auctions of \citet{GS2000} and \citet{A2006}.\footnote{Consequently, Algorithm~\ref{algo1} is not incentive compatible, as these auctions are known to lack this property already in the frictionless setting; the absence of incentive compatibility is thus inherited from the underlying paradigm rather than introduced by frictions. This does not prevent practical use, as many mechanisms used in practice are not truthful either.} 
Specifically, the set $X^*$ in Line~2 coincides with the \emph{excess demand set} \citep{GS2000}, and the direction $d^* = \chi_{\scriptscriptstyle X^*}$ satisfies \LSC (see Proposition~\ref{simple}).
In contrast, the mechanisms of \citet{A1992} and \citet{BEF2024} increase prices on a subset of $X^*$, yielding different dynamics even in the settings they~consider.
\end{remark}

The remainder of this section analyzes the fundamental properties of directional price updates.
To build intuition, we first present an example illustrating that the naive direction 
$d = \chi_{\scriptscriptstyle X^*}$ -- which trivially satisfies $\operatorname{supp}_+(d)=X^*$ -- may fail to satisfy \LSC due to the presence of frictions.

\begin{example}
\label{need_direction}
Consider a market with $N=\{1,2,3\}$ and $M=\{1,2,3\}$, each with unit supply. 
Suppose that the minimal demand sets 
at the initial price $p=(0,0,0)$ are given by
$\Dc_1(p) = \{(1,0,1), (0,1,1)\}$, \ 
$\Dc_2(p) = \{(0,1,1)\}$,\ 
$\Dc_3(p) = \{(0,1,0), (0,0,1)\}$. 
Here, each bundle is written as a vector in $\{0,1\}^M$; e.g., $(1,0,1)$ 
denotes the bundle consisting of goods~1 and~3.

To compute $X^*$, we examine the objective function $s(X) + \rhoc(M \setminus X; p)$ of the dual problem $(D_1)$ in~\eqref{PD_frictionless}. As shown in Table~\ref{rank_values1}, this function is minimized at $\{2,3\}$ and $\{1,2,3\}$, 
yielding $X^* = \{2,3\}$.
Suppose that the right derivatives of the payment functions at $p$ satisfy $q'_{ij}(0)=1$ for all $(i,j)\neq (3,2)$ and $q'_{32}(0) = 2$.

\begin{table}[t]
\centering
\caption{Values of rank functions $\rhoc_i(\cdot ;p)$ and the objective function of $(D_1)$ at $p$ in Example~\ref{need_direction}.}
\label{rank_values1}
\begin{tabular}{c|c c c c c c c c}
    $X\ (\subseteq M)$& $\emptyset$ & $\{1\}$ & $\{2\}$ & $\{3\}$ & $\{1,2\}$ & $\{2,3\}$ & $\{1,3\}$ & $\{1,2,3\}$ \\
    \hline
    $\rhoc_1(X ;p)$ & 0 & 1 & 1 & 1 & 1 & 2 & 2 & 2 \\
    $\rhoc_2(X ;p)$ & 0 & 0 & 1 & 1 & 1 & 2 & 1 & 2 \\
    $\rhoc_3(X ;p)$ & 0 & 0 & 1 & 1 & 1 & 1 & 1 & 1 \\
    $s(X)+\rhoc(M\setminus X; p)$ & 5 & 6 & 5 & 4 & 5 & \textbf{3} & 5 & \textbf{3}
\end{tabular}
\end{table}

In the following, we consider two directions $d^1, d^2\in \mathbb R^M_+$ with $\operatorname{supp}_+(d^1)=\operatorname{supp}_+(d^2)=X^*$. As will be formulated in Proposition \ref{decomposition} (iii), for any such direction $d$ and sufficiently small $\varepsilon > 0$, the minimal demand set $\Dc_i(p+\varepsilon d)$ consists of the bundles in $\Dc_i(p)$ that minimize the payment increase:
\begin{equation}
\label{demand_example}
\Dc_i(p+\varepsilon d)
=\argmax_{x_i\in \Dc_i(p)} \left\{u_i(x_i; p) - \varepsilon \sum_{j\in X^*}q'_{ij}(p_j)d_j x_{ij}\right\}
=\argmin_{x_i\in \Dc_i(p)}\sum_{j\in X^*}q'_{ij}(p_j)d_j x_{ij}.
\end{equation}
To verify \LSC for these directions, we rely on a characterization that will be formally established in Theorem~\ref{any_direction}. Specifically, a direction $d$ satisfies \LSC if and only if 
$X^*$ is the minimal minimizer of $s(X)+\rhoc(X^*\setminus X; p+\varepsilon d)$ among all subsets $X\subseteq X^*$.

Let $d^1 = \chi_{\scriptscriptstyle X^*}=(0,1,1)$ and consider $p^1\coloneqq p+\varepsilon d^1$ for sufficiently small $\varepsilon > 0$.
We apply \eqref{demand_example} to determine the minimal demand sets at $p^1$.
For Buyer 1, since $q'_{13}(p_3)d^1_3 < q'_{12}(p_2)d^1_2 + q'_{13}(p_3)d^1_3$, we have $\Dc_1(p^1) = \{(1,0,1)\}$.
For Buyer 2, since $\Dc_2(p)$ is a singleton, the minimal demand set remains unchanged: $\Dc_2(p^1) = \{(0,1,1)\}$.
For Buyer~3, since $q'_{33}(p_3)d^1_3 = 1 < 2 = q'_{32}(p_2)d^1_2$, we have  $\Dc_3(p^1) = \{(0,0,1)\}$.
As shown in Table~\ref{rank_values2}~(a), the function $s(X) + \rhoc(X^* \setminus X; p_1)$ is minimized at $\{3\}$ and $\{2,3\}$. Then, the minimal minimizer is $\{3\}\subsetneq X^*$. 
Therefore, $d^1$ does not satisfy \LSC.

Next, let $d^2 = (0, 1, 2)$ and consider $p^2 \coloneqq p+\varepsilon d^2$. 
For Buyers 1 and 2, the minimal demand sets are determined similarly to the case of $d^1$.
For Buyer~3, by  \eqref{demand_example} and $q'_{33}(p_3)d^2_3 = 2 = q'_{32}(p_2)d^2_2$, the demand set remains unchanged. Then, the minimal demand sets at $p^2$ are given by $\Dc_1(p^2) = \{(1,0,1)\},\ \Dc_2(p^2) = \{(0,1,1)\},\ \Dc_3(p^2) = \{(0,1,0), (0,0,1)\}$. 
As shown in Table~\ref{rank_values2}~(b), the value $s(X) + \rhoc(X^* \setminus X; p^2)$ is uniquely minimized at $\{2,3\}$. Therefore, the minimal minimizer remains $X^*=\{2,3\}$ and $d^2$ satisfies \LSC.

\begin{table}[t]
\caption{Values of rank functions $\rhoc_i(\cdot ;p+\varepsilon d)$ and $s(X)+\rhoc(X^*\setminus X; p+\varepsilon d)$ under different directions.}
\label{rank_values2}
\begin{minipage}{0.48\textwidth}
\centering
\subcaption{Under $d^1=(0,1,1)$}
\begin{tabular}{c|c c c c}
    $X (\subseteq X^*)$ & $\emptyset$ & $\{2\}$ & $\{3\}$ & $\{2, 3\}$ \\
    \hline
    $\rhoc_1(X; p^1)$ & 0 & 0 & 1 & 1 \\
    $\rhoc_2(X; p^1)$ & 0 & 1 & 1 & 2 \\
    $\rhoc_3(X; p^1)$ & 0 & 0 & 1 & 1 \\
    $s(X)+\rhoc(X^*\setminus X;p^1)$ & 4 & 4 & \textbf{2} & \textbf{2}
\end{tabular}
\end{minipage}
\begin{minipage}{0.48\textwidth}
\centering
\subcaption{Under $d^2=(0, 1, 2)$}
\begin{tabular}{c|c c c c}
    $X (\subseteq X^*)$ & $\emptyset$ & $\{2\}$ & $\{3\}$ & $\{2, 3\}$ \\
    \hline
    $\rhoc_1(X; p^2)$ & 0 & 0 & 1 & 1 \\
    $\rhoc_2(X; p^2)$ & 0 & 1 & 1 & 2 \\
    $\rhoc_3(X; p^2)$ & 0 & 1 & 1 & 1 \\
    $s(X)+\rhoc(X^*\setminus X; p^2)$ & 4 & 4 & 3 & \textbf{2}
\end{tabular}
\end{minipage}
\end{table}
\end{example}

In Example \ref{need_direction}, any direction $d$ with $\operatorname{supp}_+(d)=X^*$ inevitably alters the demand set of Buyer~1. However, with an appropriate choice like $d^2$, \LSC can still be satisfied. Section~\ref{Direction_Computation} investigates the structural properties enabling this and describes an algorithm to compute such a direction $d^*$.

\subsection{Monotonicity of Rank and Requirement Functions}
We first establish the monotonicity of the rank and requirement functions with respect to price changes. A key tool is the following theorem, which refines the gross substitutes and law of aggregate demand (GS\&LAD) condition \citep{MSY2013} by establishing it specifically for minimal (resp. maximal) bundles, rather than arbitrary ones.
\begin{theorem}
\label{gs_lad}
For any two price vectors $p, p' \in \R_+^M$ with $p \leq p'$, the following hold:
\begin{itemize}
    \item[{\rm (i)}] For any $x_i \in \Dc_i(p)$ (resp. $\Dh_i(p)$), there exists $x'_i \in \Dc_i(p')$ (resp. $\Dh_i(p')$) such that
    \[
        x_{ij} \leq x'_{ij} \quad (\forall j \in M : p_j = p'_j)
        \quad \text{and} \quad
        x_i(M) \geq x'_i(M).
    \]
    \item[{\rm (ii)}] For any $x'_i \in \Dc_i(p')$ (resp. $\Dh_i(p')$), there exists $x_i \in \Dc_i(p)$ (resp. $\Dh_i(p)$) such that
    \[
        x_{ij} \leq x'_{ij} \quad (\forall j \in M : p_j = p'_j)
        \quad \text{and} \quad
        x_i(M) \geq x'_i(M).
    \]
\end{itemize}
\end{theorem}

The proof is given in Appendix~\ref{DemandSets}.
Using this strengthened condition, the following proposition formalizes the monotonicity of rank and requirement functions as prices increase. While property (ii) generalizes and strengthens Theorem 2 of \citet{GS2000}, Theorem~\ref{gs_lad} enables a remarkably concise~proof.
\begin{proposition}
\label{mu_rho_relation}
Let $p, p'$ be price vectors such that $p \leq p'$, and let $Y \subseteq \{j \in M \mid p_j = p'_j\}$. Then, for each $i \in N$, the following hold:
\begin{itemize}
    \item[{\rm (i)}] The inequalities $\rhoc_i(Y; p) \leq \rhoc_i(Y; p')$ and 
    $\rhoh_i(Y; p) \leq \rhoh_i(Y; p')$ hold.
    \item[{\rm (ii)}] The inequalities $\mu_i(Y; p) \leq \mu_i(Y; p')$ and $\mu_i(M\setminus Y; p) \geq \mu_i(M\setminus Y; p')$ hold. Furthermore, if $\mu_i(M\setminus Y; p) = \mu_i(M\setminus Y; p')$, then the equalities $\rhoc_i(M; p) = \rhoc_i(M; p')$ and $\rhoc_i(Y; p) = \rhoc_i(Y; p')$ hold.
\end{itemize}
\end{proposition}

\begin{proof}
(i) Let $x \in \Dc_i(p)$ be a bundle such that $x(Y) = \rhoc_i(Y; p)$.
By Theorem~\ref{gs_lad}~(i), there exists a bundle $x' \in \Dc_i(p')$ such that $x_{j} \leq x'_{j}$ for all $j \in Y$ and $x(M) \geq x'(M)$.
Then, we have
$\rhoc_i(Y; p') = \max \{\check{x}_i(Y) \mid \check{x}_i \in \Dc_i(p')\} \geq x'(Y) \geq x(Y) = \rhoc_i(Y; p)$.
Applying the same argument to $\Dh_i$, we also have $\rhoh_i(Y; p') \geq \rhoh_i(Y; p)$.

(ii) Let $y' \in \Dc_i(p')$ be a bundle such that $y'(Y) = \mu_i(Y; p')$.
By Theorem~\ref{gs_lad}~(ii), there exists a bundle $y \in \Dc_i(p)$ such that $y_{j} \leq y'_{j}$ for each $j \in Y$.
Since $y \in \Dc_i(p)$, we have
$\mu_i(Y; p) = \min \{\check{x}_i(Y) \mid \check{x}_i \in \Dc_i(p)\} \leq y(Y) \leq y'(Y) = \mu_i(Y; p')$.

Let $z \in \Dc_i(p)$ be a bundle such that $z(Y) = \rhoc_i(Y; p)$.
By Theorem~\ref{gs_lad}~(i), there exists a bundle $z' \in \Dc_i(p')$ such that $z_{j} \leq z'_{j}$ for each $j \in Y$ and $z(M) \geq z'(M)$.
This implies $\rhoc_i(M; p) = z(M) \geq z'(M) = \rhoc_i(M; p')$ and $\rhoc_i(Y; p) = z(Y) \leq z'(Y) \leq \rhoc_i(Y; p')$.
By (\ref{mu_rho}), we have 
$\mu_i(M\setminus Y; p)=\rhoc_i(M; p) - \rhoc_i(Y; p)\geq \rhoc_i(M; p') - \rhoc_i(Y; p')=\mu_i(M\setminus Y; p')$, with equality only if 
$\rhoc_i(M; p) = \rhoc_i(M; p')$ and $\rhoc_i(Y; p) = \rhoc_i(Y; p')$.
\end{proof}

\subsection{Directional Price Updates}
To analyze an iteration of Algorithm~\ref{algo1}, we investigate the behavior of demand sets when the current price vector $p$ is perturbed along an arbitrary direction $d \in \R_{+}^M$. We first show that for small perturbations, the updated demand sets are contained within the current ones.
\begin{lemma}
\label{subset}
For any direction $d \in \R^M_+$, there exists $\delta > 0$ such that for all $\varepsilon \in (0, \delta)$ and each buyer $i \in N$, we have $D_i(p+\varepsilon d) \subseteq D_i(p)$ and $\Dc_i(p+\varepsilon d) \subseteq \Dc_i(p)$. 
\end{lemma}

Hereafter, we implicitly fix $\delta$ such that Lemma~\ref{subset} holds for the given direction.
Recall that $\rhoc_i(\cdot; p)$ and $\rhoh_i(\cdot; p)$ denote the rank functions of the minimal demand set $\Dc_i(p)$ and the maximal demand set $\Dh_i(p)$, respectively.
Let $X \coloneqq \operatorname{supp}_+(d)$ and $\rhot_i(\cdot)$ denote the rank function on $X$ obtained by contracting the set $M \setminus X$ from $\rhoc_i(\cdot; p)$, formally defined by $\rhot_i \coloneqq (\rhoc_i)^{M \setminus X}$.
The following proposition establishes the decomposition of these rank functions and the minimal demand set at~$p+\varepsilon d$.

\begin{proposition}
\label{decomposition}
For any $\varepsilon \in (0, \delta)$, the following decomposition properties hold:
\begin{enumerate}
\item[{\rm (i)}] The inclusion $D_i(p+\varepsilon d) \subseteq \{x_i\in D_i(p) \mid x_i(X)=\mu_i(X; p)\}$ holds.
\item[{\rm (ii)}] For any $X_1 \subseteq X$ and $X_2 \subseteq M \setminus X$, the rank functions decompose as $\rhoc(X_1 \cup X_2; p+\varepsilon d) = \rhoc(X_1; p+\varepsilon d) + \rhoc(X_2; p)$ and $\rhoh(X_1 \cup X_2; p+\varepsilon d)= \rhoc(X_1; p+\varepsilon d) + \rhoh(X_2; p)$.
    \item[{\rm (iii)}] For each $i\in N$, the minimal demand set $\Dc_i(p+\varepsilon d)$ admits the direct sum decomposition:
    \[
        \Dc_i(p+\varepsilon d) =  \Dc_i(p)|_{M\setminus X}\oplus  \argmin_{y_i \in B(\rhot_i)} \sum_{j \in X} q'_{ij}(p_j) d_j y_{ij}.
    \]
\end{enumerate}
\end{proposition}
Recall that $\Dc_i(p)|_{M\setminus X}=B(\rhoc_i(\cdot;p)|_{M\setminus X})$. 
Property~(i) shows that the price update restricts the demand set at $p+\varepsilon d$ to the subset of the demand set at $p$ satisfying $x_i(X)=\mu_i(X; p)$.
Property~(ii) shows that the rank functions decompose additively, 
and their values on $M \setminus X$ are determined by those at price $p$.
Property~(iii) shows that $\Dc_i(p+\varepsilon d)$ shrinks to the subset of $\Dc_i(p)$ that minimizes the payment increase. Crucially, an inappropriate choice of $d$ may cause an excessive reduction in the demand set, violating~\LSC. 

In the proof, we use the following elementary lemma from discrete convex analysis.
\begin{lemma}
\label{perfect-matching}
For each $i\in N$ and $x_i \in D_i(p)$, if $x_i(X) > \mu_i(X; p)$ for some $X \subseteq M$, then there exist $j \in X$ and $j' \in (M\setminus X)\cup\{0\}$ such that $x_i - \chi_{j} + \chi_{j'} \in D_i(p)$.
\end{lemma}

\begin{proof}[Proof of Proposition~\ref{decomposition}.]
Choose $\delta > 0$ such that Lemma~\ref{subset} holds for all $i \in N$ under $d$, and fix $\varepsilon \in (0, \delta)$.

(i) Consider an arbitrary bundle $x_i \in D_i(p+\varepsilon d)$.
By Lemma \ref{subset}, we have $x_i \in D_i(p)$. 
From the definition of $\mu_i$, it holds that
$x_i(X) \geq \min\{x'_i(X)\mid x'_i\in D_i(p)\} =\mu_i(X; p)$.
Suppose to the contrary that $x_i(X) > \mu_i(X; p)$.
Applying Lemma \ref{perfect-matching} to $X$, there exist $j \in X$ 
and $j' \in (M \setminus X)\cup \{0\}$ such that $z_i \coloneqq x_i - \chi_{j} + \chi_{j'} \in D_i(p)$. Since $x_i,z_i\in D_i(p)$, it holds $u_i(x_i; p) = u_i(z_i; p)$.
For the utility of buyer $i$ at $p+\varepsilon d$, we have 
\begin{align*}
u_i(z_i; p+\varepsilon d) = u_i(z_i; p) - \varepsilon \sum_{\ell\in X}q'_{i\ell}(p_\ell)d_\ell z_{i\ell} 
= u_i(x_i; p+\varepsilon d) +\varepsilon q'_{ij}(p_{j})d_{j}
>u_i(x_i; p+\varepsilon d), 
\end{align*}
where the inequality follows from $\operatorname{supp}_+(d)=X$.
This contradicts the assumption that $x_i \in D_i(p+\varepsilon d)$. 
Consequently, for any $x_i \in D_i(p+\varepsilon d)$, we have $x_i \in D_i(p)$ and $x_i(X) = \mu_i(X; p)$.

(ii) Property (i) shows that for each $i \in N$, the value $x_i(X)$ is constant across all $x_i \in D_i(p+\varepsilon d)$. 
This implies that $\mu_i(X; p+\varepsilon d) = \rhoc_i(X; p+\varepsilon d) = \rhoh_i(X; p+\varepsilon d)$.
Using \eqref{mu_rho}, this~yields 
\begin{align*}
\rhoc_i(X; p+\varepsilon d)
=\mu_i(X; p+\varepsilon d)
=\rhoc_i(M; p+\varepsilon d)-\rhoc_i(M\setminus X; p+\varepsilon d)
=\rhoh_i(M; p+\varepsilon d)-\rhoh_i(M\setminus X; p+\varepsilon d).
\end{align*}
By Lemma \ref{subset}, the inclusions 
$D_i(p+\varepsilon d)\subseteq D_i(p)$ and 
$\Dc_i(p+\varepsilon d)\subseteq \Dc_i(p)$ hold. 
Then, we have
\[
\rhoc_i(X_2; p+\varepsilon d) = \max \{x_i(X_2) \mid x_i \in \Dc_i(p+\varepsilon d)\} \leq \max \{x_i(X_2) \mid x_i \in \Dc_i(p)\} = \rhoc_i(X_2; p).
\]
Replacing $\Dc_i$ with $D_i$ also yields $\rhoh_i(X_2; p+\varepsilon d) \leq \rhoh_i(X_2; p)$. The opposite directions hold by Proposition~\ref{mu_rho_relation} (i).
Consequently, for any $X_1 \subseteq X$ and $X_2 \subseteq M \setminus X$, we have
$\rhoc_i(X_1 \cup X_2; p+\varepsilon d)
    = \rhoc_i(X_1; p+\varepsilon d) + \rhoc_i(X_2; p+\varepsilon d)
    = \rhoc_i(X_1; p+\varepsilon d) + \rhoc_i(X_2; p)$.
Similarly, we have $\rhoh_i(X_1 \cup X_2; p+\varepsilon d) =\rhoc_i(X_1; p+\varepsilon d) + \rhoh_i(X_2; p)$. 
Summing over all $i \in N$ yields the desired equalities. 

(iii) 
From property (ii), the minimal demand set of each buyer admits the following decomposition:
\begin{align*}
 \Dc_i(p+\varepsilon d) &= \Dc_i(p+\varepsilon d)|_{X} \oplus \Dc_i(p+\varepsilon d)|_{M \setminus X}
 =\Dc_i(p+\varepsilon d)|_{X} \oplus \Dc_i(p)|_{M \setminus X}.
\end{align*}
By Lemma \ref{subset}, we have
$\Dc_i(p+\varepsilon d)|_{X} \oplus \Dc_i(p)|_{M\setminus X}
=\Dc_i(p+\varepsilon d) \subseteq \Dc_i(p)$, which implies $\Dc_i(p+\varepsilon d)|_{X} \subseteq 
B(\rhoc^{M\setminus X}_i) = B(\rhot_i)$.
For the utility of buyer~$i$ for a bundle~$x_i \in \Dc_i(p)$ at price $p+\varepsilon d$, 
we have
\begin{align*}
    u_i(x_i; p+\varepsilon d)
    = v_i(x_i) - \sum_{j \in M} \left( q_{ij}(p_j) + \varepsilon q'_{ij}(p_j) d_j \right) x_{ij} 
    = u_i(x_i; p) - \varepsilon \sum_{j \in X} q'_{ij}(p_j) d_{j} x_{ij}.
\end{align*}
Note that for any $x_i \in \Dc_i(p)$, the value $u_i(x_i; p)$ is constant.
For sufficiently small $\varepsilon > 0$, maximizing the utility over $\Dc_i(p)$ is equivalent to minimizing the term $\sum_{j \in X} q'_{ij}(p_j) d_j y_{ij}$ over $y_i\in B(\rhot_i)$.
Thus, we have
\[
    \Dc_i(p+\varepsilon d)
    = \argmax_{x_i\in \Dc_i(p)} u_i(x_i; p+\varepsilon d)
    = \Dc_i(p)|_{M\setminus X} \oplus \argmin_{y_i\in B(\rhot_i)} \sum_{j \in X} q'_{ij}(p_j) d_j  y_{ij}.\qquad 
\]
\end{proof}

Let $X^*$ be the minimal maximally over-demanded set at $p$, i.e., the unique minimal minimizer of  $(D_1)$ for~$p$.  Using Proposition~\ref{decomposition} (ii), we characterize \LSC\, for a direction $d$ with $\operatorname{supp}_+(d)=X^*$:

\begin{theorem}
\label{any_direction}
Let $d$ be a direction with $\operatorname{supp}_+(d)=X^*$. 
For any $\varepsilon \in (0, \delta)$, the following are~equivalent: 
\begin{itemize}
    \item[{\rm (a)}] The direction $d$ satisfies {\rm \LSC}.
    \item[{\rm (b)}] The set $X^*$ is the unique minimizer of $s(X) + \rhoc(X^*\setminus X; p+\varepsilon d)$ among all subsets $X \subseteq X^*$.
    \item[{\rm (c)}] The inequality $s(X) < \rhoc(X; p+\varepsilon d)$ holds for any non-empty subset $X \subseteq X^*$.
\end{itemize}
\end{theorem}

\begin{proof}
${\rm (a)} \Leftrightarrow {\rm (b)}:$
By Proposition \ref{decomposition} (ii), the objective value of ${\rm (D_1)}$ in \eqref{PD_frictionless} at $p+\varepsilon d$ is given by
\begin{align*}
\min_{X\subseteq M} &\big\{ s(X) + \rhoc(M\setminus X; p+\varepsilon d) \big\}\\
&=\min_{X_1\subseteq X^*,\ X_2\subseteq M\setminus X^*} 
\big\{ s(X_1 \cup X_2) + \rhoc(X^*\setminus X_1; p+\varepsilon d)+
\rhoc((M\setminus X^*)\setminus X_2; p)
\big\}\\
&=
\min_{X_1\subseteq X^*} \big\{ s(X_1) + \rhoc(X^*\setminus X_1; p+\varepsilon d) \big\}+
\min_{X_2\subseteq M\setminus X^*} \big\{ s(X_2) + \rhoc((M\setminus X^*)\setminus X_2; p) \big\}.
\end{align*}
This decomposition implies that a subset $X\subseteq M$ minimizes the total objective function if and only if its components $X\cap X^*$ and $X\cap (M \setminus X^*)$ minimize the first and second terms, respectively.
Since $X^*$ is the minimal minimizer of  $(D_1)$, it follows that 
$\rhoc(M\setminus X^*; p) \le s(X_2)+\rhoc((M\setminus X^*)\setminus X_2; p)$ 
for any $X_2 \subseteq M \setminus X^*$.
This implies that the minimum of the second term is achieved at $\emptyset$. 
Therefore, \LSC\, holds if and only if $X^*$ is the unique minimizer of the first term, which corresponds to (b).

(b) $\Leftrightarrow$ (c): Condition (b) is equivalent to the inequality
$s(X^*) < s(X) + \rhoc(X^*\setminus X; p+\varepsilon d)$ for all $X \subsetneq X^*$.
Subtracting $s(X)$ from both sides yields $s(X^*\setminus X) < \rhoc(X^*\setminus X; p+\varepsilon d)$.
By setting $Y \coloneqq X^* \setminus X$, this is equivalent to $s(Y) < \rhoc(Y; p+\varepsilon d)$ for all $\emptyset \neq Y \subseteq X^*$, which is exactly (c).
\end{proof}

\subsection{Convergence to the Minimum Equilibrium Price Vector}
The set of equilibrium price vectors forms a distributive lattice, as established by \citet{S2022}. This guarantees the existence of a unique minimum equilibrium price vector, denoted by $\underline{p}\in \mathbb R^M_+$. 
We now prove Theorem~\ref{minimum_eq}, which states that the output $p^*$ of Algorithm~\ref{algo1} coincides with $\underline{p}$. Regarding the choice of $d^*$, 
our argument relies solely on \LSC.

In the following, we prove Theorem \ref{minimum_eq}. We first show that $p^*$ is an equilibrium price vector, and then verify its minimality. 
A key ingredient for the first part is the absence of under-demanded~sets.

\begin{proposition}
\label{no_underdemanded}
Throughout the execution of Algorithm~\ref{algo1}, no set is under-demanded.
\end{proposition}
We prove this proposition by induction using the following lemma, which establishes the local property that the absence of under-demanded sets is preserved under a sufficiently small price perturbation.
\begin{lemma}
\label{no_underdemanded_local}
Let $d\in \R^M_+$ be a direction with $\operatorname{supp}_+(d)=X^*$ satisfying \LSC. If no set is under-demanded at $p$, then no set is under-demanded at 
$p+\varepsilon d$ for any $\varepsilon \in (0,\delta)$.
\end{lemma}

\begin{proof}
Since no set is under-demanded at $p$, it holds that $\rhoh(X;p) - s(X)\geq 0$ for any $X\subseteq M$.
By the choice of $d$, Theorem \ref{any_direction} (c) yields $s(X_1) \leq \rhoc(X_1; p+\varepsilon d)$ for any $X_1 \subseteq X^*$. 
Then, for any $X_1 \subseteq X^*$ and $X_2 \subseteq M \setminus X^*$, we have 
$\rhoh(X_1 \cup X_2; p+\varepsilon d)- s(X_1 \cup X_2)
        = \rhoc(X_1; p+\varepsilon d) - s(X_1) + \rhoh(X_2; p) - s(X_2) 
        \geq \rhoh(X_2; p) - s(X_2)\geq 0$,
where the equality holds by Proposition~\ref{decomposition}~(ii).
Therefore, no set is under-demanded at $p+\varepsilon d$.
\end{proof}

\begin{proof}[Proof of Proposition \ref{no_underdemanded}.]
We prove this claim by induction on the number of iterations. Consider the initial state $p=\zeros$. By monotonicity, every buyer demands all available goods, so $s \in \Dh_i(\zeros)$ for each $i \in N$. This implies $\rhoh(X; \zeros) \geq s(X)$ for all $X \subseteq M$ and thus no set is initially under-demanded.

Assume that no set is under-demanded at the start of an iteration. Let $p$ be the price at the end of the iteration, reached along the direction $d^*$.
Since $d^*$ satisfies \LSC\, throughout the iteration, Lemma~\ref{no_underdemanded_local} guarantees the absence of under-demanded sets at any intermediate price. Now, consider the behavior of demand sets as the price approaches $p$ from below. 
For a sufficiently small $\varepsilon > 0$, we can ensure that the demand sets of buyers remain unchanged over the interval $[p - \varepsilon d^*, p)$. 
Then, we~have
$s(X) \leq \rhoh(X; p - \varepsilon d^*)$ for all $X \subseteq M$.
Due to the upper semicontinuity of the demand correspondence, the inclusion $D_i(p - \varepsilon d^*) \subseteq D_i(p)$ holds for every buyer $i \in N$.
This implies $\rhoh(X; p - \varepsilon d^*) \leq \rhoh(X; p)$ 
and thus we have $s(X) \leq \rhoh(X; p)$, which completes the inductive step.
Therefore, no set is under-demanded throughout the execution of Algorithm~\ref{algo1}.

\end{proof}

Now we are ready for the proof of Theorem \ref{minimum_eq}.
\begin{proof}[Proof of Theorem \ref{minimum_eq}.]
The termination condition of Algorithm~\ref{algo1}, together with Proposition~\ref{no_underdemanded}, guarantees that no set is over-demanded or under-demanded at $p^*$.
By Lemma~\ref{price_characterization}, this implies that $p^*$ is an equilibrium price vector and satisfies $p^* \geq \underline{p}$. Therefore, it suffices to show that $p^* \le \underline{p}$.

Suppose to the contrary that $p^* \not\le \underline{p}$.
On the price path of Algorithm~\ref{algo1}, there exists a price $p$ such that $p \le \underline{p}$, but any further update along the current direction $d^*$, which satisfies \LSC\, at $p$, violates this bound; specifically, $p_j = \underline{p}_j$ for some $j \in \operatorname{supp}_+(d^*)$.
Since $d^*$ satisfies \LSC\, at $p$, the set $X^* \coloneqq \operatorname{supp}_+(d^*)$ is the minimal minimizer of ${\rm (D_1)}$ at $p$ and thus satisfies
\begin{equation}
\label{min_overdemanded}
s(X^*)+\rhoc(M\setminus X^*; p) < s(X)+\rhoc(M\setminus X; p)\quad (X\subsetneq X^*).
\end{equation}


Now we partition $X^*$ into $X_1 \coloneqq \{ j \in X^* \mid p_j = \underline{p}_j \}$ and $X_2 \coloneqq X^* \setminus X_1$.
As $p_j = \underline{p}_j$ holds for at least one $j \in \operatorname{supp}_+(d^*)=X^*$, we have $X_1 \neq \emptyset$, which implies $X_2 \subsetneq X^*$.
We now consider a (virtual) direction 
$d' \in \R^M_+$ with $\operatorname{supp}_+(d')=X_2$.
Choose $\delta > 0$ such that Lemma \ref{subset} holds for this new $d'$.
Since $p_j < \underline{p}_j$ for all $j \in X_2$, there exists a sufficiently small $\varepsilon \in (0,\delta)$ 
such that $p + \varepsilon d' < \underline{p}$.
Note that $p_j+\varepsilon d'_j =  \underline{p}_j$ for all $j \in X_1$.
Applying Proposition \ref{mu_rho_relation} to $p+\varepsilon d'$ and $\underline{p}$ for the set $X_1$, we have 
\begin{equation*}
\mathcal O(X_1; \underbar{$p$})=\sum_{i\in N}\mu_i(X_1; \underbar{$p$})-s(X_1)
\geq \sum_{i\in N}\mu_i(X_1;  p+\varepsilon d')-s(X_1).
\end{equation*}
For the right-hand side (RHS) of this inequality, we have 
\begin{align*}
{\rm (RHS)}&=
\rhoc(M;  p+\varepsilon d')-\rhoc(M\setminus X_1;  p+\varepsilon d')-s(X_1)&& ({\rm by\  (\ref{mu_rho})})\\
&=\rhoc(X_2; p+\varepsilon d')+\rhoc(M\setminus X_2; p)-\rhoc(M\setminus X_1; p+\varepsilon d') -s(X_1)&& ({\rm by\  Proposition\ \ref{decomposition}\ (ii)})\\
&\geq -\rhoc(M\setminus X^*; p+\varepsilon d')+\rhoc(M\setminus X_2; p)-s(X_1) && ({\rm by\ submodularity\ of\ }\rhoc)\\
&=- \left(\rhoc(M\setminus X^*; p)+s(X^*)\right)+\left(s(X_2)+\rhoc(M\setminus X_2; p)\right)  && ({\rm by\  Proposition\  \ref{decomposition}\ (ii)})\\
&>0. &&({\rm by\  (\ref{min_overdemanded})})
\end{align*}
This implies $\mathcal O(X_1; \underbar{$p$})>0$, 
which contradicts Lemma \ref{price_characterization}.
Therefore, we conclude that $p^* = \underline{p}$.
\end{proof}

\section{Computing the Directions for Price Updates}
\label{Direction_Computation}

In this section, we focus on the computation of a direction satisfying \LSC. Such a direction is not necessarily unique; we compute a specific one derived from an $\mathrm{L}^{\natural}$-convex minimization problem. We first present an explicit construction for a special case, followed by a general method based on lexicographic optimization and duality. Throughout this section, we fix the current price vector $p$ and let $X^*$ denote the minimal maximally over-demanded set at $p$.

\subsection{Warm-Up: Direction for the Separable Case}
Here, we examine a special case where the payment functions are separable across buyers and goods. 
\begin{assumption}
\label{simple-case}
There exist $\alpha\in \R^N_{++}$ and $\beta\in \R^M_{++}$ such that 
$q_{ij}'(p_j) = \alpha_i\,\beta_j$ for all $i \in N$, $j \in M$, and $p_j\in \R_+$.
\end{assumption}
This assumption captures scenarios such as international trade and regional taxation, where frictions arise from the interplay between buyer-specific and good-specific factors.\footnote{This assumption extends to the case where each $\beta_j$ is piecewise constant in $p_j$. In this case, 
whenever some $\beta_j$ changes, the auction restarts from the current price vector as its initial point, with the potential function (introduced in Section~\ref{potential}) updated accordingly.} 
Under this setting, Algorithm~\ref{algo1} generalizes the ascending auctions of \citet{GS2000} and \citet{A2006}, recovering their models when $\alpha_i = 1$ and $\beta_j = 1$ for all $i$ and $j$.
We now show that a simple direction supported on $X^*$ satisfies~\LSC.

\begin{proposition}
\label{simple}
Under Assumption~\ref{simple-case}, let $d^* \in \R_+^M$ be a direction defined by $d^*_j = 1/\beta_j$ for all $j \in X^*$ and $d^*_j = 0$ for all $j \in M \setminus X^*$.
Then, the direction $d^*$ satisfies \LSC.
\end{proposition}
\begin{proof}
Under Assumption \ref{simple-case}, setting $d^*_j = 1/\beta_j$ for all $j \in X^*$ yields, for any $i \in N$:
\[
    \argmin_{y_i \in B(\rhot_i)} \sum_{j \in X^*} q'_{ij}(p_j) d^*_j y_{ij}
    = \argmin_{y_i \in B(\rhot_i)} \sum_{j \in X^*}  \alpha_i y_{ij}
    =\argmin_{y_i \in B(\rhot_i)}  y_{i}(X^*)=B(\rhot_i),
\]
where the last equality holds since $y_i(X^*)$ is constant over $B(\rhot_i)$.\footnote{Specifically, $B(\rhot_i)$ corresponds to the contraction of $B(\rhoc_i)$ by $M \setminus X^*$ (see Proposition~\ref{decomposition}). Thus, for any $y_i \in B(\rhot_i)$, the total sum is constant: $y_i(X^*) = \rhoc_i(M) - \rhoc_i(M \setminus X^*)$.}
By Proposition~\ref{decomposition} (iii), we have
\begin{equation}
\label{decompose_separable}
    \Dc_i(p+\varepsilon d^*)
    = \Dc_i(p)|_{M\setminus X^*} \oplus \argmin_{y_i \in B(\rhot_i)} \sum_{j \in X^*} q'_{ij}(p_j) d^*_j y_{ij}
    = \Dc_i(p)|_{M\setminus X^*} \oplus B(\rhot_i).
\end{equation}
Since $X^*$ is the minimal minimizer of $(D_1)$ for $p$, $s(X^*)+\rhoc(M\setminus X^*;p)<s(X^*\setminus X)+\rhoc(X\cup (M\setminus X^*);p)$ for any non-empty $X \subseteq X^*$.
This yields $s(X) < \rhoc(X\cup (M\setminus X^*);p)-\rhoc(M\setminus X^*;p)=\sum_{i\in N}\rhot_i(X)=\rhoc(X;p+\varepsilon d^*)$, 
where the last equality holds by \eqref{decompose_separable}.
By Theorem~\ref{any_direction} [${\rm (c)}\Rightarrow{\rm (a)}$], we conclude that $d^*$ satisfies \LSC.
\end{proof}

When a direction $d$ satisfying \LSC is available in closed form -- as in Proposition~\ref{simple} -- Line~2 of Algorithm~\ref{algo1} reduces to identifying $X^*$ and outputting the associated direction $d$.
The computational complexity of this iteration is $O(n \cdot \mathrm{DO} + nm^3 \cdot \mathrm{ExO})$ following the algorithm of \citet{ENPRVV2025}, where $\mathrm{DO}$ and $\mathrm{ExO}$ denote the costs of a single call to the demand and exchange oracles, respectively.
In the general case, however, such a closed-form direction is not guaranteed to exist. This necessitates a more systematic~approach.

\subsection{Direction in the General Case}
To compute a valid direction, 
we develop a framework utilizing lexicographic optimization and duality. 
We outline the high-level idea below, with the technical roadmap in Figure~\ref{fig:flowchart}. 
Since the structure of the optimal set in Fact~\ref{greedy_structure} (i) depends solely on the ordering of coefficients, taking the logarithm preserves this structure while transforming the multiplicative coefficients into additive terms. For any 
$d$ with $\operatorname{supp}_+(d)=X^*$, the decomposition in Proposition~\ref{decomposition}~(iii) is equivalent~to:
    \begin{equation}
    \label{decomposition_log}
\Dc_i(p+\varepsilon d) 
=\Dc_i(p)|_{M\setminus X^*}\oplus \argmin_{y_i \in B(\rhot_i)} \sum_{j \in X^*} (\log q'_{ij}(p_j) + \log d_j) y_{ij},
    \end{equation}
where we recall that $\rhot_i \coloneqq (\rhoc_i (\cdot;p))^{M \setminus X^*}$.
This indicates that the demand restricted to $X^*$ at $p+\varepsilon d$ is determined by minimizing a linear objective with coefficients $\log q'_{ij}(p_j) + \log d_j$.
Exploiting the fact that the term $\log q'_{ij}(p_j)$ is fixed at the current price, we extend $(P_1)$ to a lexicographic optimization problem, denoted as $(P_2)$, that minimizes $\sum_{i,j} \log q'_{ij}(p_j) x_{ij}$ as a secondary objective.\footnote{This approach is consistent with previous frameworks for ITU models (e.g.,~\citet{A1992} and \citet{BEF2024}) but substantially generalizes them by accommodating strong substitutes valuations under demand and exchange oracles.}

We then analyze the dual of $(P_2)$, denoted as $(D_2)$, and construct the desired direction 
from its minimal optimal solution.
To justify this construction, we utilize the characterization in Theorem \ref{any_direction} and 
the $\mathrm{L}^{\natural}$-convexity of the dual objective function. The dual analysis of the lexicographic extension, leveraging its discrete convexity to determine the direction, is our main technical contribution.

\begin{figure}[t]
    \centering
    \begin{tikzpicture}[
  font=\small,
  >=Latex,
  node distance=10mm and 18mm,
  every node/.style={align=left, inner sep=2pt},
  dashedlink/.style={dashed, line width=0.6pt},
  vlink/.style={-Latex, line width=0.6pt},
  dblv/.style={{Latex}-{Latex}, line width=0.6pt},
  dblone/.style={double, -{Latex}, line width=0.6pt},
  dbldbl/.style={double, {Latex}-{Latex}, line width=0.6pt}
  flow/.style={-{Latex}, line width=0.6pt},
  flowd/.style={-{Latex}, dashed, line width=0.6pt},
  imp/.style={double, -{Implies}, line width=0.6pt},
  equiv/.style={double, {Implies}-{Implies}, line width=0.6pt},
]

\node (P1) at (0, 1.6) {$(P_1)$};
\node (D1) [right=28mm of P1] {$(D_1)$};
\draw[dashedlink] (P1.east) -- (D1.west)
  node[midway, above=1mm] {Duality};

\node (P2) [below=12mm of P1] {$(P_2)$};
\node (D2) [right=28mm of P2] {$(D_2)$};
\draw[dashedlink] (P2.east) -- (D2.west)
  node[midway, below=1mm] {Duality (Thm 4.3)};

\draw[vlink] (P1.south) -- (P2.north)
  node[midway, left=2mm] {Lexicographic\\Extension};


\node (F1) [right=1mm of D1, yshift=-1mm]
  {$\displaystyle \min_{X\subseteq M}\ \{\,s(X)+\check{\rho}(M\setminus X;p)\,\}$};

\node (F2) [right=2mm of D2, yshift=-1mm] {$\displaystyle \min_{z\ge \bm{0}}\ \ \mathcal L(z)$};
\node (Xstard) [below=8mm of D2] {Both $X^{*}$ and $d^{*}$ can be recovered \\ from the minimal solution};
\draw[vlink, dashed] (Xstard.north) -- (D2.south) 
  node[midway, anchor=west, right=2mm] {};

\node (poly) [right=19mm of D2]
  {$\cdots$\quad Polyhedral ${\rm L}^{\natural}$-convex \ \ (Thm 4.3)};
\node (ineq) [below=7mm of poly]
  {\qquad\qquad\qquad$s(X)\ <\ \check{\rho}(X; p+\varepsilon d^{*})$ \ for\ any\ $\emptyset\neq X\subseteq X^{*}$};

\draw[imp] (poly.south) -- (ineq.north)
  node[midway, right=3mm] {(Thm 4.4, 4.5)};
\node (lsc) [below=7mm of ineq, font=\bfseries] {(LSC)};

\draw[equiv] (ineq.south) -- (lsc.north)
  node[midway, right=3mm, font=\small] {(Thm 3.10)};
\end{tikzpicture}
\caption{Overview of the proposed framework. The subsequent analysis focuses on the problem $(D_2)$, the dual of $(P_2)$. Unlike $(D_1)$, which yields only the set $X^*$, the minimal solution $z^*$ of $(D_2)$ identifies both the set $X^*$ and the direction $d^*$ satisfying \LSC.}
    \label{fig:flowchart}
\end{figure}
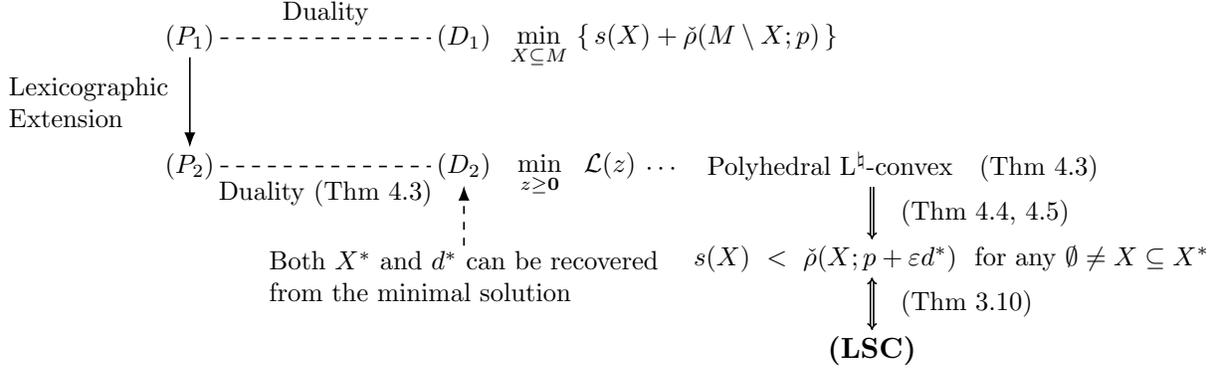
For the lexicographic extension of $(P_1)$, define the weight vector $w \in \R^{N \times M}_+$ by
$w_{ij} \coloneqq 1 - \gamma \log q'_{ij}(p_j)$ 
for all  $(i, j) \in N \times M$, 
where $\gamma>0$ is a sufficiently small constant. 
Hereafter, comparisons involving $\gamma$ are based on the lexicographic order; that is, we define $a_1 + \gamma b_1 > a_2 + \gamma b_2$ if $a_1 > a_2$ or ($a_1 = a_2$ and $b_1 > b_2$).
The other relations are defined analogously. 
Consider the following weighted polymatroid sum problem:
\[
   (P_2) \quad \max_{(x_1, \dots, x_n)} \Biggl\{\, \sum_{i \in N} \sum_{j \in M} w_{ij}x_{ij} \;\mid\; x_i \in P(\rhoc_i) \ (\forall i),\ \ \sum_{i \in N} x_i \le s \,\Biggr\},
\]
where we abbreviate $P(\rhoc_i(\cdot; p))$ as $P(\rhoc_i)$.
Since $(P_2)$ reduces to $(P_1)$ as $\gamma \to 0$, the problem $(P_2)$ seeks a solution from the set of optimal solutions to $(P_1)$ that minimizes the secondary objective $\sum_{i,j} \log q'_{ij}(p_j) x_{ij}$.
The corresponding dual problem $(D_2)$ is to minimize the function $\mathcal L: \mathbb R^M_+ \to \mathbb R_+$, defined by:
\begin{equation*}
\begin{aligned}
(D_2) \quad  \min_{z\in \mathbb R^M_{+}}\ \  \mathcal L(z), \quad
\text{where} \quad  \mathcal L(z) \coloneqq \sum_{i \in N} \max_{x_i\in P(\rhoc_i)} \left\{\sum_{j\in M} (w_{ij}-z_j) x_{ij}\right\} + \sum_{j \in M} s_j z_j.
\end{aligned}
\end{equation*}

We begin our analysis by investigating the structural properties of the dual problem $(D_2)$.
\begin{theorem}
For the problem $(D_2)$, the following hold:
\label{convex-dual}
\begin{itemize}
\item[{\rm (i)}] The set of minimizers of $(D_2)$ forms a lattice and the unique minimal solution $z^*$ exists.
\item[{\rm(ii)}] The optimal value of $(D_2)$ coincides with that of the problem $(P_2)$.
\item[{\rm(iii)}] The unique minimal solution $z^*$ can be expressed as $z^* =\chi_{\scriptscriptstyle X^*} + \gamma  t$ for some $ t \in \R^M$.
\end{itemize}
\end{theorem}
The proof is given in Section~\ref{Lexico}. Theorem~\ref{convex-dual}~(i) follows from the $\mathrm{L}^{\natural}$-convexity of the dual objective $\cL$. The proofs of (ii) and (iii) rely on the complementarity of convex flows. 
Specifically, we characterize the optimal solution of a perturbed convex flow as the sum of a base solution and a perturbation term. Applying this to $(P_2)$ and $(D_2)$ yields the primal-dual relationship and the structural form of $z^*$.

\begin{algorithm}[thb]
\caption{\textsc{Computing $(X^*, d^*)$ for Line 2 of Algorithm~\ref{algo1}}}
\label{alg:direction}
\begin{algorithmic}[1]
\REQUIRE Current price vector $p$, values $q'_{ij}(p_j)$, and demand and exchange oracles (see Section~\ref{polymatroid_theory}).
\STATE Construct the symbolic weight vector $w \in (\R^2)^{N \times M}$ by $w_{ij} \coloneqq (1, -\log q'_{ij}(p_j))$.
\STATE Compute the unique minimal solution $z^* = (z^{*0}, z^{*1}) \in (\R^2)^M$ of $(D_2)$.
\STATE Recover $X^* \coloneqq \{j \in M \mid z^{*0}_j = 1\}$ and set $t \coloneqq z^{*1}$.
\STATE Construct the direction $d^* \in \R_+^M$ by setting $d^*_j \coloneqq \exp(t_j)$ for $j \in X^*$ and $d^*_j \coloneqq 0$ otherwise.
\STATE \textbf{return} $(X^*, d^*)$.
\end{algorithmic}
\end{algorithm}

Algorithm~\ref{alg:direction} implements the computation for Line~2 of Algorithm~\ref{algo1}, leveraging our key finding that the fractional component of the minimal dual solution $z^*$ encodes the desired direction $d^*$.
To compute $z^*$, we apply efficient algorithms for weighted polymatroid intersection (see \citet[Section 3]{F2005}).
This algorithm is \emph{fully combinatorial} in the sense that it uses only additions, subtractions, and comparisons of input values.
Although this algorithm is designed for $(P_2)$, the duality guarantees that it also yields the minimal solution of $(D_2)$ via the optimality criterion.
 
For the implementation, we adopt a \emph{symbolic} approach to handle the lexicographic order rigorously.
Instead of choosing an actual value of $\gamma$, we represent a value $a+\gamma b$ for $a,b \in \R$ as a pair $(a,b)$ and perform lexicographic arithmetic.
Specifically, for $a_1,a_2,b_1,b_2 \in \R$, $(a_1,b_1) \pm (a_2,b_2)$ results in $(a_1 \pm a_2, b_1 \pm b_2)$, and comparisons of $(a_1,b_1)$ and $(a_2,b_2)$ are performed lexicographically. Any fully combinatorial algorithm works in this model.
This procedure yields $z^*$ directly in decomposed form $(\chi_{\scriptscriptstyle X^*}, t)$, consistent with the characterization in Theorem~\ref{convex-dual}~(iii).
Consequently, the algorithm outputs the pair $(X^*, d^*)$ in strongly polynomial time, given oracle access. The following theorem establishes the correctness of  $d^*$.

\begin{theorem}
\label{d_weighted}
The direction $d^*$ returned by Algorithm~\ref{alg:direction} satisfies \LSC.
\end{theorem}
In the proof, we establish a local property of $\mathcal L$ around $z^*$ to show that $d^*$ satisfies Theorem~\ref{any_direction}~(c).
\begin{theorem}
\label{ineq_L}
Let $z^*$  be the minimal solution of $(D_2)$ and $d^*$ be the direction returned by Algorithm~\ref{alg:direction}.
For any non-empty $X \subseteq X^*$, 
there exists $\tau >0$ such that for all $\xi \in (0, \tau)$, it~holds 
\[
    \mathcal L(z^*-\gamma \xi\chi_{\scriptscriptstyle X} ) \leq \mathcal L(z^*) - \gamma \xi \left( s(X)-\rhoc(X; p+\varepsilon d^*) \right).
\]
\end{theorem}

\begin{proof}[Proof of Theorem \ref{d_weighted}.]
Fix a non-empty subset $X \subseteq X^*$, and let $\tau$
be the constant given by Theorem~\ref{ineq_L} for~$X$. 
For any $\xi \in (0, \tau)$, Theorem~\ref{convex-dual}~(iii) shows that $z^*-\gamma \xi\chi_{\scriptscriptstyle X}\geq 0$.
Since $z^*$ is the unique minimal solution, any feasible $z < z^*$ is not a minimizer, and thus $\mathcal L(z^*-\gamma \xi \chi_{\scriptscriptstyle X}) > \mathcal L(z^*)$. Combining this with Theorem~\ref{ineq_L}, we obtain $s(X) < \rhoc(X; p+\varepsilon d^*)$.
By Theorem~\ref{any_direction} [${\rm (c)}\Rightarrow{\rm (a)}$], we conclude that $d^*$ satisfies~\LSC.
\end{proof}

\subsection{Proof of Theorem \ref{ineq_L}}
We first establish a decomposition property for the maximization term in $\mathcal L(z)$. By Fact~\ref{greedy_structure}, this term decomposes into two parts whenever $z$ takes the form $\chi_{\scriptscriptstyle X^*} + \gamma t'$ for some $t' \in \R^M$.

\begin{lemma}
\label{difference_L2}
Let $z \coloneqq \chi_{\scriptscriptstyle X^*} + \gamma t'$ with $t'\in\R^M$.
For any non-empty $X \subseteq X^*$, $\xi \in \mathbb R_+$, and $i\in N$, it holds that 
\begin{align}
\label{eq:L}
&\max_{x_i\in P(\rhoc_i)} \sum_{j\in M} \Bigl(w_{ij}-\bigl(z_j-\gamma \xi(\chi_{\scriptscriptstyle X})_j \bigr)\Bigr) x_{ij} \notag \\
&\qquad = \max_{x'_i\in P(\rhoc_i)|_{M\setminus X^*}} \sum_{j\in M\setminus X^*} \left(w_{ij}-z_j \right) x'_{ij}
 + \max_{y_i\in P(\rhot_i)} \sum_{j\in X^*} \Bigl(w_{ij}-\bigl(z_j-\gamma \xi(\chi_{\scriptscriptstyle X} )_{j}\bigr)\Bigr) y_{ij}.
\end{align}
\end{lemma}

\begin{proof}
Since $z = \chi_{\scriptscriptstyle X^*} + \gamma t'$, on the left-hand side of \eqref{eq:L}, the primary component of the objective coefficients evaluates to $1$ for $j \in M \setminus X^*$ and $0$ for $j \in X^*$.
Then, the coefficients for $M \setminus X^*$ are lexicographically strictly larger than those for $X^*$.
By Fact~\ref{greedy_structure} (ii), the strict lexicographic order ensures that the goods in $M \setminus X^*$ are prioritized to saturate the projection $P(\rhoc_i)|_{M \setminus X^*}$, while the remaining capacity is allocated to $X^*$ according to the contraction $\rhot_i$.
Since $(\chi_{\scriptscriptstyle X})_j = 0$ for $j \in M\setminus X^*$, this decomposition directly yields \eqref{eq:L}.
\end{proof}

Now we evaluate $\mathcal L(z^*-\gamma\xi \chi_{\scriptscriptstyle X}) - \mathcal L(z^*)$ for any non-empty $X \subseteq X^*$. 
The following two lemmas link this difference to the decomposition in \eqref{decomposition_log}. 
Let $c^{0}_{ij} \coloneqq \log q'_{ij}(p_j) + t_j$ and $c^{\xi}_{ij} \coloneqq c^{0}_{ij} - \xi (\chi_{\scriptscriptstyle X})_j$ for simplicity.
\begin{lemma}
\label{difference_L}
For any non-empty subset $X \subseteq X^*$ and $\xi \in\mathbb R_+$, it holds that 
\begin{align*}
\mathcal L(z^*-\gamma\xi \chi_{\scriptscriptstyle X} )-\mathcal L(z^*)=
-\gamma \xi s(X)+\gamma \sum_{i\in N}\left(\max_{y_i\in P(\rhot_i)} \sum_{j\in X^*}(-c^{\xi}_{ij}) y_{ij}-\max_{y_i\in P(\rhot_i)} \sum_{j\in X^*} (-c^{0}_{ij}) y_{ij}\right).
\end{align*}
\end{lemma}
\begin{proof}
First, the left-hand side (LHS) of the equality is given by:
\begin{align}
\label{Lz_subtraction}
\text{(LHS)} =-\gamma\xi s(X)+\sum_{i\in N}\left\{\max_{x_i\in P(\rhoc_i)} \sum_{j\in M} \Bigl(w_{ij}-\bigl(z^*_j-\gamma\xi(\chi_{\scriptscriptstyle X})_j \bigr)\Bigr) x_{ij}-\max_{x_i\in P(\rhoc_i)} \sum_{j\in M} \Bigl(w_{ij}-z^*_j \Bigr) x_{ij} \right\}.
\end{align}
By Theorem \ref{convex-dual} (iii), the minimal solution $z^*$ for $(D_2)$ can be expressed as $z^* =\chi_{\scriptscriptstyle X^*} + \gamma  t$ for some $ t \in \R^M$.
Applying Lemma \ref{difference_L2}, the maximization term inside the summation can be decomposed~as
\begin{align}
\label{decomposition_Lz}
\max_{x_i\in P(\rhoc_i)} \sum_{j\in M} \Bigl(w_{ij}-\bigl(z^*_j-\gamma\xi(\chi_{\scriptscriptstyle X})_j \bigr)\Bigr) x_{ij}&=
\max_{x'_i\in P(\rhoc_i)|_{M\setminus X^*}} \sum_{j\in M\setminus X^*} (w_{ij}-z^*_j) x'_{ij}
+\max_{y_i\in P(\rhot_i)} \gamma\sum_{j'\in X^*} (-c^{\xi}_{ij'}) y_{ij'}
\end{align}
where the equality holds by  
$w_{ij'}-\bigl(z^*_{j'}-\gamma\xi(\chi_{\scriptscriptstyle X})_{j'}\bigr)=
-\gamma\log q'_{ij'}(p_{j'}) - \gamma t_{j'} + \gamma\xi (\chi_{\scriptscriptstyle X})_{j'}=-\gamma c^\xi_{ij'}$ for each $j'\in X^*$. 
Substituting \eqref{decomposition_Lz} with $\xi=0$ and an arbitrary $\xi$ into \eqref{Lz_subtraction},
we have 
\begin{align*}
\text{(LHS)} 
=-\gamma\xi s(X)+\gamma\sum_{i\in N}\left(\max_{y_i\in P(\rhot_i)} \sum_{j\in X^*}(-c^{\xi}_{ij}) y_{ij}-\max_{y_i\in P(\rhot_i)} \sum_{j\in X^*} (-c^{0}_{ij}) y_{ij}\right),
\end{align*}
where the first term on the right-hand side of \eqref{decomposition_Lz} is independent of $\xi$ and cancels out in the subtraction.
\end{proof}

From Lemma \ref{difference_L}, it suffices to establish the following inequality for any $i$ and non-empty $X \subseteq X^*$:
\begin{equation}
\label{goal}
\max_{y_i\in P(\rhot_i)} \sum_{j\in X^*}(-c^{\xi}_{ij}) y_{ij}-\max_{y_i\in P(\rhot_i)} \sum_{j\in X^*} (-c^{0}_{ij}) y_{ij} \leq \xi \rhoc_i(X; p+\varepsilon d^*).
\end{equation}

By the construction of $d^*$, we have
$\log q'_{ij}(p_j) + \log d^*_j=\log q'_{ij}(p_j) + t_j =c^0_{ij}$ for each $j\in X^*$.  
Then, the decomposition in \eqref{decomposition_log} can be expressed as
$\Dc_i(p+\varepsilon d^*) =\Dc_i(p)|_{M\setminus X^*}\oplus B_{d^*}(\rhot_i)$, 
where 
\[
B_{d^*}(\rhot_i)\coloneqq 
\argmin_{y_i \in B(\rhot_i)} \sum_{j \in X^*} (\log q'_{ij}(p_j) + \log d^*_j) y_{ij}
=\argmax_{y_i \in B(\rhot_i)} \sum_{j \in X^*}  (-c^0_{ij}) y_{ij}.
\]
This observation links the value $\rhoc_i(\cdot; p+\varepsilon d^*)$ to the linear optimization problem over $B(\rhot_i)$.
The following lemma guarantees a common optimal solution $y^*_i$ for both maximization terms in \eqref{goal}.
\begin{lemma}
\label{common_optimal}
For each $i\in N$, let $X^i_{\leq 0} \coloneqq \{ j \in X^* \mid c^{0}_{ij} \le 0 \}$. There exists $\tau_i >0$ such that 
every $\xi \in (0, \tau_i)$ admits a common optimal solution 
\begin{equation*}
y^*_i \in \argmax_{y_i \in P(\rhot_i)} \sum_{j \in X^*} (-c^{\xi}_{ij}) y_{ij}\ 
\cap\ \argmax_{y_i \in P(\rhot_i)} \sum_{j \in X^*}  (-c^{0}_{ij}) y_{ij}
\end{equation*}
satisfying $y^*_i|_{X^i_{\leq 0}}\in B_{d^*}(\rhot_i)|_{X^i_{\leq 0}}$ and $y^*_{ij}=0$ for all $j\in X^*\setminus X^i_{\leq 0}$.
\end{lemma}
\begin{proof}
Let $\mathcal{V} \coloneqq \{ c^{0}_{ij} \mid j \in M \} \cup \{ 0 \}$. By choosing $\tau_i > 0$ smaller than the minimum non-zero absolute difference between any two elements in $\mathcal{V}$, any perturbation $\xi \in (0, \tau_i)$ preserves the signs and ordering of the coefficients.\footnote{If all the elements of $\mathcal V$ coincide, no such non-zero difference exists; in this case we let $\tau_i$ be an arbitrary positive constant.} Specifically, for any $j,k\in M$, we have
$c^{\xi}_{ij} \leq 0 \implies c^{0}_{ij} \leq 0$ and $c^{\xi}_{ij} \leq c^{\xi}_{ik} \implies c^{0}_{ij} \leq c^{0}_{ik}$.

Since the structure in Fact \ref{greedy_structure} depends solely on the signs and ordering of the objective coefficients, the set of optimal solutions for $-c^{\xi}$ is a subset of that for $-c^{0}$. Therefore, any optimal solution $y^*_i$ for $-c^{\xi}$ is also optimal for $-c^{0}$. To show the structural requirement, we consider two cases according to the sign of~$-c^{\xi}_{ij}$:
\begin{itemize}
    \item For $j \in X^* \setminus X^i_{\leq 0}$, we have $-c^{\xi}_{ij} < 0$. By Fact \ref{greedy_structure} (ii), any optimal solution must satisfy $y^*_{ij} = 0$.
    \item For $j \in X^i_{\leq 0}$, we have $-c^{\xi}_{ij} \geq 0$. By Fact \ref{greedy_structure} (ii), the optimal set is a direct sum of base polyhedra for components with $-c^{\xi}_{ij} > 0$ and a polymatroid for those with $-c^{\xi}_{ij} = 0$. 
    Since a base polyhedron is a subset of its polymatroid, we can extract a solution that belongs to the base polyhedron even for components where $-c^{\xi}_{ij} = 0$. As the direct sum of these base polyhedra forms $B_{d^*}(\rhot_i)|_{X^i_{\leq 0}}$, we obtain $y^*_i|_{X^{i}_{\leq 0}}\in B_{d^*}(\rhot_i)|_{X^i_{\leq 0}}$.
\end{itemize}
\end{proof}

\begin{proof}[Proof of Theorem \ref{ineq_L}.]
For each $i \in N$, let $\tau_i$ be the constant in Lemma~\ref{common_optimal}. For any $\xi \in (0, \tau_i)$, let $y^*_i$ be the common optimal solution guaranteed by Lemma~\ref{common_optimal}.
Since $y^*_i$ maximizes both terms in the left-hand side (LHS) of (\ref{goal}), we can evaluate them at $y^*_i$. 
By $c^{\xi}_{ij} = c^{0}_{ij} - \xi$ for $j \in X$ and $c^{\xi}_{ij} =c^{0}_{ij}$ for $j\in X^*\setminus X$, we have:
\begin{align*}
\text{(LHS) of (\ref{goal})}
= \left( \sum_{j\in X^*\setminus X}  (-c^{0}_{ij}) y^*_{ij}+\sum_{j\in X} (-c^{0}_{ij} + \xi) y^*_{ij} \right) - \sum_{j \in X^*} (-c^{0}_{ij}) y^*_{ij}
= \sum_{j \in X} \xi y^*_{ij}
= \xi y^*_{i}(X\cap X^i_{\leq 0}),
\end{align*}
where the last equality follows from  $y^*_{ij}=0$ for any $j \in X^* \setminus X^i_{\leq 0}$.
By $y^*_i|_{X^i_{\leq 0}}\in B_{d^*}(\rhot_i)|_{X^i_{\leq 0}}$, we~have:
\begin{align*}
y^*_{i}(X \cap X^i_{\leq 0})
\leq \max_{y_i \in B_{d^*}(\rhot_i)|_{X^i_{\leq 0}}} y_i(X \cap X^i_{\leq 0})
\leq \max_{y_i \in B_{d^*}(\rhot_i)} y_i(X).
\end{align*}
Now we use the relation $\Dc_i(p+\varepsilon d^*)|_{X^*} = B_{d^*}(\rhot_i)$,  obtained from $\Dc_i(p+\varepsilon d^*) =\Dc_i(p)|_{M\setminus X^*}\oplus B_{d^*}(\rhot_i)$.
By the property of rank functions described in~(\ref{rank}), the inequality (\ref{goal}) follows from:
\[
\text{(LHS)\ of\ (\ref{goal})}  
\leq \xi \max_{y_i \in \Dc_i(p+\varepsilon d^*)|_{X^*} } y_i(X) 
= \xi \rhoc_i(X; p+\varepsilon d^*).
\]
Setting $\tau \coloneqq \min_{i \in N} \tau_i$ ensures that the inequality holds for all $i\in N$, completing the proof.
\end{proof}

\subsection{Economic Interpretation of the Direction}
\label{Interpretation_abst}
Here, we provide an interpretation of $d^*$ computed in Line 4 of Algorithm~\ref{alg:direction}.
Since $d^*$ is determined solely by the restriction $z^*|_{X^*}$, 
we consider a \emph{virtual market} restricted to the goods in $X^*$ with the supply vector $s|_{X^*}$, forming the feasible domain $[\zeros,s]_{\Z}^{X^*} \coloneqq \{ y \in \Z^{X^*}_+ \mid 0 \le y_j \le s_j \text{ for all } j \in X^* \}$.
On this domain, we define the virtual valuation $\tilde{v}_i: [\zeros,s]_{\Z}^{X^*} \to \R \cup \{-\infty\}$ by
\begin{equation*}
\tilde{v}_i(y_i) \coloneqq 
\begin{cases}
\sum_{j \in X^*} w_{ij} y_{ij} & \text{if } y_i \in P(\rhot_i), \\
-\infty & \text{otherwise}.
\end{cases}
\end{equation*}
Note that the effective domain $P(\rhot_i)$ is an $\mathrm{M}^\natural$-convex set (see Appendix~\ref{DemandSets} for the formal definition). Since $\tilde{v}_i$ coincides with a linear function over this domain, the following observation is immediate:
\begin{observation}
The virtual valuation $\tilde{v}_i$ is an $\mathrm{M}^\natural$-concave function on $[\zeros,s]_{\Z}^{X^*}$.
\end{observation}
Using this valuation, we can characterize $z^*|_{X^*}$ as the minimum equilibrium price vector of this virtual~market.
\begin{lemma}
\label{independence}
Let $z^*$ be the minimal solution of $(D_2)$. The restriction $z^*|_{X^*}$ is the unique minimal solution~to:
\begin{equation*}
\min_{z'\in \mathbb R^{X^*}_{+}}\left\{ \sum_{i \in N} \max_{y_i \in [\zeros,s]_{\Z}^{X^*}} \left( \tilde{v}_i(y_i) - \sum_{j \in X^*} z'_j y_{ij} \right) + \sum_{j \in X^*} s_j z'_j\right\} .
\end{equation*}
\end{lemma}
The objective function is the potential function, 
called the \emph{Lyapunov function}, for the market with valuations $\{\tilde{v}_i\}_{i \in N}$ and supply $s|_{X^*}$.
This implies that $z^*|_{X^*}$ corresponds to the minimum equilibrium price vector in this virtual market.
This reveals an intriguing feature of our framework: to compute the equilibrium with frictions, in every iteration we compute the minimum equilibrium price vector of an auxiliary market without frictions. The existence of this vector, in turn, guarantees the existence of a direction satisfying \LSC.

Since each $\tilde{v}_i$ is additive within its effective domain, the above problem can be viewed as an assignment problem for a virtual market with unit-demand buyers by replacing each buyer with copies corresponding to their allocated units (for the formal construction of the problem, see Appendix~\ref{Interpretation_detail}).
By the classical result of \citet{L1983}, the minimum equilibrium prices in such a market coincide with the VCG prices \citep{V1961, C1971, G1973}. 
Therefore, the direction $d^*$ is constructed from the VCG prices in this auxiliary 
assignment market.

\section{Finite Convergence (Proof of Theorem \ref{finite})}
\label{Convergence}
In this section, we show that events (a) and (b) in Line~4 occur only finitely many times. Theorem \ref{finite} is an immediate consequence of the following two propositions.
In the first one, which extends the analysis of \citet[Lemma 7]{GS2000}, we bound the number of occurrences of event (a). While Algorithm~\ref{algo1} employs directional price updates, \LSC\, provides a bound on the number of changes in the minimal maximally over-demanded set.

\begin{proposition}
\label{event_a}
Event (a) in Line 4 of Algorithm~\ref{algo1} occurs at most $O(s(M) mn)$ times.
\end{proposition}
Throughout this section, we repeatedly use the following lemma:
\begin{lemma}
\label{monotone_c}
Let $p$ and $p'$ be two price vectors in the same iteration of Algorithm~\ref{algo1} such that $p \lneq p'$, and let $X^*$ be the minimal maximally over-demanded set in that iteration.
Then, it holds $\mathcal{O}(X^*; p) \geq \mathcal{O}(X^*; p')$.
\end{lemma}
\begin{proof}
By $\operatorname{supp}_+(d^*) = X^*$, 
we have $p_j = p'_j$ for all $j \in M \setminus X^*$.
Applying Proposition~\ref{mu_rho_relation} with $Y = M \setminus X^*$, we have  
$\mu_i(X^*; p) \geq \mu_i(X^*; p')$ for all $i\in N$.
Summing this over all buyers yields 
$\mathcal{O}(X^*; p)\geq~\mathcal{O}(X^*; p')$.
\end{proof}

\begin{proof}[Proof of Proposition \ref{event_a}.]
Let $p^{(t)}$ denote the price vector at the $t$-th occurrence of event (a), and let $X^*_t$ denote the minimal maximally over-demanded set  at $p^{(t)}$.
Define an integer sequence $\{ \theta_t\}$ by:
$\theta_t \coloneqq (m+1) \mathcal{O}(X^*_{t}; p^{(t)})-|X^*_{t}|$.
By Lemma~\ref{monotone_c}, the value $\mathcal{O}(X^*_t; p)$ is non-increasing between two consecutive occurrences of event~(a).
In the following, we show that the sequence strictly decreases, i.e., $ \theta_{t} >  \theta_{t+1}$ for each $t$, and that the sequence is bounded from below.
To this end, we establish the following three properties:
\begin{itemize}
    \item[(i)] $\mathcal{O}(X^*_{t}; p^{(t)}) \geq \mathcal{O}(X^*_{t+1}; p^{(t+1)})$.
    \item[(ii)] If $\mathcal{O}(X^*_{t}; p^{(t)})=\mathcal{O}(X^*_{t+1}; p^{(t+1)})$, then $|X^*_{t}| < |X^*_{t+1}|$.
    \item[(iii)] $ \theta_t\geq 0$ and $ \theta_t=0$ if and only if $X^*_{t} = \emptyset$.
\end{itemize}

Let $d$ be the direction immediately preceding the arrival at $p^{(t+1)}$.
Consider $p' \coloneqq p^{(t+1)} - \varepsilon d$, where $\varepsilon > 0$ is sufficiently small so that event (b) does not occur and the demand sets of buyers remain invariant at all prices strictly between $p'$ and $p^{(t+1)}$.
By the upper semicontinuity of the demand correspondence, $D_i(p') \subseteq D_i(p^{(t+1)})$ holds for each $i \in N$. 
Then, property (i) holds by 
\begin{align*}
\mathcal{O}(X^*_{t}; p^{(t)}) 
&=  \mu(X^*_{t}; p^{(t)})-s(X^*_{t}) && \text{(By (\ref{overdemand}))} \\
&\geq \mu(X^*_{t}; p')-s(X^*_{t})  && \text{(By Lemma \ref{monotone_c})} \\
&\geq  \mu(X^*_{t+1}; p')-s(X^*_{t+1}) && (\text{By the choice of } X^*_{t}\ \text{and \LSC}) \\ 
&\geq  \mu(X^*_{t+1}; p^{(t+1)})-s(X^*_{t+1}) && (\text{By}\ D_i(p') \subseteq D_i(p^{(t+1)})\ \text{for each}\ i) \\ 
&= \mathcal{O}(X^*_{t+1}; p^{(t+1)}) &&  \text{(By (\ref{overdemand}))}.
\end{align*}

Suppose that $\mathcal{O}(X^*_{t}; p^{(t)}) = \mathcal{O}(X^*_{t+1}; p^{(t+1)})$. 
Specifically, $\mathcal{O}(X^*_{t}; p') = \mathcal{O}(X^*_{t+1}; p')$ means that $X^*_{t+1}$ is also a maximizer of $\mathcal{O}(\cdot; p')$.
Since $\mathcal{O}(\cdot; p')$ is supermodular, the set of its maximizers forms a lattice and is thus closed under intersection. The \textit{minimal} maximally over-demanded set $X^*_{t}$ is unique, and since $X^*_{t} \neq X^*_{t+1}$, the strict inclusion $X^*_{t} \subsetneq X^*_{t+1}$ holds. 
Therefore, we have $|X^*_{t}| < |X^*_{t+1}|$, which yields (ii).

Property (iii) holds immediately by the choice of $X^*_{t}$. Specifically, if $X^*_t\neq \emptyset$, then $\mathcal{O}(X^*_{t}; p)\geq 1$ for any $p$ during the execution of Algorithm~\ref{algo1}, and thus $ \theta_t>0$ by $|X^*_t| \leq m$.

Consequently, $\{ \theta_t\}$ is a strictly decreasing sequence of integers starting from $ \theta_0$.
By (iii), the algorithm continues as long as $ \theta_t > 0$ and terminates when $ \theta_T = 0$ (i.e., $X^*_T = \emptyset$). Since the sequence decreases by at least $1$ in each iteration, the total number of occurrences of event~(a) is bounded by $\theta_0 \leq s(M)(m+1)n$. 
\end{proof}

We now address event (b), occurring when \LSC\, is violated while $X^*$ 
remains invariant. The second proposition characterizes the triggers of such changes and establishes their finiteness.
\begin{proposition}
\label{event_b}
Event (b) in Line 4 of Algorithm~\ref{algo1} is caused by one of the following:
\begin{enumerate}
    \item[{\rm (i)}] The right derivatives of the payment functions for some buyers change.
    \item[{\rm (ii)}] The value of the over-demand function for $X^*$ strictly decreases.
    \item[{\rm (iii)}] The optimal value of the dual problem $(D_2)$ strictly increases.
    \item[{\rm (iv)}] The unique minimal solution of $(D_2)$ strictly increases while the optimal value remains invariant.
\end{enumerate}
Moreover, event (b) triggers only a finite number of updates before the next occurrence of event (a).
\end{proposition}
\begin{proof}
Consider Case (i), which affects the weights $w$ and thus the optimal solution to $(D_2)$. Since the payment functions are piecewise-linear, such changes occur only finitely many times.

Hereafter, we assume that the slopes of the payment functions remain unchanged. 
Then, a violation of \LSC\, implies that the solution to $(D_2)$, computed at the beginning of the iteration and used to construct $d^*$, is no longer the minimal optimal solution.
Since $w$ is unchanged, this violation must be caused by a change in the polymatroids 
$\{P(\rhot_i(\cdot; p))\}$ at the current price $p$. 
We now focus on when such a change occurs.

We first address Case (ii). Since $X^*$ remains the minimal maximally over-demanded set, this implies that $\mu_i(X^*; p)$ decreases for some $i$.
By \eqref{mu_rho}, this corresponds to either a decrease in $\rhoc_i(M; p)$ or an increase in $\rhoc_i(M\setminus X^*; p)$.
In either scenario, $\rhot_i(X^*; p)\coloneqq \rhoc_i(M; p)-\rhoc_i(M\setminus X^*; p)$ is updated, thereby altering $P(\rhot_i(\cdot;p))$. 
By Lemma~\ref{independence}, such a change may shift the minimal optimal solution to $(D_2)$. Since each such change strictly decreases $(m+1)\mathcal{O}(X^*; p) - |X^*|$, and this quantity can decrease only a finite number of times, as established in Proposition~\ref{event_a} and Lemma~\ref{monotone_c}, the change occurs only finitely many times.

Suppose that each $\mu_i(X^*;p)$ is unchanged. 
Let $\tilde{p}$ denote the price vector immediately preceding $p$. 
By Proposition~\ref{mu_rho_relation} (ii), this implies $\rhoc_i(M; p)=\rhoc_i(M; \tilde{p})$ and $\rhoc_i(M\setminus X^*; p)=\rhoc_i(M\setminus X^*; \tilde{p})$.
By $\rhoc_i(M; p)=\rhoc_i(M; \tilde{p})$ and the upper semicontinuity of the demand correspondence, the inclusion $\Dc_i(\tilde{p}) \subseteq \Dc_i(p)$ holds. 
This implies that $\rhoc_i$ is non-decreasing from $\tilde{p}$ to $p$. 
Combined with $\rhoc_i(M\setminus X^*; p)=\rhoc_i(M\setminus X^*; \tilde{p})$,
this gives, for any $Y \supseteq M\setminus X^*$, 
$\rhot_i(Y; p) = \rhoc_i(Y; p)- \rhoc_i(M\setminus X^*; p) \geq 
\rhoc_i(Y; \tilde{p})- \rhoc_i(M\setminus X^*; \tilde{p}) = \rhot_i(Y; \tilde{p})$,
which implies $P(\rhot_i(\cdot; \tilde{p})) \subseteq P(\rhot_i(\cdot; p))$. 
This means that the optimal value of $(D_2)$ does not decrease from  $\tilde{p}$ to $p$. If \LSC\, is violated, this inclusion must be strict for some $i$, i.e., $P(\rhot_i(\cdot; \tilde{p})) \subsetneq P(\rhot_i(\cdot; p))$.

Now consider Case (iii). Since each $\rhot_i(\cdot; p)$ is determined by $D_i(p) \subseteq [\zeros,s]_{\Z}$, the rank functions $\{\rhot_i\}_{i \in N}$ take only finitely many distinct forms. Hence the set of possible optimal values of $(D_2)$ is finite, so this value can strictly increase only finitely many times.

Finally, we consider Case (iv). 
Let $\mathcal L_p$ and $\mathcal L_{\tilde{p}}$ denote the objective functions of $(D_2)$ at $p$ and $\tilde{p}$, respectively. 
By $P(\rhot_i(\cdot; \tilde{p})) \subseteq P(\rhot_i(\cdot; p))$, we have $\mathcal L_{\tilde{p}}(z) \leq \mathcal L_{p}(z)$ for all $z \in \mathbb R^M_+$.
Consider any $z'\in\mathbb R^M_+$ with $z'\lneq z^*$. 
Since \LSC\, is satisfied at $\tilde{p}$ and the optimal value is unchanged, the minimality of $z^*$ yields
$\mathcal L_{p}(z') \geq \mathcal L_{\tilde{p}}(z') > \mathcal L_{\tilde{p}}(z^*) 
= \min_{z}\mathcal L_{\tilde{p}}(z) = \min_{z}\mathcal L_{p}(z)$.
Let $z^{**}$ be the unique minimal solution for $\min_{z}\mathcal L_{p}(z)$.
The inequality $\mathcal L_{\tilde{p}}(z^{**}) \le \mathcal L_{p}(z^{**}) =\min_{z} \mathcal L_{p}(z) = \min_{z} \mathcal L_{\tilde{p}}(z)$ implies that $z^{**}$ also minimizes $\mathcal L_{\tilde{p}}$. Since $z^*$ is the unique minimal minimizer of $\mathcal L_{\tilde{p}}$ and the minimal solution strictly changes (due to the violation of \LSC), we obtain  $z^{**} \gneq z^*$. 
Since the minimal solution strictly increases over a finite set of candidates, the same solution never recurs, so Case~(iv) can recur only finitely many times before Case~(iii).

Since each case occurs only finitely many times, event~(b) triggers only finitely many updates. 
\end{proof}

Unlike frictionless settings where uniform price increases preserve buyers' demands within $X^*$, our directional updates can shift demands within $X^*$, triggering Cases~(iii) and~(iv); see Appendix~\ref{Example} for a concrete example.\footnote{While \citet{DHW2015} and \citet{BEF2024} claim polynomial bounds, their proofs appear to overlook an inherent internal demand shift: a new edge between a buyer and a good, both belonging to the maximal alternating tree, may enter the demand graph. These situations correspond precisely to Cases~(iii) and~(iv). Therefore, our finite termination guarantee is not inferior to theirs.} Nevertheless, our construction of $d^*$ via Algorithm~\ref{alg:direction} guarantees finite termination.

Finally, we observe that under Assumption \ref{simple-case}, the dynamics are significantly simplified.
\begin{observation}
\label{simple_iteration}
Under Assumption~\ref{simple-case}, event (b) in Line 4 is triggered only by Case (i).
\end{observation}
This follows from Proposition~\ref{simple}: when the slopes of the payment functions are unchanged, the direction $d^*$ depends only on $X^*$. Our mechanism therefore inherits the iteration bounds of \citet{GS2000}.

\section{Potential Functions and Generalized Discrete Convexity}
\label{potential}
\citet{A2006} provided a potential-based interpretation of the auction dynamics in the setting without frictions. 
In our framework, this corresponds to the special case of Assumption~\ref{simple-case} where $\alpha_i = 1$ and $\beta_j = 1$ for all $i \in N$ and $j \in M$.
The potential function $L:\mathbb R^M_+\to \mathbb R_+$, called the \textit{Lyapunov} function, is defined by
\begin{equation}
\label{lyapunov}
L(p) \coloneqq \sum_{i \in N} V_i(p) + \sum_{j \in M} p_j s_j,
\quad \text{where } V_i(p) \coloneqq \max_{x_i \in [\zeros,s]_{\mathbb Z}} \left\{ v_i(x_i) - \sum_{j \in M} p_j x_{ij} \right\}.
\end{equation}
The dynamics of the auction are captured by the directional derivative: for any $p \in \R_+^M$ and $X \subseteq M$, 
\begin{equation}
\label{derivative}
L'(p;X)\coloneqq
\lim_{\varepsilon \to +0}
\frac{L(p + \varepsilon \chi_{\scriptscriptstyle X}) - L(p)}{\varepsilon}
=-\mathcal O (X; p).
\end{equation}
This shows that the minimal maximally over-demanded set determines the direction of steepest descent for $L$ at $p$. \citet{MSY2016} revealed that this function belongs to the class of polyhedral ${\mathrm L}^{\natural}$-convex functions, interpreting the auction as its steepest descent algorithm.
The following result, originally derived by \citet{A2006}, can be understood as a consequence of 
${\mathrm L}^{\natural}$-convexity, where a local minimum is a global~minimum.
\begin{theorem}[\citet{A2006}]
\label{minimizer}
Consider the case of Assumption~\ref{simple-case} where $\alpha_i = 1$ and $\beta_j = 1$ for all $i \in N$, $j \in M$.
Let $p'$ be  a lower bound on some minimizer of $L$.
Then, the following are equivalent:
(i) $p'$ is an equilibrium price vector, (ii) $p'\in \arg\min L(p)$, and (iii) $L'(p';X)\geq 0$ for all $X\subseteq M$.
\end{theorem}

When no set is over-demanded (i.e., $\mathcal{O}(X; p) \leq 0$ for all $X\subseteq M$), equation \eqref{derivative} implies that condition (iii) is satisfied, and thus $L$ attains its minimum.
Since the set of minimizers of an ${\mathrm L}^{\natural}$-convex function forms a lattice, 
the set of equilibrium prices inherits this structure.

Extending this potential-based framework to our general model is challenging for two reasons. First, our algorithm employs directional price updates that depend on the current price, whereas in the standard framework the direction is a $0$--$1$ vector regardless of prices. Second, and more fundamentally, the set of equilibrium prices need not be convex under general frictions, so it cannot coincide with the minimizer set of a convex function. We therefore seek not to recover all equilibrium prices, but rather to construct a generalized potential function, also denoted by $L$, that is minimized at the \emph{minimum} equilibrium price $\underline{p}$.
\begin{problem}
\label{op}
Let $\underline{p}$ be the minimum equilibrium price. Determine whether there exists a polyhedral convex function $L: \R^M_+ \to \R_{+}$ and a direction map $d: \R^M_+ \to \R^M_{++}$ satisfying the following for any $p\leq \underline{p}$:
\begin{itemize}
\item[(i)] The subdifferential of $L$ at $p$, denoted by $\partial L(p)$, satisfies
\begin{equation}
\label{g_subgradient}
\partial L(p) = d^{-1}_p \circ \Bigl(s - \sum_{i \in N} D_i(p) \Bigr),
\end{equation}
where $d^{-1}_p$ denotes the element-wise inverse of $d_p\coloneqq d(p)$, i.e.,  $(d^{-1}_p)_j = 1/(d_p)_j$, $\circ$ denotes the Hadamard product, and $\sum_{i \in N} D_i(p)$ denotes the Minkowski sum of the demand sets.
\item[(ii)] The local optimality condition is equivalent to the non-negativity of directional derivatives:
\[
0 \in \partial L(p)\Leftrightarrow L'(p;X)\coloneqq \lim_{\varepsilon \to +0}\frac{L(p + \varepsilon (d_p\circ \chi_{\scriptscriptstyle X})) - L(p)}{\varepsilon}\geq 0\ \ {\rm for\ all}\  X\subseteq M.
\]
\end{itemize}
\end{problem}

If such a function $L$ exists, our auction is interpreted as its steepest descent algorithm. To see this, we first note that the directional derivative $L'(p; X)$ is characterized by the over-demand function. Since $L$ is convex, this derivative coincides with the support function of its subdifferential:
\begin{align*}
L'(p; X)&=\max \bigl\{ \sum_{j\in M} (d_p\circ \chi_{\scriptscriptstyle X})_j y_{j} \mid y \in \partial L(p) \bigr\} 
=\max \bigl\{ x(X) \mid x \in \bigl(s - \sum_{i \in N} D_i(p)\bigr) \bigr\} \\
&=\max \bigl\{ s(X) - \sum_{i\in N} x_i(X)\mid x_i \in D_i(p)\ (\forall i \in N)\bigr\}=s(X) - \sum_{i\in N} \mu_i(X; p) = -\mathcal O(X; p).
\end{align*}
Second, Theorem \ref{minimizer}, linking equilibrium prices and minimizers of $L$, naturally extends to our model:

\begin{theorem}
Suppose that a polyhedral convex function $L$ and a direction map $d$ satisfying the conditions in Problem~\ref{op} exist. Then, for any $p \le \underline{p}$, the following are equivalent: (i) $p$ is an equilibrium price vector, (ii) $\zeros \in \partial L(p)$, and (iii) $L'(p;X) \ge 0$ for all $X \subseteq M$.
\end{theorem}

\begin{proof}
Condition (iii) is equivalent to $\zeros \in \partial L(p)$ by Problem \ref{op} (ii). Thus, it suffices to show that (i) and (ii) are equivalent to $\bm{0} \in \partial L(p)$. By convexity of $L$, we have
\begin{align*}
p \in \operatorname*{arg\,min} L(p)
\Leftrightarrow \bm{0} \in \partial L(p)
\Leftrightarrow s \in \sum_{i\in N} D_i(p) 
\Leftrightarrow p \text{ is an equilibrium price},
\end{align*}
where the second equivalence follows from Property (i) in Problem \ref{op}.
\end{proof}

Under Assumption \ref{simple-case}, where $q_{ij}'(p_j) = \alpha_i \beta_j$, we have successfully identified such a global potential function. 
We call this the \emph{scaled Lyapunov function}, defined by
\begin{equation}
\label{scaled-def}
    \tilde{L}(p) \coloneqq \sum_{i \in N} \tilde{V}_i(p) + \sum_{j \in M}  \beta_j p_j s_j ,\ \  \text{where } \tilde{V}_i(p) \coloneqq \max_{x_i \in [\zeros,s]_{\mathbb Z}} \left\{ \frac{v_i(x_i)}{\alpha_i} - \sum_{j \in M} \beta_j p_j x_{ij} \right\}.
\end{equation}

\begin{theorem}
\label{scaled}
Under Assumption \ref{simple-case}, the scaled Lyapunov function defined in \eqref{scaled-def} satisfies all the requirements in Problem~\ref{op} with the direction map $d$ given by $(d(p))_j \coloneqq 1/\beta_j$ for each $j \in M$.
\end{theorem}
In the proof, we use the following relation. 
\begin{observation}
\label{obs}
Consider a situation where the valuation function of each buyer $i \in N$ is modified to $v_i(\cdot)/\alpha_i$. Let $V_{i,\alpha_i}$ and $L_{\alpha}$ denote the (standard) indirect utility function and Lyapunov function under these valuations, respectively. Then, the scaled functions satisfy the following relations:
\begin{align*}
\tilde{V}_i(p) = V_{i,\alpha_i}(\beta \circ p)\quad (i \in N)
\quad \text{and} \quad
\tilde{L}(p) = L_{\alpha}(\beta \circ p),
\end{align*}
where $\beta \circ p$ denotes the Hadamard product of $\beta$ and $p$, i.e., $(\beta \circ p)_j \coloneqq \beta_j p_j$ for each~$j \in M$.
\end{observation}

\begin{proof}[Proof of Theorem \ref{scaled}.]
(i) By Danskin's theorem \citep{D1966} (see also \citet[Proposition B.22]{B2016}), the subdifferential of the indirect utility function is the convex hull of the gradients of the objective function evaluated at the maximizers. Letting conv denote the convex hull, we have:
\begin{align*}
\partial \tilde{V}_i(p) 
= \operatorname{conv} \left\{ - x \circ \beta \mid x \in \operatorname*{arg\,max}_{x \in [\zeros,s]_{\mathbb Z}} \left( \frac{v_i(x)}{\alpha_i} - \sum_{j \in M} \beta_j p_j x_{ij} \right) \right\} 
= \operatorname{conv} \{ - x \circ \beta \mid x \in D_i(p) \} = -\beta\circ D_i(p).
\end{align*}
Since the subdifferential is additive for sums of convex functions, we have
\[
\partial \tilde{L}(p) = \sum_{i \in N} \partial \tilde{V}_i(p) +s \circ \beta=- \sum_{i \in N} \beta \circ D_i(p)+ s \circ \beta =\beta \circ \Bigl(s - \sum_{i \in N} D_i(p) \Bigr).
\]
(ii) 
The directional derivative of $\tilde{L}$
along $\beta^{-1}\circ \chi_{\scriptscriptstyle X}$,
where $\beta^{-1}$ is the element-wise inverse of $\beta$, satisfies:
\begin{align}
\label{directional}
\tilde{L}'(p;X)
=\lim_{\varepsilon \to +0}\frac{\tilde{L}(p + \varepsilon (\beta^{-1}\circ \chi_{\scriptscriptstyle X})) - \tilde{L}(p)}{\varepsilon}
=\lim_{\varepsilon \to +0}\frac{L_{\alpha}(\beta\circ p + \varepsilon \chi_{\scriptscriptstyle X}) - L_{\alpha}(\beta\circ p)}{\varepsilon}=L_{\alpha}'(\beta \circ p; X),
\end{align}
where the second equality holds by Observation~\ref{obs}. 
Again by Observation~\ref{obs}, we have
$\zeros \in \partial \tilde{L}(p)
\Leftrightarrow \zeros \in \partial L_{\alpha}(\beta \circ p)
\Leftrightarrow \tilde{L}'(p;X) \geq 0$ for all $X\subseteq M$, 
where the last equivalence holds by Theorem~\ref{minimizer} and \eqref{directional}.

\end{proof}
\begin{remark}
As is evident from the proof, $\tilde{L}$ satisfies conditions~(i) and~(ii) for \emph{all} $p \in \R^M_+$. Hence $\arg\min \tilde{L}$ coincides with the entire set of equilibrium prices, which is therefore convex under Assumption~\ref{simple-case}.
\end{remark}

Now we introduce a new class of polyhedral convex functions, called 
\emph{scaled} polyhedral $\mathrm{L}^{\natural}$-convex functions.
Let $p \vee p'$ and $p \wedge p'$ denote the element-wise maximum and minimum of $p, p' \in \R_{+}^M$, respectively.
\begin{definition}
\label{def_scaled}
Let $g : \R^M_+ \to \R \cup \{+\infty\}$ be a function with $\operatorname{dom} 
g\coloneqq \{p\mid g(p)<\infty\} \neq \emptyset$, whose epigraph $\{(p, K) \mid p \in \R^M_+,\, K \in \R,\, g(p) \leq K\}$ is a polyhedron.
Then, $g$ is a \emph{scaled polyhedral L$^{\natural}$-convex} function if it satisfies translation-submodularity 
with respect to some $\tau\in \mathbb R^M_{++}$:
\[
g(p) + g(p') \geq g((p - \lambda \tau) \vee p') + g(p \wedge (p' + \lambda \tau))
\quad (\forall p, p' \in \operatorname{dom} g, \forall \lambda \in \R_+).
\]
\end{definition}
When $\tau$ is the all-ones vector, Definition~\ref{def_scaled} reduces to that of standard polyhedral $\mathrm{L}^{\natural}$-convex functions; hence our scaled class strictly includes the standard one. Although the scaled Lyapunov function $\tilde{L}$ lies outside the standard class, it belongs to this new class for $\tau = \beta^{-1}$:
\begin{lemma}
\label{L_scaled}
The scaled Lyapunov function $\tilde{L}$ is scaled polyhedral $\mathrm{L}^{\natural}$-convex with respect to $\beta^{-1}$.
\end{lemma}

The converse inclusion fails: $\tilde{L}$ is not necessarily translation-submodular with respect to the all-ones vector, so it is not polyhedral \Ln-convex in the standard sense. The following example illustrates~this.
\begin{example}
\label{Failed_tsubmodularity}
Consider a market with two buyers $N=\{1, 2\}$ and two types of goods $M=\{1, 2\}$,
each with a unit supply $s=(1, 1)$. We set the scaling parameters to $\alpha = (1, 1)$ and $\beta = (1, 2)$.
Suppose that Buyer 1 has an additive valuation $v_1(x) = x_1 + x_2$, and Buyer 2 has a unit-demand valuation $v_2(x) = \max(x_1, x_2)$ for $x \in [0,s]_{\mathbb{Z}}$.
Then, the scaled Lyapunov function is directly given by:
\begin{align*}
\tilde{L}(p) &=\tilde{V}_1(p)+\tilde{V}_2(p)+\sum_{j\in M}s_j \beta_j p_j \\
&=\max(0,1-p_1)+ \max(0,1-2p_2) + \max(0, 1-p_1, 1-2p_2) + p_1 + 2p_2.
\end{align*}
As shown in Lemma~\ref{L_scaled}, $\tilde{L}$ satisfies translation-submodularity with respect to $\beta^{-1} = (1, 0.5)$, meaning that it is a scaled polyhedral \Ln-convex function.
To see that it is not a standard polyhedral \Ln-convex function, we show that it violates translation submodularity with respect to $\mathbf{1} \coloneqq (1,1)$.
Let $\lambda = 0.1$, $p = (0.5, 0.2)$, and $p' = (0, 0.1)$, so that $(p - \lambda \mathbf{1}) \vee p' = (0.4, 0.1)$ and $p \wedge (p' + \lambda \mathbf{1}) = (0.1, 0.2)$.
Since $p_1 \leq 1$ and $p_2 \leq 0.5$ hold at all four of these points, $\tilde{L}$ simplifies there to $(1-p_1)+(1-2p_2) + \max(0, 1-p_1, 1-2p_2) + p_1 + 2p_2= 2 + \max(1-p_1, 1-2p_2)$. We then have
\begin{align*}
\tilde{L}(p) + \tilde{L}(p') =  5.6
<5.7=\tilde{L}((p - \lambda \mathbf{1}) \vee p') + \tilde{L}(p \wedge (p' + \lambda \mathbf{1})).
\end{align*}
\end{example}

While we constructed a potential function for the separable case, a general solution to Problem~\ref{op} remains a major open question. Solving it would substantially broaden the scope of market design to accommodate more general frictions. Given that scaled polyhedral ${\mathrm L}^{\natural}$-convexity already generalizes established concepts in discrete convexity, such a global potential function would also expand the theory of discrete optimization.

\section{Lexicographic Optimization for Convex Flow}
\label{Lexico}
We first provide the mathematical preliminaries on convex flows and lexicographic optimization, and then prove Theorem~\ref{convex-dual}, which plays a crucial role in the proof of Theorem~\ref{d_weighted}. 
For simplicity, we often denote the inner product of vectors $a$ and $b$ by $\langle a, b \rangle$.

\subsection{Convex Flow}
Let $D = (V, A)$ be a directed graph, $c \in \R^A$ an arc cost vector, and $\mathcal P \subseteq \R^V$ a polyhedron.
A vector $f \in \R^A$ is called a \emph{flow} on $D$.
We consider the following linear program
\begin{align}\label{prob:flow-primal}
    \begin{array}{ll}
        \displaystyle\minimize_{f \in \R^A} & \angle{c, f} \\
        \text{subject to} & \partial f \in \mathcal P,
    \end{array}
\end{align}
where ${(\partial f)}_v \coloneqq f(\delta^+ v) - f(\delta^- v) \quad (v \in V)$.
The problem~\eqref{prob:flow-primal} is a special case of the convex flow problem.
The Legendre--Fenchel dual of~\eqref{prob:flow-primal}~is
\begin{align}\label{prob:flow-dual}
    \begin{array}{ll}
        \displaystyle\maximize_{\pi \in \R^V} & \displaystyle \min_{y \in \mathcal P} \angle{\pi, y} \\
        \text{subject to} & \delta \pi = c,
    \end{array}
\end{align}
where ${(\delta \pi)}_a \coloneqq \pi_u - \pi_v \quad (a = (u,v) \in A)$.
A vector $\pi \in \R^V$ is called a \emph{potential} on $D$.

\begin{theorem}[{\citet[Section~9]{M2003}}]
\label{thm:flow}
    Suppose that~\eqref{prob:flow-primal} and~\eqref{prob:flow-dual} have optimal solutions.
    \begin{enumerate}
        \item[{\rm (i)}] Problems~\eqref{prob:flow-primal} and~\eqref{prob:flow-dual} have the same optimal values.
        \item[{\rm (ii)}] A feasible flow $f$ and a feasible potential $\pi$ are both optimal if and only if $\partial f \in \argmin_{y \in \mathcal P} \angle{\pi, y}$.\label{item:complementarity}
    \end{enumerate}
\end{theorem}

\subsection{Lexicographic Optimization}
We consider the following lexicographic optimization:
\begin{align}
    &\begin{array}{ll}
        \displaystyle\minimize_{f \in \R^A} & \angle{c^0 + \varepsilon c^1, f} \\
        \text{subject to} & \partial f \in \mathcal P,
    \end{array}
    &&
    \begin{array}{ll}
        \displaystyle\maximize_{\pi \in \R^V} & \displaystyle\min_{y \in \mathcal P} \angle{\pi, y} \\
        \text{subject to} & \delta \pi = c^0 + \varepsilon c^1,
    \end{array}\label{prob:flow-lex-primal}
    \end{align}
    where $c^0, c^1 \in \R^A$ and $\varepsilon > 0$ is sufficiently small.
    Theorem~\ref{thm:flow}~(ii) implies that, for any fixed optimal potential $\pi^0$ of~\eqref{prob:flow-lex-primal} with $\varepsilon = 0$, a flow is optimal for~\eqref{prob:flow-lex-primal} if and only if it is optimal for 
    \begin{align}
    &\begin{array}{ll}
        \displaystyle\minimize_{f \in \R^A} & \angle{c^1, f} \\
        \text{subject to} & \displaystyle \partial f \in \mathcal P^0,
    \end{array}
    &&
    \begin{array}{ll}
        \displaystyle\maximize_{\pi \in \R^V} & \displaystyle\min_{y \in \mathcal P^0} \angle{\pi, y} \\
        \text{subject to} & \delta \pi = c^1,
    \end{array}\label{prob:flow-lex-primal2}
\end{align}
where $\mathcal P^0 \coloneqq \argmin_{y \in \mathcal P} \angle{\pi^0, y}$.

\begin{theorem}
\label{lexico_characterization}
    Let $\pi^0$ be an optimal potential of~\eqref{prob:flow-lex-primal} with $\varepsilon = 0$ and $\pi^1$ an optimal potential of~\eqref{prob:flow-lex-primal2}.
    Then, $\pi^0 + \varepsilon \pi^1$ is an optimal potential of~\eqref{prob:flow-lex-primal}.
    Conversely, any optimal potential of~\eqref{prob:flow-lex-primal} is written in this~form.
\end{theorem}
\begin{proof}
    Let $\pi^* \coloneqq \pi^0 + \varepsilon \pi^1$.
    By the feasibility of $\pi^0$ for~\eqref{prob:flow-lex-primal} with $\varepsilon = 0$ and $\pi^1$ for~\eqref{prob:flow-lex-primal2}, we have $\delta \pi^* = \delta \pi^0 + \varepsilon \delta \pi^1 = c^0 + \varepsilon c^1$.
    Thus, $\pi^*$ is feasible for~\eqref{prob:flow-lex-primal}.
    Next, to show the optimality of $\pi^*$, let $f^*$ be an optimal flow for~\eqref{prob:flow-lex-primal}, which is optimal for~\eqref{prob:flow-lex-primal2} as well.
    By Theorem~\ref{thm:flow}~(ii) and the optimality of $\pi^1$, we have $\partial f^* \in \argmin_{y \in \mathcal P^0} \angle{\pi^1, y}$.
    By definition of $\mathcal P^0$, $y = \partial f^*$ minimizes $\angle{\pi^0, y}$ over $\mathcal P$, and, subject to this, minimizes $\angle{\pi^1, y}$.
    This lexicographic optimization is expressed as the minimization of $\angle{\pi^0 + \varepsilon \pi^1, y} = \angle{\pi^*, y}$ over $\mathcal P$, i.e., $\partial f^* \in \argmin_{y \in \mathcal P} \angle{\pi^*, y}$.
    This means the optimality of $\pi^*$ for~\eqref{prob:flow-lex-primal} by Theorem~\ref{thm:flow}~(ii).

    We next show the converse direction.
    Let $f^*$ and $\pi^*$ be an optimal flow and an optimal potential, respectively, for~\eqref{prob:flow-lex-primal}.
    By Theorem~\ref{thm:flow}~(ii), we have $\partial f^* \in \argmin_{y \in \mathcal P} \angle{\pi^*, y}$.
    Since $f^*$ is also optimal for~\eqref{prob:flow-lex-primal} with $\varepsilon = 0$, there exists an optimal potential $\pi^0$ for~\eqref{prob:flow-lex-primal} with $\varepsilon = 0$ such that $\partial f^* \in \argmin_{y \in \mathcal P} \angle{\pi^0, y}$.
    We show that $\pi^1 \coloneqq \frac{\pi^*-\pi^0}{\varepsilon}$ is an optimal potential for~\eqref{prob:flow-lex-primal2}, which proves the claim.
    The feasibility follows from $\varepsilon \delta \pi^1 = \delta \pi^* - \delta \pi^0 = (c^0 + \varepsilon c^1) - c^0 = \varepsilon c^1$.
    As for the optimality, we have $\argmin_{y \in \mathcal P^0} \angle{\pi^1, y} = \argmin_{y \in \mathcal P} \angle{\pi^*, y} \ni \partial f^*$.
    Hence, $\pi^1$ is optimal for~\eqref{prob:flow-lex-primal2} by Theorem~\ref{thm:flow}~(ii).
\end{proof}

\subsection{Proof of Theorem \ref{convex-dual}.}
\label{reduction_convexflow}
We first explain how we can reduce ($P_2$) to the convex flow problem. 
For notational simplicity, hereafter, we denote the polymatroid $P(\rhoc_i(\cdot;p))$ simply by $P(\rhoc_i)$, where $p$ is the current price vector in an iteration of Algorithm~\ref{algo1}. 
For $i \in N$, let $M_i^+ = \Set{j_i^+}{j \in M}$ be a copy of $M$.
We also make one more copy $M^- = \Set{j^-}{j \in M}$ of $M$.
Let $D = (V, A)$ be a directed graph defined by
\begin{align*}
    V \coloneqq \bigcup_{i \in N} M_i^+ \cup M^-, \quad
    A \coloneqq \Set{(j_i^+, j^-)}{i \in N, j \in M}.
\end{align*}
We define a feasible polyhedron $\mathcal P \subseteq \R^V$ as 
$\mathcal P \coloneqq \bigoplus_{i \in N} \mathcal P_i^+ \oplus \mathcal P^-$,
where $\mathcal P_i^+ \coloneqq P(\rhoc_i) \subseteq \R^{M_i^+}$ for $i \in N$ and $\mathcal P^- \coloneqq \bigoplus_{j \in M} [-s_j, +\infty) \subseteq \R^{M^-}$.
For $i \in N$ and $j \in M$, the cost $c_a$ of $a = (j_i^+, j^-) \in A$ is set to $-w_{ij}$.
Then, a flow $ f \in \R^A$ is identified with a vector $x \in \R^{N \times M}$ by $ f_{j_i^+,j^-} = x_{ij}$, where $\angle{c,  f} = -\angle{w, x}$ holds.

We check that the feasibility $\partial  f \in \mathcal P$ is equivalent to that of $x$ in ($P_2$).
For $i \in N$ and $j \in M$, we have ${(\partial  f)_{j_i^+}} =  f_{j_i^+,j^-} = x_{ij}$; hence $(\partial  f)_{v \in M_i^+} \in \mathcal P_i^+$ is equivalent to $x_i\coloneqq {(x_{ij})}_{j \in M} \in P(\rhoc_i)$.
For $j \in M$, we have ${(\partial  f)_{j^-}} = -\sum_{i \in N}  f_{j_i^+} \in [-s_j, +\infty)$, implying $\sum_{i \in N} x_{ij} \le s_j$.
Thus, $ f$ is feasible if and only if $x$ is feasible.

We next derive the dual problem.
Let $\pi = \bigoplus_{i \in N} \pi^+_i \oplus \pi^- \in \R^V$ be a potential on $D$, where $\pi^+_i = {(\pi^+_{ij})}_{j \in M} \in \R^{M_i^+}$ and $\pi^- = {(\pi^-_j)}_{j \in M} \in \R^{M^-}$.
The dual problem~\eqref{prob:flow-dual} reads
\begin{align*}
    \begin{array}{lll}
        - &\displaystyle\maximize_{\pi \in \R^V} & \displaystyle \sum_{i \in N} \min_{x_i \in P(\rhoc_i)} 
        \left\{\sum_{j\in M} \pi^+_{ij} x_{ij}\right\}+ \sum_{j \in M} \min_{-s_j \le y} \pi^-_j y\\
        &\text{subject to} & \pi^+_{ij} - \pi^-_j = -w_{ij} \quad (i \in N, j \in M).
    \end{array}
\end{align*}
Letting $z \coloneqq \pi^-$ and eliminating $\pi_i^+$ by substituting $\pi_{i}^+ = \pi^- - w_i = z - w_i$ for $i \in N$, we have
\begin{align*}
    \begin{array}{lll}
        -&\displaystyle\maximize_{z \in \R^M} & \displaystyle \sum_{i \in N} \min_{x_i \in P(\rhoc_i)} \left\{\sum_{j\in M} (z_j-w_{ij}) x_{ij}\right\} + \sum_{j \in M} \min_{-s_j \le y} z_j y \\
        =& \displaystyle\minimize_{z \in \R^M} & \displaystyle \sum_{i \in N} \max_{x_i \in P(\rhoc_i)} \left\{\sum_{j\in M} (w_{ij}-z_j) x_{ij}\right\} + \sum_{j \in M} \max\set{0, s_j z_j}.
    \end{array}
\end{align*}
Finally, if $z_j$ for some $j \in M$ is negative, increasing it to $0$ does not increase the objective value.
Thus, we may assume $z \in \mathbb R^M_{+}$, implying that the dual problem is
\begin{align*}
    \begin{array}{lll}
        \displaystyle\minimize_{z \in \mathbb R^M_{+}} & \displaystyle \sum_{i \in N} \max_{x_i \in P(\rhoc_i)}  \left\{\sum_{j\in M} (w_{ij}-z_j) x_{ij}\right\} +  \sum_{j \in M} s_j z_j,
    \end{array}
\end{align*}
which is nothing but ($D_2$). 
Therefore, we can apply Theorems \ref{thm:flow} and \ref{lexico_characterization}.
As a preliminary step toward proving Theorem~\ref{convex-dual}, we identify the unique minimal solution for the case of $\gamma=0$.
\begin{lemma}
\label{gamma_0}
Let $X^*$ be the minimal maximally over-demanded set at $p$.
Then, the vector $z_0\coloneqq \chi_{X^*}$ is the unique minimal optimal solution to the problem 
\begin{equation}
\label{L_0}
\min_{z\in \mathbb R^M_{+}}\ \  \mathcal L_0(z), \quad
\text{where} \quad \mathcal L_0(z)\coloneqq
\sum_{i \in N} \max_{x_i\in P(\rhoc_i)} \left\{\sum_{j\in M} (1-z_j) x_{ij}\right\} 
+ \sum_{j \in M} s_j z_j.
\end{equation}
\end{lemma}

\begin{proof}[Proof of Theorem \ref{convex-dual}.]
(i) It suffices to show that $\mathcal L(z)$ is a polyhedral \Ln-convex function.
Recall that the support function of an \Mn-convex set (i.e., the polymatroid $P(\rhoc_i)$) is an \Ln-convex function (by a property derived via the discrete Legendre transformation, see \citet[Section 8]{M2003}).
Consequently, each term $\max_{x_i \in P(\rhoc_i)} \sum_{j\in M} z_j x_{ij}$ is polyhedral \Ln-convex with respect to $z$.
Since the class of polyhedral \Ln-convex functions is closed under affine transformation and addition, $\mathcal L(z)$ is also polyhedral \Ln-convex.

(ii) This follows immediately from Theorem \ref{thm:flow}.

(iii) By Theorem \ref{lexico_characterization}, any optimal solution to the problem $(D_2)$ can be expressed in the form $z^* = z^0 + \gamma t$ for some $t \in \R^M$, where $z^0$ is an optimal solution to the unperturbed problem \eqref{L_0}.
By Lemma \ref{gamma_0}, the minimal optimal solution for the problem \eqref{L_0} is uniquely determined as $\chi_{\scriptscriptstyle X^*}$.
Therefore, the minimal optimal solution to the problem $(D_2)$ must take the form $z^* = \chi_{\scriptscriptstyle X^*} + \gamma t$.
\end{proof}

\section{Concluding Remarks}
\label{Concluding}
We developed a unified framework for computing the \emph{minimum} Walrasian equilibrium in combinatorial markets with frictions. Our auction extends the auctions of \citet{GS2000} and \citet{A2006} by replacing uniform price increments with directional updates, and simultaneously generalizes the unit-demand ITU models of \citet{A1989,A1992} to a combinatorial setting. 
Our technical contribution consists of the characterization of admissible update directions and a strongly polynomial-time algorithm to compute them, obtained by formulating a lexicographic extension of the polymatroid sum problem and reducing its dual to a convex flow problem. Exploiting the \Ln-convexity, we showed that the desired direction is constructed from the minimal dual solution. We also identified a scaled Lyapunov function for separable frictions, linking the auction dynamics to a new class of discrete convexity. While the general case remains an open problem, resolving it would establish a comprehensive framework for markets with general frictions.

Finally, we highlight areas for future research. 
Mathematically, the class of scaled polyhedral \Ln-convex functions warrants further investigation. 
Specifically, determining the extent to which they retain the structural properties of standard \Ln-convexity -- such as closure under basic operations, conjugacy, and optimality conditions -- offers rich potential for theoretical exploration.
Economically, extending the framework beyond strong substitutes is a natural progression. 
Promising starting points include markets with GM-concave valuations \citep{SY2009, SY2015} and trading networks \citep{CEV2021}. 
Developing ascending procedures with frictions for these broader settings would deepen the connection between auctions and discrete optimization, bridging the gap between abstract structural theory and practical market design.

\section*{Acknowledgements}
	We thank Akihisa Tamura for helpful comments.
	Taihei Oki was supported by JST FOREST Grant Number JPMJFR232L, JST BOOST Grant Number JPMJBY25A6, and JSPS KAKENHI Grant Number JP22K17853.
    Ryosuke Sato was supported by JST ERATO Grant Number JPMJER2301 and JSPS KAKENHI Grant Number JP26K21174.

\bibliography{competitive}
\bibliographystyle{plainnat}

\appendix
\section{Properties of Demand Sets}
\label{DemandSets}
We provide properties of demand sets and prove Theorem~\ref{gs_lad}.
\subsection{M-convex Sets}
We first recall the concepts of M-convex and \Mn-convex sets in discrete convex analysis, which play a central role in characterizing the structure of the demand sets.
\begin{definition}[\citet{M2003}]
\label{def:m_convex}
A subset $B \subseteq \mathbb{Z}^M$ is called an M-convex set if for any $x, y \in B$ and $j \in \mathrm{supp}_{+}(x-y)$, there exists $j' \in \mathrm{supp}_{-}(x-y)$ such that $x - \chi_{j} + \chi_{j'} \in B$ and $y + \chi_{j} - \chi_{j'} \in B$.
A subset $B \subseteq \mathbb{Z}^M$ is called an \Mn-convex set if for any $x, y \in B$ and $j \in \mathrm{supp}_{+}(x-y)$, there exists $j' \in \mathrm{supp}_{-}(x-y)\cup \{0\}$ such that
$x - \chi_{j} + \chi_{j'} \in B$ and $y + \chi_{j} - \chi_{j'} \in B$.
\end{definition}

\begin{fact}[\citet{M2003}]
\label{m_convex_fact}
For any price vector $p$ and each $i\in N$, the demand set $D_i(p)$ is an M$^{\natural}$-convex set, and each of $\Dc_i(p)$ and $\Dh_i(p)$ is an M-convex set.
\end{fact}

In general, integral polymatroids and their base polyhedra belong to the classes of \Mn-convex sets and M-convex sets, respectively, and thus possess these exchange properties. Building on these properties, we state the following fact on the existence of minimal demand bundles.
\begin{fact}
\label{minimal_existence}
For any buyer $i$, price vector $p$, and any non-minimal bundle $x_i \in D_i(p)$, there exists a minimal demand bundle $x'_i \in \Dc_i(p)$ such that $x'_i \lneq x_i$.
\end{fact}

This fact is easily verified by taking an $x'_i \in \check{D}_i(p)$ that minimizes the $L_1$-distance to $x_i$. If $x_{ij} < x'_{ij}$ for some $j \in M$, the exchange property of M$^{\natural}$-convex sets immediately yields a contradiction.

\subsection{Proof of Theorem \ref{gs_lad}}
\label{proof_A}

Theorem~\ref{gs_lad} provides a new variant of the GS\&LAD condition for the minimal and maximal demand sets. Its proof follows the argument for the standard GS\&LAD condition in the Japanese edition of \citet{M2022}, which we include here for completeness. We use the following lemma:
\begin{lemma}[{\citet[Lemma 4.3]{MS1999}}]
\label{size_exchange}
Let $v: [\zeros,s]_{\Z} \to \R_+$ be an M$^{\natural}$-concave function.
For any $x, y \in [\zeros,s]_{\Z}$ with $x(M) < y(M)$, there exists $k \in \{j \in M \mid x_j < y_j\}$ such that $v(x) + v(y) \leq v(x + \chi_k) + v(y - \chi_k)$.
\end{lemma}
\begin{proof}[Proof of Theorem \ref{gs_lad}.]
(i) We first consider the case where $x_i \in \Dc_i(p)$.
Choose a bundle $x'_i \in \Dc_i(p')$ that minimizes the $L_1$-distance $\|x_i - x'_i\|_1$.
Suppose to the contrary that there exists $j \in M$ such that $p_j = p'_j$ and $x_{ij} > x'_{ij}$.
Then, by the \Mn-concavity of $v_i$, there exist $j' \in \operatorname{supp}_-(x_i - x'_i) \cup \{0\}$ such that
$v_i(x_i) + v_i(x'_i) \leq v_i(x_i - \chi_j + \chi_{j'}) + v_i(x'_i + \chi_j - \chi_{j'})$.
For the utility of buyer $i$, we have the following relation:
\begin{align*}
&u_i(x_i;p) + u_i(x'_i;p')\\
\qquad&\qquad= v_i(x_i) + v_i(x'_i) - \sum_{l \in M} q_{il}(p_l)x_{il} - \sum_{l \in M} q_{il}(p'_l)x'_{il} \\
\qquad&\qquad\leq v_i(x_i - \chi_j + \chi_{j'}) + v_i(x'_i + \chi_j - \chi_{j'}) 
-\sum_{l \in M} q_{il}(p_l)x_{il} -\sum_{l \in M} q_{il}(p'_l)x'_{il} \\
\qquad&\qquad=u_i(x_i - \chi_j + \chi_{j'}; p)+ u_i(x'_i + \chi_j - \chi_{j'}; p') 
- (q_{ij}(p_j) - q_{ij}(p'_j)) + (q_{ij'}(p_{j'}) - q_{ij'}(p'_{j'})).
\end{align*}
Since each $q_{ij}$ is strictly increasing, by $p_j = p'_j$ and $p_{j'} \leq p'_{j'}$, 
the last two terms are nonpositive.
Therefore, we obtain 
$u_i(x_i;p) + u_i(x'_i;p')\leq u_i(x_i - \chi_j + \chi_{j'}; p)+ u_i(x'_i + \chi_j - \chi_{j'}; p')$.
This implies $x_i - \chi_j + \chi_{j'} \in D_i(p)$ and $x'_i + \chi_j - \chi_{j'} \in D_i(p')$.
If $j' = 0$, then $x_i - \chi_j \in D_i(p)$, which contradicts the minimality of $x_i$.
If $j' \neq 0$, then $x'_i + \chi_j - \chi_{j'} \in \Dc_i(p')$ 
and $\|(x'_i + \chi_j - \chi_{j'}) - x_i\|_1=\|x'_i - x_i\|_1-2$, 
which contradicts the choice of $x'_i$. 

Suppose to the contrary that $x_i(M)<x'_i(M)$. 
By Lemma \ref{size_exchange}, there exists $k \in \operatorname{supp}_-(x_i - x'_i)$ such that 
\begin{align*}
u_i(x_i;p) + u_i(x'_i;p')&\leq u_i(x_i + \chi_k ; p)+ u_i(x'_i - \chi_k; p') + (q_{ik}(p_{k}) - q_{ik}(p'_{k}))\\
&\leq u_i(x_i + \chi_k ; p)+ u_i(x'_i - \chi_k; p').
\end{align*}
This implies $x_i+\chi_k\in D_i(p)$ and $x'_i - \chi_k \in D_i(p')$, which contradicts the minimality of $x'_i$.

Consider the case where $x_i \in \Dh_i(p)$. We choose $x'_i \in \Dh_i(p')$ that minimizes $\|x_i - x'_i\|_1$. The same argument yields $x_i - \chi_j + \chi_{j'} \in D_i(p)$ and $x'_i + \chi_j - \chi_{j'} \in D_i(p')$. Here, $j' = 0$ implies $x'_i + \chi_j \in D_i(p')$, contradicting the maximality of $x'_i$, while $j' \neq 0$ implies $x'_i + \chi_j - \chi_{j'} \in \Dh_i(p')$, contradicting the choice of $x'_i$. If $x_i(M) < x'_i(M)$, Lemma~\ref{size_exchange} similarly gives $x_i+\chi_k\in D_i(p)$, which contradicts the maximality of $x_i$.

(ii) The proof is symmetric to that of (i). Given $x'_i \in \Dc_i(p')$ (resp., $x'_i \in \Dh_i(p')$), we choose $x_i \in \Dc_i(p)$ (resp., $x_i \in \Dh_i(p)$) that minimizes $\|x_i - x'_i\|_1$. The same argument as (i) ensures $x_i - \chi_j + \chi_{j'} \in D_i(p)$ and $x'_i + \chi_j - \chi_{j'} \in D_i(p')$. As before, $j' \neq 0$ immediately contradicts the choice of $x_i$. If $j' = 0$, we have $x_i - \chi_j \in D_i(p)$, contradicting the minimality of $x_i$ for the $\Dc_i$ case, and $x'_i + \chi_j \in D_i(p')$, contradicting the maximality of $x'_i$ for the $\Dh_i$ case, respectively. Finally, if $x_i(M) < x'_i(M)$, applying Lemma~\ref{size_exchange} yields $x'_i - \chi_k \in D_i(p')$ and $x_i + \chi_k \in D_i(p)$, which contradicts the minimality of $x'_i$ for the $\Dc_i$ case and the maximality of $x_i$ for the $\Dh_i$ case, respectively.
\end{proof}

\section{Computing Equilibrium Allocations}
\label{Allocation}

In this section, we explain how to find an equilibrium allocation from an equilibrium price using the demand and exchange oracles.

\subsection{Preliminaries on Submodular Polyhedra}
This section provides preliminaries on submodular polyhedra.
Interested readers may refer to~\citet{F2005}.

Let $E \coloneqq \set{1, \dotsc, m}$ be a finite ground set and $\rho\colon 2^E \to \Z$ a submodular function with $\rho(\emptyset) = 0$, which is not necessarily monotone.
Generalizing polymatroids, the \emph{submodular polyhedron} $P(\rho)$ and the \emph{base polyhedron} $B(\rho)$ are defined as
\begin{align*}
    P(\rho) &\coloneqq \Set{x \in \Z^E}{\text{$x(S) \le \rho(S)$ for all $S \subseteq E$}}, \\
    B(\rho) &\coloneqq \Set{x \in P(\rho)}{x(E) = \rho(E)}.
\end{align*}
For $x \in P(\rho)$, $e \in E$, and $f \in (E \setminus \set{e}) \cup \set{0}$, define
$c(x, e, f) \coloneqq \max\Set{\alpha \in \Z_+}{x + \alpha(\chi_e - \chi_f) \in P(\rho)}$.
The value $c(x,e,f)$ is called the \emph{exchange capacity} if $f \ne 0$ and the \emph{saturation capacity} if $f = 0$.
Using $\rho$, we can express $c(x,e,f)$ as
\begin{align}\label{eq:capacity-submo}
    c(x,e,f) = \min\Set{\rho(S) - x(S)}{S \subseteq E, \, e \in S \not\ni f}.
\end{align}

Here, we present another formulation of $c(x,e,f)$.

\begin{proposition}\label{capacity}
    For $x \in P(\rho)$, $e \in E$, and $f \in (E \setminus \set{e}) \cup \set{0}$, 
    \begin{align*}
        c(x,e,f) = \max\Set{y(e)}{y \in B(\rho), \, y|_{E \setminus \set{f}} \ge x|_{E \setminus \set{f}}} - x(e).
    \end{align*}
\end{proposition}
\begin{proof}
    We first show ``$\le$''. Since $\tilde x \coloneqq x + c(x,e,f) (\chi_e - \chi_f) \in P(\rho)$, there exists $y \in B(\rho)$ such that $y \ge \tilde x$.
    This $y$ satisfies $y(e) \ge \tilde x(e) = x(e) + c(x,e,f)$.
    Since $y|_{E \setminus \set{f}} \ge \tilde x|_{E \setminus \set{f}} \ge x|_{E \setminus \set{f}}$, we have $c(x,e,f) \le y(e)-x(e) \le \max\Set{y(e)}{y \in B(\rho), \, y|_{E \setminus \set{f}} \ge x|_{E \setminus \set{f}}} - x(e)$.

    We next show ``$\ge$''.
    Let $S^*$ be a minimizer of the right-hand side of~\eqref{eq:capacity-submo}.
    This implies $c(x,e,f) = \rho(S^*) - x(S^*)$.
    Then, for any $y \in B(\rho)$ with $y|_{E \setminus \set{f}} \ge x|_{E \setminus \set{f}}$, we have
    \[
        c(x,e,f) + x(e)
        = \rho(S^*) - x(S^* \setminus \set{e})
        \ge y(S^*) - y(S^* \setminus \set{e}) = y(e). 
    \]
\end{proof}

It is a classical result that a linear objective function can be optimized over a base polyhedron $B(\rho)$ by local search on $B(\rho)$ (see, e.g.,~\citet{M2003}).
Recently, \citet{OS2026} presented a long-step version of this algorithm for general M-convex functions, establishing the following bound on the number of iterations.

\begin{proposition}[{\citet{OS2026}}]\label{linear-opt-submo-base}
    Given the base polyhedron $B$ of a submodular function as an exchange capacity oracle, a point $x \in B$, and a linear objective function $c$ over $B$, we can compute a point maximizing $c$ over $B$ in polynomially many oracle queries.
\end{proposition}

The \emph{submodular intersection problem} is defined as follows: given two submodular polyhedra, the goal is to find a common point in their base polyhedra.
This problem can be solved by using only exchange and saturation capacities.

\begin{proposition}[{see~\citet[Theorem~4.11]{F2005}}]
\label{polymatroid-intersection-algo}
    Given two submodular polyhedra as exchange and saturation capacity oracles, we can find a common point of their base polyhedra in strongly polynomial~time.
\end{proposition}

\subsection{Reduction to Submodular Intersection}

Let $p^* \in \R_+^M$ be an equilibrium price.
Let
\begin{align*}
    B_1 &\coloneqq \Set{{(x_i)}_{i \in N}}{\text{$x_i \in D_i(p^*)$ $(i \in N)$, $\sum_{i \in N} x_i(M) = s(M)$}}, \\
    \quad
    B_2 &\coloneqq \Set{{(x_i)}_{i \in N}}{\text{$x_i \in \Z_+^M$ $(i \in N)$, $\sum_{i \in N} x_i = s$}}.
\end{align*}
Then, $B_1 \cap B_2$ is exactly the set of equilibrium allocations.
Note that $B_1 \cap B_2 \ne \emptyset$ because $p^*$ is an equilibrium price by assumption.

Since each demand set $D_i(p^*)$ is an \Mn-convex set, $B_1$ is an M-convex set, or equivalently, the base polyhedron of some submodular function.
The corresponding submodular polyhedron is given by
\begin{align*}
    P_1 \coloneqq \Set{x \in \Z^{N \times M}}{\text{$x \le y$ for some $y \in B_1$}}.
\end{align*}
The other set $B_2$ is also an M-convex set in $\Z^{N \times M}$, having the associated polyhedron
\begin{align*}
    P_2 &\coloneqq \Set{{(x_i)}_{i \in N}}{\text{$x_i \in \Z^M$ $(i \in N)$, $\sum_{i \in N} x_i \le s$}}.
\end{align*}
Therefore, the problem of finding an equilibrium allocation is merely the intersection of the two submodular polyhedra $P_1$ and $P_2$.
By Proposition~\ref{polymatroid-intersection-algo}, it suffices to construct oracles for the saturation and exchange capacities of $P_1$ (those for $P_2$ are easy).

\subsection{Constructing Oracles for Exchange and Saturation Capacities}

Let $x = {(x_i)}_{i \in N} \in P_1$, $e \in N \times M$, and $f \in ((N \times M) \setminus \set{e}) \cup \set{0}$.
We present an algorithm for computing the capacity $c(x,e,f)$ for $P_1$.
We first focus on the case where $x \in B_1$, and then consider a general $x \in P_1$.
For $i \in N$, $x_i \in D_i(p^*)$, and distinct $j, j' \in M \cup \set{0}$, let $c_i(x_i, j, j') \coloneqq \max\Set{\alpha \in \Z_+}{x_i + \alpha(\chi_j - \chi_{j'}) \in D_i(p^*)}$ denote the output of the exchange oracle for $D_i(p^*)$.

Suppose that $x \in B_1$.
If $f = 0$, then $c(x,e,f) = 0$, since $x(E)$ is maximal over $P_1$.
Now suppose that $f \in (N \times M) \setminus \{e\}$, and write
$e = (i,j)$ and $f = (i',j')$.
If $i = i'$, then $x + \alpha(\chi_e - \chi_f) \in B_1$ if and only if
$x_i + \alpha(\chi_j - \chi_{j'}) \in D_i(p^*)$.
Therefore, we have $c(x,e,f) = c_i(x_i,j,j')$.
If $i \neq i'$, then $x + \alpha(\chi_e - \chi_f) \in B_1$ is equivalent to
$x_i + \alpha \chi_j \in D_i(p^*)$ and
$x_{i'} - \alpha \chi_{j'} \in D_{i'}(p^*)$, implying $c(x,e,f) = \min\set{c_i(x_i, j, 0), c_{i'}(x_{i'}, 0, j')}$.
Therefore, $c(x,e,f)$ can be computed by using exchange oracles in any case.

We next consider a general point $x \in P_1$.
By Proposition~\ref{capacity}, this reduces to solving a linear optimization problem over $B_1$ under the constraint $y(e') \ge x(e')$ for $e' \in E \setminus \set{f}$.
To capture this constraint in a general form, for $S \subseteq E$, let
$B_S \coloneqq \Set{y \in B_1}{\text{$y(e') \ge x(e')$ for all $e' \in S$}}$.
The feasible set of the above problem is written as $B_{E \setminus \set{f}}$.
Since $B_S$ is the intersection of $B_1$ with box constraints, it is again
the base polyhedron of a submodular function.
By Proposition~\ref{linear-opt-submo-base}, this linear optimization problem can be solved by (i) constructing a point in $B_S$ and (ii) computing exchange capacities on $B_S$.
We describe these steps below.

(i) We first construct a point in $B_1$ and then transform it into a point in $B_S$.
Let $y_i$ be a point in $D_i(p^*)$ given by the demand oracle for $i \in N$ and $y \coloneqq {(y_i)}_{i \in N}$.
If the sum $y(N \times M) = \sum_{i \in N} y_i(M)$ of the components in $y$ is less than $s(M)$, we repeatedly improve $y$ by choosing $i \in N$ and $j \in M$ with $c_i(y_i, j, 0) > 0$ and replacing $y_i$ with $y_i + \alpha\chi_j$, where $\alpha \coloneqq \min\set{c_i(y_i, j, 0), s(M) - y(N \times M)}$.
Since each $(i, j)$ is chosen at most once, this procedure halts in polynomial time.
An analogous procedure applies in the opposite case where $y(N \times M) > s(M)$.

We then transform $y$ into a point in $B_S$.
Let $e_1, \dotsc, e_{|S|}$ be an arbitrary ordering of $S$ and let $S_k \coloneqq \set{e_1, \dotsc, e_k}$ for $k = 0, \dotsc, |S|$.
Suppose that we have $y_k \in B_{S_k}$ for some $0 \le k < |S|$; we first take $y_0 \coloneqq y \in B_{S_0}$.
Consider the following optimization problem:
\begin{align}\label{eq:y_to_B_S_k}
    \max\Set{y(e_{k+1})}{y \in B_{S_k}}.
\end{align}
By Proposition \ref{linear-opt-submo-base}, this problem can be solved with a polynomial number of oracle queries for exchange capacities on $B_{S_k}$.
As seen in (ii), the exchange capacities are computable  from the exchange oracles of the demand sets.
Since an optimal solution to~\eqref{eq:y_to_B_S_k} is in $B_{S_{k+1}}$, we can construct $y_0, y_1, \dotsc, y_{|S|}$ in order.

(ii) Let $y = {(y_i)}_{i \in N} \in B_S$, $i, i' \in N$, and $j, j' \in M$ with $(i,j) \ne (i',j')$.
The exchange capacity $c_S(y, (i,j), (i',j'))$ on $B_S$ is computed from the exchange capacity on $B_1$ as follows.
If $(i',j') \notin S$, then $\tilde y \coloneqq y + \alpha(\chi_{(i,j)} - \chi_{(i',j')})$ is in $B_S$ if and only if $\tilde y \in B_1$.
Therefore, $c_S(y, (i,j), (i',j')) = c(y, (i,j), (i',j'))$ holds.
If $(i',j') \in S$, we have
\begin{align*}
    c_S(y, (i,j), (i',j'))
    &= \max\Set{\alpha \in \Z_+}{y + \alpha(\chi_{(i,j)} - \chi_{(i',j')}) \in B_S} \\
    &= \max\Set{\alpha \in \Z_+}{y + \alpha(\chi_{(i,j)} - \chi_{(i',j')}) \in B_1, \, y_{i'j'} - \alpha \ge x_{i'j'}} \\
    &= \min\set{c(y, (i,j), (i',j')), y_{i'j'} - x_{i'j'}}.
\end{align*}
Thus, we can compute the exchange capacity $c_S(y, (i,j), (i',j'))$ from $c(y, (i,j), (i',j'))$.
Since the exchange capacity on $B_1$ is computable using the exchange oracles as explained above, our reduction is complete.

\section{Formulation of the Auxiliary Assignment Problem}
\label{Interpretation_detail}
This section formalizes the construction of the assignment problem in Section~\ref{Interpretation_abst}. Specifically, we illustrate how the virtual market on $X^*$ is formulated as an assignment problem on the exchange~graph.

\subsection{Construction of the Exchange Graph}
\label{exchange_graph}
Let $x \coloneqq (x_i)_{i \in N}$ be an optimal solution to $(P_2)$ under the current price vector $p$. 
Since $x_i \in P(\rhoc_i(\cdot; p))$, we can extend it to $x'_i \in \Dc_i(p)$ by adding goods. Note that this extension may cause the total allocation to exceed the actual supply $s$. 
Based on such an allocation $x'$, we construct the \emph{exchange graph} $G(x') = (K, M, F)$, a bipartite graph in which each original buyer is split into unit-demand copies. The left nodes $K$ are these copies, and the right nodes are the goods $M$. Specifically, for every pair $(i,j)$ with $x'_{ij} > 0$, we add $x'_{ij}$ nodes to $K$. For each such node $k$, we set $\Gamma(k) = i$ (its original buyer) and $\phi(k) = j$ (the good assigned to $k$). The edge set $F$ consists of the following directed edges:

\begin{itemize}
    \item \textbf{Forward edges:} For each copy $k \in K$, we add a directed edge $(k, \phi(k))$.
    \item \textbf{Backward edges:} For each $k \in K$ and each good $j' \in M \setminus \{\phi(k)\}$, we add a directed edge $(j', k)$ if the corresponding swap is feasible for buyer $\Gamma(k)$, i.e., $x'_{\Gamma(k)} - \chi_{\phi(k)} + \chi_{j'} \in \Dc_{\Gamma(k)}(p)$.
\end{itemize}
This graph represents all feasible local exchanges. We call a good $j$ \emph{oversold} if its total allocation exceeds its supply, i.e., $\sum_{i \in N} x'_{ij} > s_j$. The set $X^*$ then admits a clear graph-theoretic characterization.
\begin{theorem}[\citet{ENPRVV2025}, Theorem~3.3]
\label{thm:overdemanded}
The minimal maximally over-demanded set $X^*$ consists of all goods 
from which there exists a directed path to an \emph{oversold} good in the exchange graph.
\end{theorem}

\subsection{Assignment Problem on the Induced Graph}
Let $G^* = (K^*, X^*, E^*)$ be the induced exchange graph, where $K^* \subseteq K$ is the set of copies matched to goods in $X^*$ (i.e., $\phi(k) \in X^*$), and $E^*$ is the set of undirected edges between $K^*$ and $X^*$ derived from the forward and backward edges of $G(x')$. For each copy $k \in K^*$ and each good $j$ with $(k,j) \in E^*$, we define $v_{kj} \coloneqq 1 - \gamma \log q'_{\Gamma(k) j}(p_j)$. We now introduce the assignment problem on $G^*$, seeking an equilibrium allocation $y$ and the minimum equilibrium price vector~$z$.
\begin{problem}[Assignment Problem on $G^*$]
\label{prob:assignment}
Find an allocation $y = (y_{kj})_{(k,j)\in E^*}$ and a price vector $z \in \R^{X^*}_+$
satisfying the following conditions:
\begin{align*}
    y_{kj} = 1  \quad &\Rightarrow \quad v_{kj} - z_{j} = \max\Bigl\{0, \ \max_{j':(k,j')\in E^*} \{ v_{kj'} - z_{j'} \}\Bigr\}
    && ((k,j) \in E^*),  \\
    \sum_{j:(k,j)\in E^*} y_{kj} = 0 \quad &\Rightarrow \quad v_{kj'}-z_{j'} \leq 0 \quad (\forall j' : (k,j')\in E^*)
    && (k \in K^*),  \\
    y_{kj} \in \{0,1\} \ \ ((k,j) \in E^*), \quad &\sum_{j:(k,j)\in E^*} y_{kj} \leq 1 \ \ (k \in K^*), \quad
    \sum_{k:(k,j)\in E^*} y_{kj} = s_j \ \ (j \in X^*).
\end{align*}
We call such a $z$ an equilibrium price vector of $G^*$ and seek the minimum one.
\end{problem}

Lemma \ref{independence} implies that the dual variable $z^*|_{X^*}$ computed by Algorithm~\ref{alg:direction} is precisely the minimum equilibrium price $z$ of this problem. By the classical result of \citet{L1983}, the minimum equilibrium prices of a unit-demand market coincide with the VCG prices. Hence $d^*$ encodes the VCG prices of this assignment market, formalizing our economic interpretation.

\section{Other Omitted Proofs}
\label{proofs}
This section provides the  proofs omitted from the main text. Note that the proof of Theorem~\ref{gs_lad} is provided in Appendix~\ref{proof_A}. All other proofs are presented in this~section.

\subsection{Proof of Lemma \ref{price_characterization}}
Recall that $p$ is an equilibrium price vector if and only if $s \in D(p)$, where $D(p)$ is the Minkowski sum of $\{D_i(p)\}$. 
Since the class of ${\rm M}^\natural$-convex sets is closed under Minkowski sum \citep[Theorem 4.23]{M2003}, $D(p)$ is also ${\rm M}^\natural$-convex. 
Furthermore, because the convex hull of an ${\rm M}^\natural$-convex set is an integral polyhedron \citep[Section~4.7]{M2003} and $s$ is integral, $s \in D(p)$ is equivalent to $s \in \mathrm{conv}(D(p))$. 
This convex hull is given by $\mathrm{conv}(D(p)) = \{ y \in \R^M_+ \mid \mu(X; p) \le y(X) \le \rhoh(X; p)$ for all $X \subseteq M \}$.
Therefore, $s \in D(p)$ is equivalent to: $\mu(X; p) \le s(X) \le \rhoh(X; p)$ for all $X \subseteq M$. These two inequalities state that no set is over-demanded or under-demanded at $p$, respectively.

\subsection{Proof of Lemma \ref{subset}}
If $d=\zeros$, the claim holds trivially. Assume that $d\neq \zeros$.
Consider a small price increase $\varepsilon\in (0,\delta_i)$ along $d$, where $\delta_i$ is defined by  
\begin{align}
\label{delta}
\delta_i\coloneqq \min\left(\frac{\max_{x'_i\in D_i(p)} u_i(x'_i; p)-\max_{x'_i \in [\zeros,s]_{\mathbb Z}\setminus D_i(p)} u_i(x'_i; p)}{s(M) \max_{j\in M} d_j q'_{ij}(p_j)}, \ r_i\right).
\end{align}
Note that $r_i$ denotes the minimum positive step size at which the slope
of buyer $i$'s payment function changes along the direction $d$ starting from $p$.
For any $x^*_i \in D_i(p)$ and $x_i \notin D_i(p)$, 
we have
\begin{align}
\label{utility-difference}
u_i(x^*_i; p) -
\left(\max_{x'_i\in D_i(p)} u_i(x'_i; p)-\max_{x'_i \in [\zeros,s]_{\mathbb Z}\setminus D_i(p)} u_i(x'_i; p)\right)
\geq u_i(x_i; p).
\end{align}
For a sufficiently small price increase $\varepsilon\in (0,\delta_i)$, 
it holds 
\begin{align*}
u_i(x^*_i; p+\varepsilon d) 
 &\geq u_i(x^*_i; p) - \varepsilon s(M) \max_{j\in M} d_j q'_{ij}(p_j) \\
&>u_i(x^*_i; p) - \left(\max_{x'_i\in D_i(p)} u_i(x'_i; p)-\max_{x'_i \in [\zeros,s]_{\mathbb Z}\setminus D_i(p)} u_i(x'_i; p) \right) && (by \ (\ref{delta}))\\
&\geq u_i(x_i; p)  && (by\ (\ref{utility-difference}))\\
&\geq u_i(x_i; p+\varepsilon d).
\end{align*}
Therefore, the bundle $x_i \notin D_i(p)$ does not maximize $u_i(\cdot;p+\varepsilon d)$.
This implies that if a bundle yields the optimal utility at 
$p+\varepsilon d$, it must also be optimal at $p$.
Thus, we have $D_i(p+\varepsilon d) \subseteq D_i(p)$.

Next, we show that $\Dc_i(p+\varepsilon d) \subseteq \Dc_i(p)$.
Suppose that $x^*_i$ is a non-minimal bundle in $D_i(p)$.
By Fact \ref{minimal_existence}, there exists a minimal bundle $x^{\dag}_i \in \Dc_i(p)$ such that $x^{\dag}_i \lneq x^*_i$.
Since both $x^*_i$ and $x^{\dag}_i$ are in $D_i(p)$, they yield the same utility at price $p$.
For any  $\varepsilon \in (0,\delta_i)$, the inequality $x^{\dag}_i \lneq x^*_i$ implies 
\begin{align*}
u_i(x^{\dag}_i; p+\varepsilon d) 
=u_i(x^{\dag}_i; p) - \varepsilon \sum_{j \in M} q'_{ij}(p_j)d_j x^{\dag}_{ij}
\geq u_i(x^{*}_i; p) - \varepsilon \sum_{j \in M} q'_{ij}(p_j)d_j x^{*}_{ij}
=u_i(x^{*}_i; p+\varepsilon d),
\end{align*}
where the inequality holds by $x^{\dag}_i \lneq x^*_i$. 
This implies that $x^*_i\notin \Dc_i(p+\varepsilon d)$.
By $D_i(p+\varepsilon d) \subseteq D_i(p)$, any minimal bundle in $D_i(p+\varepsilon d)$ must be a minimal bundle in $D_i(p)$, i.e.,  $\Dc_i(p+\varepsilon d) \subseteq \Dc_i(p)$.

This relationship holds simultaneously for all buyers if we set $\delta\coloneqq\min_{i\in N}\delta_i$.

\subsection{Proof of Lemma \ref{perfect-matching}}
\label{perfect}
Choose a bundle $x'_i \in D_i(p)$ that minimizes the $L_1$-distance $\|x_i - x'_i\|_1$ subject to $x'_i(X) = \mu_i(X; p)$. Since $x_i(X) > x'_i(X)$, there exists $j \in X$ such that $x_{ij} > x'_{ij}$. 
By the exchange property of M$^{\natural}$-convex sets, for such an index $j \in \mathrm{supp}_+(x_i-x'_i)$, 
there exists an index $j' \in \mathrm{supp}_{-}(x_i-x'_i)\cup\{0\}$ such that 
$x''_i\coloneqq x'_i + \chi_{j} - \chi_{j'} \in D_i(p)$. 
Suppose to the contrary that $j' \in X$. Then, $x''_i(X) = \mu_i(X; p)$ and this strictly reduces the distance to $x_i$, contradicting the choice of $x'_i$.
Therefore, we must have $j' \in (M \setminus X)\cup\{0\}$.

\subsection{Proof of Lemma \ref{independence}}
Applying the same reduction as in Section~\ref{reduction_convexflow} -- replacing $M$ with $X^*$ and setting $\mathcal P_i^+ \coloneqq P(\rhot_i)$ instead of $P(\rhoc_i)$ -- the minimization problem in the lemma reduces to the dual of a convex flow problem. 
Then, we can apply Theorem~\ref{lexico_characterization} in the subsequent analysis.

\begin{proof}
It suffices to show that $z^*|_{X^*}$ is the unique minimal solution to the following problem:
\begin{equation}
\label{L_X_0}
\min_{z'\in \mathbb R^{X^*}_{+}}\cL_{X^*}(z'), \quad
\text{where} \quad  \mathcal L_{X^*}(z')\coloneqq \sum_{i \in N} \max_{y_i\in P(\rhot_i)} \sum_{j\in X^*} (w_{ij}-z'_{j}) y_{ij} + \sum_{j \in X^*} s_j z'_j.
\end{equation}

We first consider the base case where $\gamma = 0$. Imposing the condition $z_j = 0$ for all $j \in M \setminus X^*$ in problem \eqref{L_0} yields exactly the objective function of \eqref{L_X_0} with $\gamma = 0$, up to a constant term. 
Since the unique minimal solution $z_0 \coloneqq \chi_{\scriptscriptstyle X^*}$ to \eqref{L_0} (see Lemma \ref{gamma_0}) satisfies this condition, its restriction $\chi_{\scriptscriptstyle X^*}$ naturally serves as the unique minimal solution to \eqref{L_X_0} for $\gamma = 0$.

Next, we apply Theorem~\ref{lexico_characterization} to the restricted problem \eqref{L_X_0} for $\gamma > 0$. 
Since $\chi_{\scriptscriptstyle X^*}$ is the unique minimal solution for $\gamma = 0$, the unique minimal solution to \eqref{L_X_0} for $\gamma > 0$ must take the form $z' = \chi_{\scriptscriptstyle X^*} + \gamma t'$, where $t'$ is the unique minimal solution to the following problem:
\begin{equation}
\label{t}
\min_{t'' \in \mathbb R^{X^*}} 
\left( \sum_{i \in N} \max_{y_i \in P(\rhot_i)} \sum_{j \in X^*} (-\log q'_{ij}(p_j) - t''_{j}) y_{ij} + \sum_{j \in X^*} s_j t''_j \right).
\end{equation}

Now, we relate this to the original problem $(D_2)$. By Theorem~\ref{convex-dual} (iii), its unique minimal solution $z^*$ is characterized as $z^* = \chi_{\scriptscriptstyle X^*} + \gamma t$ for some $t \in \R^M$. Then, Lemma~\ref{difference_L2} with $\xi=0$ implies that the objective function $\cL(z^*)$ decomposes into two independent components:
\begin{align}
    \cL(z^*) &= 
    \left(\sum_{i \in N} \max_{x'_i \in P(\rhoc_i)|_{M \setminus X^*}} \sum_{j \in M \setminus X^*} (w_{ij} - z^*_j) x'_{ij}+ \sum_{j \in M\setminus X^*} s_j z^*_j\right) \\
    &\quad+\left(\sum_{i \in N} \max_{y_i \in P(\rhot_i)} \sum_{j \in X^*} (w_{ij} - z^*_j) y_{ij} + \sum_{j \in X^*} s_j z^*_j\right).
\end{align}
This guarantees that $t|_{X^*}$ and $t|_{M \setminus X^*}$ can be optimized independently. Notice that the last two terms exactly correspond to $\cL_{X^*}(z^*|_{X^*})$. Substituting $z^*|_{X^*} = \chi_{\scriptscriptstyle X^*} + \gamma (t|_{X^*})$ into this component reveals that minimizing it with respect to $t|_{X^*}$ is equivalent to solving \eqref{t}. Since $z^*$ is the unique minimal minimizer of $\cL$, its component $t|_{X^*}$ must be the unique minimal solution to \eqref{t}. We have $t' = t|_{X^*}$, which implies $z' = \chi_{\scriptscriptstyle X^*} + \gamma t' = z^*|_{X^*}$ and thus $z^*|_{X^*}$ is the unique minimal solution to \eqref{L_X_0}.
\end{proof}

\subsection{Proof of Lemma \ref{L_scaled}}
Observation~\ref{obs} ensures that $\tilde{L}$ inherits the regularity conditions in Definition~\ref{def_scaled} from standard polyhedral $\mathrm{L}^{\natural}$-convex functions. It thus suffices to establish translation-submodularity with respect to $\beta^{-1}$: for any $p, p' \in \R_{+}^{M}$ and any $\lambda \in \R_+$,
$\tilde{L}(p) + \tilde{L}(p') \ge \tilde{L}((p - \lambda \beta^{-1}) \vee p') + \tilde{L}(p \wedge (p' + \lambda \beta^{-1}))$. 
We establish this inequality for each $\tilde{V}_i$ and the linear term.

We first establish the inequality for $\tilde{V}_i$.
By Observation~\ref{obs}, we have $\tilde{V}_i(p) = V_{i, \alpha_i}(\beta \circ p)$ for some indirect utility function $V_{i, \alpha_i}$.
By \Ln-convexity, $V_{i, \alpha_i}$ satisfies translation-submodularity with respect to the all-ones vector $\mathbf{1}$.
Since the coordinate-wise scaling $p \mapsto \beta \circ p$ preserves lattice operations, applying translation-submodularity of $V_{i, \alpha_i}$ to $\beta \circ p$ and $\beta \circ p'$ yields
\begin{align*}
\tilde{V}_i(p) + \tilde{V}_i(p')
&= V_{i, \alpha_i}(\beta \circ p) + V_{i, \alpha_i}(\beta \circ p') \\
&\ge V_{i, \alpha_i}((\beta \circ p - \lambda \mathbf{1}) \vee (\beta \circ p')) + V_{i, \alpha_i}((\beta \circ p) \wedge (\beta \circ p' + \lambda \mathbf{1})) \\
&= V_{i, \alpha_i}(\beta \circ ((p - \lambda \beta^{-1}) \vee p')) + V_{i, \alpha_i}(\beta \circ (p \wedge (p' + \lambda \beta^{-1}))) \\
&= \tilde{V}_i((p - \lambda \beta^{-1}) \vee p') + \tilde{V}_i(p \wedge (p' + \lambda \beta^{-1})).
\end{align*}

For the linear term, it suffices to verify the equality for each coordinate $j$. 
Using the identities $\min(a+c, b+c) = \min(a,b) + c$ and 
$\max(a,b) + \min(a,b) = a+b$, we obtain:
\begin{align*}
\max(p_j - \lambda (\beta^{-1})_j, p'_j) + \min(p_j, p'_j + \lambda (\beta^{-1})_j)
= \left(p_j- \lambda (\beta^{-1})_j\right) + p'_j + \lambda (\beta^{-1})_j
= p_j + p'_j.
\end{align*}
Therefore, the scaled Lyapunov function $\tilde{L}$ satisfies translation-submodularity with respect to $\beta^{-1}$.

\subsection{Proof of Lemma \ref{gamma_0}}
Before proceeding to the proof, we recall a key fact regarding the Lov\'{a}sz extension of submodular functions and its connection to polymatroid optimization.
\begin{fact}[Lov\'{a}sz Extension; {see~\citet[Section 6.3]{F2005}}]
\label{fact:lovasz}
Let $f: 2^M\to \mathbb{R}_+$ be a submodular function with $f(\emptyset) = 0$. Its Lov\'{a}sz extension $\hat{f}: [0, 1]^M \to \mathbb{R}_+$ is defined by $\hat{f}(z) = \sum_{k=1}^K \lambda_k f(U_k)$ for some integer $K$, where $z = \sum_{k=1}^K \lambda_k \chi_{\scriptscriptstyle U_k}$ is a chain decomposition with $\emptyset \subseteq U_1 \subset \cdots \subset U_K \subseteq M$, $\lambda_k > 0$, and $\sum_{k=1}^K \lambda_k = 1$. Then, the following properties hold:
\begin{enumerate}
    \item[{\rm (i)}] $\max \bigl\{ \sum_{j\in M} z_j x_j \mid x \in P(f) \bigr\} = \hat{f}(z)$ for any $z \in [0, 1]^M$. In particular, $\hat{s}(z) = \sum_{j \in M} s_j z_j$ for the modular function $s$ with $s(X) = \sum_{j \in X} s_j$ for any $X \subseteq M$.
    \item[{\rm (ii)}] The extension is additive, i.e., $\widehat{f_1+f_2} = \hat{f}_1 + \hat{f}_2$ for any submodular functions $f_1, f_2$.
\end{enumerate}
\end{fact}

\begin{proof}[Proof of Lemma \ref{gamma_0}.]
The term $\max\{\sum_{j\in M} (1-z_j) x_{ij}\}$ is a linear optimization over a polymatroid. For any $z_j > 1$, the coefficient $1 - z_j$ is negative, so the corresponding component $x_{ij}$ vanishes (Fact~\ref{greedy_structure}~(ii)). Hence the first term is constant in $z_j$, while the second term $\sum_j s_j z_j$ strictly increases. Consequently, the objective function strictly increases for $z_j > 1$, so any optimal solution must satisfy $z \in [0, 1]^M$.

For $z \in [0, 1]^M$, we use Fact~\ref{fact:lovasz}. 
The maximization term over $P(\rhoc_i)$ evaluates the Lov\'{a}sz extension of $\rhoc_i$ at $\ones-z$. 
Crucially, this corresponds to the Lov\'{a}sz extension of the set function $X \mapsto \rhoc_i(M \setminus X)$ evaluated at $z$. 
Combined with the extension of the modular term $s(X)$ and the additivity of the extension, $\cL_0(z)$ exactly coincides with the Lov\'{a}sz extension $\hat{F}(z)$ of $F(X) \coloneqq \sum_{i \in N} \rhoc_i(M \setminus X) + s(X)$. 
Thus, $\cL_0(z)$ can be expressed via a chain decomposition as $\cL_0(z) = \sum_{k=1}^K \lambda_k F(U_k)$. 
For $z$ to be optimal, this convex combination must attain the minimum of $F$, requiring every set $U_k$ to independently minimize $F$. 
Since $F$ is submodular and $X^*$ is the unique minimal minimizer of $F$, any optimal solution must satisfy $X^* \subseteq U_k$ for all $k$. 
This immediately yields $z \ge \chi_{\scriptscriptstyle X^*}$,
which implies that $\chi_{\scriptscriptstyle X^*}$ is the unique 
minimal optimal solution to \eqref{L_0}. 
\end{proof}

\section{Non-Uniqueness of the Direction}
In this section, we show that a direction $d^*$ satisfying \LSC is not necessarily unique. 
To illustrate this, we consider the case of additive valuations. Under such valuations, each buyer's minimal demand set consists of a single bundle that contains all the goods with positive marginal utility.

\begin{proposition}
\label{additive_weight}
Suppose that each buyer $i \in N$ has an additive valuation. 
Then, any direction $d$ with $\operatorname{supp}_+(d)=X^*$ satisfies {\rm \LSC}. 
\end{proposition}
\begin{proof}
Suppose that each buyer $i$ has an additive valuation given by $v_i(x_i) = \sum_{j \in M} v_{ij} x_{ij}$ for any $x_i \in [\zeros, s]_{\Z}$, where the vector $\{v_{ij}\}_{j \in M} \in \R_+^M$ represents the marginal valuations of buyer $i$.
Then, $\Dc_i(p)$ consists of a unique bundle $x^*_i = (x^*_{ij})_{j \in M}$ defined~by $x^*_{ij}=s_j$ if $v_{ij} - q_{ij}(p_j) > 0$ and $x^*_{ij}= 0$ otherwise.

By Lemma~\ref{subset}, $\Dc_i(p+\varepsilon d)$ is a non-empty subset of $\Dc_i(p)$ for sufficiently small $\varepsilon > 0$. 
This implies $\Dc_i(p+\varepsilon d) = \{x^*_i\} = \Dc_i(p)$.
Then, the requirement on any set $X$ is locally unchanged:
\[
    \mu_i(X; p) = x^*_i(X) = \mu_i(X; p+\varepsilon d) \quad \text{for all } X \subseteq M.
\]
Therefore, $\mathcal O(X;p)=\sum_{i\in N}x^*_i(X)- s(X)=\mathcal O(X;p+\varepsilon d)$ holds for all $X\subseteq M$. 
The minimal maximizer $X^*$ at $p$ remains the minimal maximizer at $p+\varepsilon d$, and thus \LSC is satisfied.
\end{proof}

The following example illustrates that this non-uniqueness can occur even for non-additive valuations.
In this example, as in Proposition~\ref{additive_weight}, the non-uniqueness arises because each minimal demand set consists of a single bundle.
\begin{example}
\label{ex:non_unique_d}
Consider a market with two goods (denoted by $1$ and $2$), each with a unit supply, and two buyers.
Suppose that both buyers have the same valuation function:
\begin{equation}
    v_i(\emptyset)=0, \quad v_i(\{1\})=2, \quad v_i(\{2\})=3, \quad v_i(\{1,2\})=4 \quad (i=1,2).
\end{equation}
Note that this valuation is not additive, as $v_i(\{1\}) + v_i(\{2\}) = 5 \neq 4 = v_i(\{1,2\})$.
However, it belongs to the class of OXS functions, a subclass of $\mathrm{M}^{\natural}$-concave functions.

At the price vector $p=(0,0)$, the minimal demand set is 
uniquely determined as $\Dc_i(p)=\{(1,1)\}$ for $i=1,2$.
Recall the logic used in the proof of Proposition~\ref{simple}: if $\Dc_i(p)$ is a singleton and the inclusion $\Dc_i(p+\varepsilon d) \subseteq \Dc_i(p)$ holds (by Lemma~\ref{subset}), then we must have $\Dc_i(p+\varepsilon d) = \Dc_i(p)$.
This invariance implies that \LSC holds for \emph{any} direction $d$ with $\operatorname{supp}_+(d)=\{1,2\}$.
For instance, both $d=(1,1)$ and $d'=(1,2)$ satisfy the condition, yet $d$ cannot be represented as a scalar multiple of $d'$.
Therefore, we conclude that the admissible direction is not unique.
\end{example}

\section{Example Illustrating Violations of \LSC}
\label{Example}
This section provides an example illustrating Cases~(iii) and~(iv) in Proposition~\ref{event_b}. In the algorithms of \citet{DHW2015} and \citet{BEF2024}, these correspond to a new edge entering the demand graph between a buyer and a good, both in the maximal alternating tree.
\begin{example}
\label{ex:event_b_cases}
Consider a market with three buyers $N=\{1, 2, 3\}$ and two goods $M=\{1, 2\}$, each with a unit supply. Each valuation function belongs to the class of unit-demand functions, a subclass of $\mathrm{M}^{\natural}$-concave functions: $v_i(0,0) = 0$ and $v_i(1,1) = \max\{v_i(1,0), v_i(0,1)\}$ for all $i$, with
\[
v_1(1,0) = 8.2,\quad v_1(0,1) = 7,\quad v_2(1,0) = 8,\quad v_2(0,1) = 9.5,\quad v_3(1,0) = v_3(0,1) = 10.
\]
The payment slopes are price-independent, with $q'_{11} = 2$, $q'_{12} = 1.6$, $q'_{21} = 0.5$, $q'_{22} = 2$, and $q'_{31} = q'_{32} = 1$.

As explained in Appendix~\ref{Interpretation_detail}, computing the minimal solution $z^*$ to $(D_2)$ reduces to finding the VCG payments in a virtual market where buyer $i$'s valuation for good $j$ is $1 - \gamma \log q'_{ij}(p_j)$. Specifically, $z^*$ gives the VCG payments associated with the allocated goods.

At the initial price $p=\zeros$, the minimal demand set for each buyer is $\Dc_1(p) = \{(1,0)\}$, $\Dc_2(p) = \{(0,1)\}$, and $\Dc_3(p) = \{(1,0), (0,1)\}$.
Importantly, the minimal maximally over-demanded set $X^*$ is $\{1, 2\}$ throughout this example. 
The objective function of $(D_2)$ is given by 
\[
\mathcal L(z) = \max(0, 1-\gamma \log 2-z_1)
+\max(0, 1-\gamma \log 2-z_2)+\max(0, 1-z_1, 1-z_2) + (z_1+z_2).
\]
In this case, the minimal solution $z^*\coloneqq (z^*_1, z^*_2)$ to $(D_2)$ is given by $z^*_1=z^*_2=1-\gamma \log 2=1+\gamma \log 0.5$.
The optimal value of $(D_2)$ is $2+\gamma \log 0.5$.
From Theorem~\ref{d_weighted}, the desired direction is $d^0 = (0.5, 0.5)$.
After the price perturbation along $d^0$, the minimal demand sets remain unchanged for all buyers.

\textbf{Violation of \LSC\, in Case~(iii).}
Until the price reaches $p = (1, 1)$, the minimal demand sets remain unchanged.
At $p = (1, 1)$, Buyer~2 becomes indifferent between good~1 and good~2 since 
$u_2((1,0);p)=u_2((0,1);p)= 7.5$ and their minimal demand set changes to $\Dc_2(p)=\{(1,0), (0,1)\}$. 
If this direction were maintained, the minimal demand set would change to $\Dc_2(p+\varepsilon d^0)=\{(1,0)\}$. 
This violates \LSC for the direction $d^0$ since good~2 is no longer over-demanded.
The objective function of $(D_2)$ is given by 
\[
\mathcal L(z) = \max(0, 1-\gamma \log 2-z_1)
+\max(0, 1-\gamma \log 0.5-z_1, 1-\gamma \log 2-z_2)+\max(0, 1-z_1, 1-z_2) + (z_1+z_2).
\]
The optimal value of $(D_2)$ increases to $2-\gamma\log 0.5=2+\gamma\log 2$ from $2+\gamma\log 0.5$. This shows the violation is due to Case~(iii). 
In this case, $z^*\coloneqq (z^*_1, z^*_2)$ changes to $z^*_1=1-\gamma \log 2=1+\gamma \log 0.5$ and $z^*_2=1-\gamma \log 8=1+\gamma \log 0.125$ and the desired direction is $d^1 = (0.5, 0.125)$.
After the price perturbation along $d^1$, the minimal demand sets become:
$\Dc_1(p+\varepsilon d^1) = \{(1,0)\}$,\ $\Dc_2(p+\varepsilon d^1)=\{(1,0), (0,1)\}$,\ and $\Dc_3(p+\varepsilon d^1) = \{(0,1)\}$, 
whereas $X^*=\{1, 2\}$ remains unchanged.

\textbf{Violation of \LSC\, in Case~(iv).}
The algorithm proceeds until $p = (1.5, 1.125)$ with minimal demand sets unchanged.
At this point, Buyer~1 becomes indifferent between good~1 and good~2 since 
$u_1((1,0);p)=u_1((0,1);p)= 5.2$ and their minimal demand set becomes $\Dc_1(p)=\{(1,0), (0,1)\}$. 
If this direction were maintained, the minimal demand set would become $\Dc_1(p+\varepsilon d^1)=\{(0,1)\}.$ This violates \LSC\, for the direction $d^1$ since good~1 is no longer over-demanded.
The objective function of $(D_2)$ is given by 
\begin{align*}
\mathcal L(z) &= \max(0, 1-\gamma \log 2-z_1, 1-\gamma \log 1.6-z_2)\\
&\quad+\max(0, 1-\gamma \log 0.5-z_1, 1-\gamma \log 2-z_2)+\max(0, 1-z_2) + (z_1+z_2).
\end{align*}

The optimal value remains unchanged at $2+\gamma\log 2$, whereas $z^*\coloneqq (z^*_1, z^*_2)$ changes to $z^*_1=1-\gamma \log 2=1+\gamma \log 0.5$ and $z^*_2=1-\gamma \log 1.6=1+\gamma \log 0.625$.
This shows the violation is due to Case~(iv).
Then, the desired direction is $d^{2}=(0.5,0.625)$. After the price perturbation along $d^{2}$, the minimal demand sets become: $\Dc_1(p+\varepsilon d^2) = \{(1,0), (0,1)\}$,\ $\Dc_2(p+\varepsilon d^2)=\{(1,0)\}$,\ and $\Dc_3(p+\varepsilon d^2) = \{(0,1)\}$, whereas $X^*=\{1, 2\}$ remains unchanged.
\end{example}

\end{document}